\documentclass[11pt]{report}
\usepackage{epsfig}
\usepackage{amsmath, amscd}
\usepackage{xypic}
\usepackage{psfrag}
\usepackage{comment}
\usepackage{amssymb}
\usepackage{amsthm}
\usepackage{graphicx}
\usepackage{color}
\usepackage[vcentermath]{youngtab}
\usepackage{young}
\newtheorem{prop}{Proposition}[chapter]
\newtheorem{question}{Question}[chapter]
\newtheorem{problem}{Problem}[chapter]
\newtheorem{thm}{Theorem}[chapter]
\newtheorem{lemma}{Lemma}[chapter]
\newtheorem{cor}{Corollary}[chapter]
\newtheorem{conj}{Conjecture}[chapter]
\newtheorem{defi}{Definition}[chapter]
\newtheorem{defn}{Definition}[chapter]
\newtheorem{ex}{Exercise}[chapter]

\newtheorem{hyp}{Hypothesis}[chapter]

\newtheorem{claim}{Claim}[chapter]
\newcommand{\PP}{\ensuremath{\mathbb{P}}}

\newcommand{\cc}[1]{\ensuremath{\mbox{\textbf{#1}}}}

\newcommand{\ignore}[1]{}
\newcommand{\bitlength}[1]{{\langle {#1} \rangle}}
\newcommand{\Z}{\ensuremath{\mathbb{Z}}}
\newcommand{\C}{\ensuremath{\mathbb{C}}}
\newcommand{\Q}{\ensuremath{\mathbb{Q}}}
\newcommand{\N}{\ensuremath{\mathbb{N}}}

\newcommand{\End}{\ensuremath{\mbox{End}}}
\newcommand{\Ker}{\ensuremath{\mbox{Ker}}}

\newcommand{\Hom}{\ensuremath{\mbox{Hom}}}
\newcommand{\Tr}{\ensuremath{\mbox{Tr}}}
\newcommand{\im}{\ensuremath{\mbox{Im}}}

\newcommand{\id}{\ensuremath{\mbox{id}}}
\newcommand{\Id}{\ensuremath{\mbox{Id}}}

\newcommand{\Implies}{\ensuremath{\Rightarrow} }

\def\af{{\mathrm{Aff}}}
\def\z2{{\Z_{<2>}}}
\def\ztn{{\Z_{<2>}^d}}
\def\beq{\begin{equation}}
\def\eeq{\end{equation}}
\def\beqs{\begin{eqnarray*}}
\def\eeqs{\end{eqnarray*}}
\newcommand{\R}{\mathbb{R}}

\newcommand{\x}{\mathbf{x}}

\newcommand{\n}{\mathbf{n}}

\def\a{{\mathbf{\alpha}}}
\def\b{{\mathbf{\beta}}}
\def\l{{\mathbf{\lambda}}}

\def\g{{\mathbf{\gamma}}}

\def\poly{\mbox{poly}}
\def\perm{\mbox{perm}}
\newcommand{\Sym}{\ensuremath{\mbox{Sym}}}

\def\ie{i.\,e.\,}
\begin{document}
\title{Geometric Complexity Theory: Introduction}
\author{
Dedicated to Sri Ramakrishna \\ \\
Ketan D. Mulmuley \footnote{Part of the work on GCT was done while 
the  first author was visiting I.I.T. Mumbai to which he is grateful for its 
hospitality}
 \\
The University of Chicago
 \\   \\ Milind Sohoni\\ I.I.T.,
Mumbai \\ \\
Technical Report TR-2007-16\\
Computer Science Department \\
The University of Chicago \\
September 2007}

\maketitle

\begin{center} {\bf \LARGE Foreword} \end{center}
\bigskip
\bigskip

These are lectures notes for the introductory graduate courses on 
geometric complexity theory (GCT) in the computer science department,
the university of Chicago. Part I consists of the lecture notes 
for the course given by the first author in the spring quarter, 2007. It
gives introduction to the  basic structure of GCT. 
Part II consists of the lecture notes for the
course given by the second author in the spring quarter, 2003. It gives
introduction to invariant theory with a view towards GCT. 
No background in algebraic geometry or representation theory is assumed. 
These lecture notes in conjunction with the article \cite{GCTflip1},
which describes in detail the basic plan of GCT based on 
the  principle  called the  {\em flip}, 
should provide a high level picture of GCT assuming familiarity with 
only basic notions of algebra, such as groups, rings, fields etc.
Many of the theorems in these lecture notes are stated without proofs,
but after giving enough motivation so that they can be taken on faith.
For the readers interested in further study, Figure~\ref{ftree} shows 
logical dependence among the various papers of GCT and a suggested 
reading sequence. 

The first author is grateful to Paolo Codenotti, Joshua Grochow, 
Sourav Chakraborty and Hari Narayanan for taking notes for his lectures.

\begin{figure}[!p]
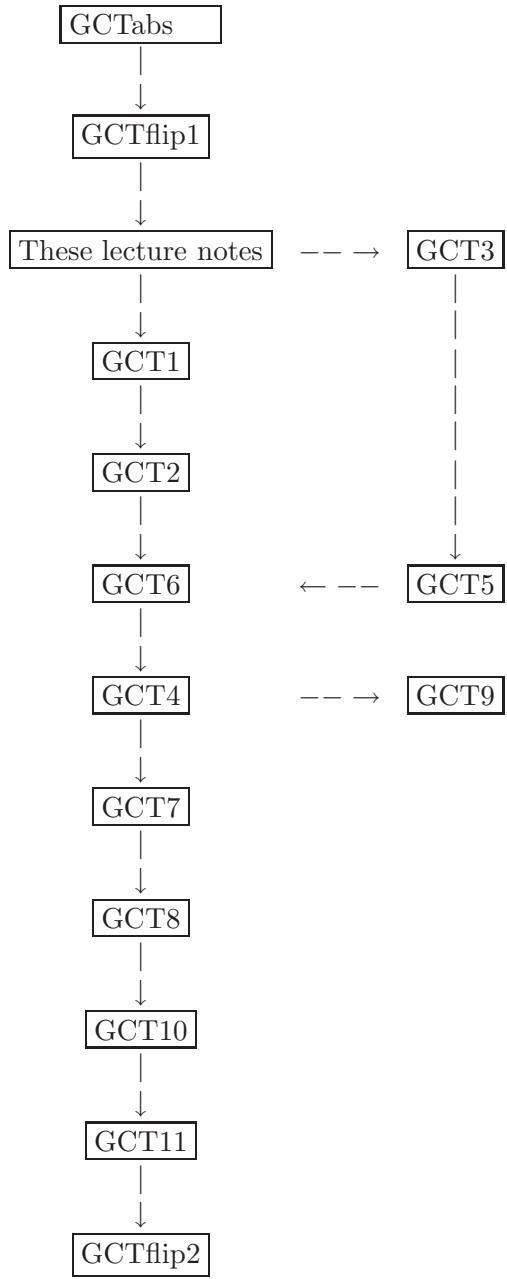

\centering
\[\begin{array} {ccc} 
\fbox{\parbox{.75in}{GCTabs}} \\
|\\
\downarrow\\
\fbox{GCTflip1}\\
|\\
\downarrow\\
\fbox{These lecture notes} & --\rightarrow & \fbox{GCT3} \\
| & & | \\
\downarrow & & | \\
\fbox{GCT1}& & | \\
|& & | \\
\downarrow & & | \\
\fbox{GCT2} & & | \\
|& & | \\
\downarrow & & \downarrow \\
\fbox{GCT6} & \leftarrow -- & \fbox{GCT5} \\
|\\
\downarrow\\
\fbox{GCT4} & --\rightarrow & \fbox{GCT9} \\
|\\
\downarrow\\
\fbox{GCT7}\\
|\\
\downarrow\\
\fbox{GCT8}\\
|\\
\downarrow\\
\fbox{GCT10}\\
|\\
\downarrow\\
\fbox{GCT11}\\
|\\
\downarrow\\
\fbox{GCTflip2}\\
\end{array}\]
\caption{Logical dependence among the GCT papers} 
\label{ftree}
\end{figure}
 
\tableofcontents

\vfil \eject

\quad\vfil \eject

\part{The basic structure of GCT \\ {\normalsize By Ketan D. Mulmuley}} 

\chapter{Overview}
\begin{center} {\Large  Scribe: Joshua A. Grochow} \end{center}

%\pagestyle{myheadings}
%\markboth{GCT Lecture 1}{GCT Lecture 1}

\noindent{\bf Goal:} An overview of GCT.

The purpose of this course is to give an introduction to Geometric Complexity Theory (GCT), which is an approach to proving $\cc{P}\neq \cc{NP}$ via
algebraic geometry and representation theory. 
A basic plan of this approach is described in \cite{GCTflip1,GCTflip2}. It
is partially implemented in a series of articles \cite{GCT1}-\cite{GCT11}. 
The  paper \cite{GCThyderabad} is a conference announcement of
GCT. The paper \cite{lowerbound} gives
an unconditional lower bound in a PRAM model without bit operations based on
elementary algebraic geometry, and was a starting point for 
the GCT investigation via algebraic geometry.

The only mathematical prerequisites for this course are a basic knowledge of abstract algebra (groups, ring, fields, etc.) and a knowledge of computational complexity.  In the first month
we plan to cover the representation theory of finite groups,
the symmetric group $S_n$, and  $GL_n(\C)$, and
enough algebraic geometry so that in the remaining lectures 
we  can cover  basic GCT.
Most of the background results will only be sketched or omitted.  

This  lecture uses slightly more algebraic geometry and representation theory than the reader is assumed to know in order to give a 
more complete picture of GCT. 
As the course continues, we will cover this material.

\section{Outline}
Here is an outline of the GCT approach.
Consider the  \cc{P} vs. \cc{NP} question in characteristic 0; i.e.,
over  integers.  So bit operations are not allowed, and basic operations on integers are considered to take constant time.  For a similar approach in nonzero characteristic (characteristic 2 being the classical case from a computational complexity point of view), see GCT 11.

The basic principle  of GCT is the called the {\em flip} \cite{GCTflip1}.
It  ``reduces'' (in essence, not formally)
 the lower bound problems such as 
 $\cc{P}$ vs. $\cc{NP}$ in characteristic 0 to
upper bound problems: showing that  certain decision problems  in algebraic geometry and representation theory belong to $P$.  Each of these  
decision problems is of
the form: is  a 
given (nonnegative) structural 
constant associated to some algebro-geometric or representation theoretic
object  nonzero?
This is akin to 
the  decision problem: given a matrix, is its permanent nonzero?
(We know how to solve this particular problem in polynomial time via 
reduction to the perfect matching problem.)

Next,  the preceding upper bound problems are reduced to
purely mathematical positivity hypotheses \cite{GCT6}.
The goal is to show that these and other auxilliary 
structural constants have positive formulae.
By a positive formula we mean a formula that does not involve any alternating
signs like the  usual positive formula for the permanent; in contrast 
the usual formula for the determinant involves alternating signs.

Finally, these positivity hypotheses are
``reduced'' to  conjectures in the theory of quantum groups  
\cite{GCT6,GCT7,GCT8,GCT10}
intimately related to the Riemann hypothesis over finite fields proved 
in \cite{weil2},
and
the related works \cite{beilinson,kazhdan1,lusztigbook}.
A pictorial summary of the GCT approach is shown in
Figure~\ref{fbasic}, where  the arrows 
represent reductions, rather than implications.

\begin{figure}
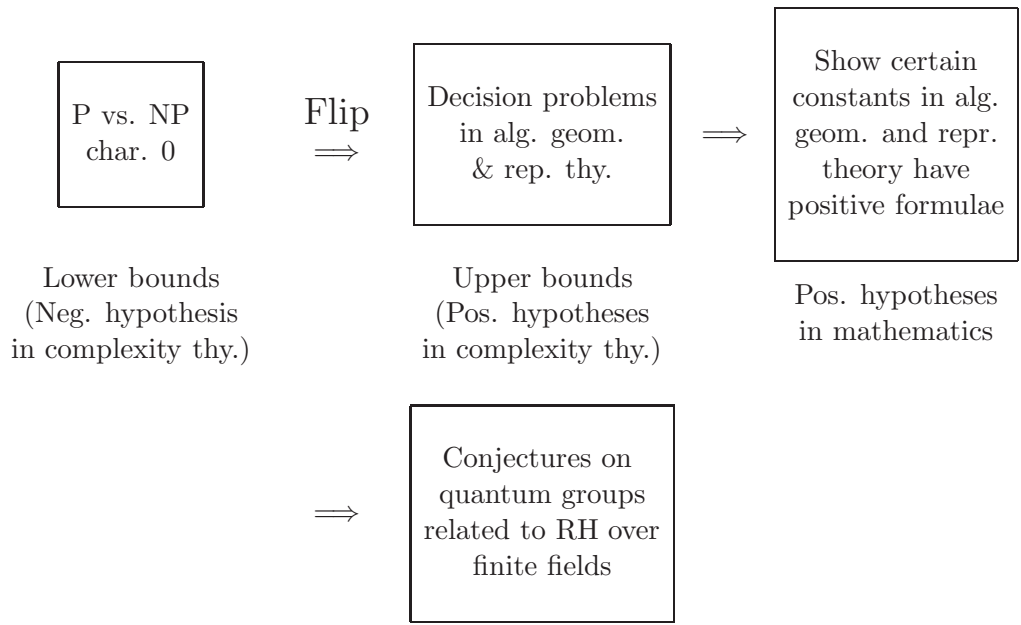

$$
\begin{array}{ccccc}
\begin{array}{|c|}
\hline
\\
\mbox{P vs. NP} \\
\mbox{char. 0} \\
\\ \hline
\end{array} & 
\begin{array}{c}
\mbox{{\Large Flip}} \\
\Longrightarrow 
\end{array}
& 
\begin{array}{|c|}
\hline \\
\mbox{Decision problems} \\
\mbox{in alg. geom.} \\
\mbox{\& rep. thy.} \\
\\ \hline
\end{array} & 
\Longrightarrow &
\begin{array}{|c|}
\hline \\
\mbox{Show certain} \\
\mbox{constants in alg.} \\
\mbox{geom. and repr.} \\ \mbox{theory have} \\
\mbox{positive formulae} \\
\\ \hline
\end{array} \\ 
\begin{array}{c}
\mbox{Lower bounds} \\
\mbox{(Neg. hypothesis} \\
\mbox{in complexity thy.)}
\end{array}
 & &
\begin{array}{c}
\mbox{Upper bounds} \\
\mbox{(Pos. hypotheses} \\
\mbox{in complexity thy.)}
\end{array}
 & &
 \begin{array}{c}
\mbox{Pos. hypotheses} \\
\mbox{in mathematics}
\end{array} \\ \\
&\Longrightarrow &
\begin{array}{|c|}
\hline \\
\mbox{Conjectures on } \\
\mbox{quantum groups} \\
\mbox{related to RH over} \\
\mbox{finite fields} \\
\\ \hline
\end{array} 
\end{array}
$$
\caption{The basic approach of GCT}
\label{fbasic}
\end{figure}

To recap:
we move from a negative hypothesis in complexity theory (that there {does
 not} exist a polynomial time algorithm for an $\cc{NP}$-complete problem)
to a positive hypotheses in complexity theory 
(that there  exist  polynomial-time algorithms
for certain decision problems) to positive hypotheses in mathematics (that certain structural constants have positive formulae) to  conjectures on
quantum groups related to the Riemann hypothesis over finite fields, the
related works  and their possible extensions.
The first reduction here is the  \emph{flip}:
we reduce a question about lower bounds, which are notoriously difficult, to
the  one about upper bounds, which we have a much better handle on.
This flip from negative to positive is already present in G\"{o}del's work:
to show something is impossible it suffices to show that something else is possible.  This was one of the motivations for the GCT approach.
The G\"{o}delian flip would not 
work for the  \cc{P} vs. \cc{NP} problem 
because it relativizes.  We can think of GCT as a form 
of nonrelativizable (and non-naturalizable, if reader knows what that means)
diagonalization.  

In summary, this approach very roughly ``reduces'' the lower bound problems
such as  $\cc{P}$ vs.  $\cc{NP}$ in characteristic zero 
to as-yet-unproved quantum-group-conjectures
related to the Riemann Hypothesis over finite fields.
As with the classical RH, there is  experimental evidence to suggest these conjectures hold -- which indirectly suggests that certain generalizations of the Riemann hypothesis over finite fields also hold -- and there are hints on
how the problem might be attacked. See \cite{GCTflip1,GCT6,GCT7,GCT8}
for  a more detailed exposition.

\section{The G\"{o}delian Flip}
We now re-visit G\"{o}del's original flip in modern language
to get the  flavor of the GCT flip.

G\"{o}del set out to answer the question:
$$
\mbox{Q: Is truth provable?}
$$
But  what ``truth'' and ``provable'' means here is not so obvious \emph{a priori}.  We start by setting the stage: in any mathematical theory, we have the syntax  (i.e. the language used) and the semantics (the domain of discussion).  In this case, we have:

\begin{table}[h]
\begin{tabular}{l|l}
Syntax (language) & Semantics (domain) \\ \hline
First order logic & \\
 ($\forall, \exists, \neg, \vee, \wedge, \dots$) & \\
Constants & 0,1 \\
Variables & $x,y,z,\dots$ \\
Basic Predicates & $>$, $<$, $=$ \\
Functions & $+$,$-$,$\times$,exponentiation \\
Axioms & Axioms of the natural numbers \N \\
 & Universe: $\N$
\end{tabular}
\end{table}

A sentence is a valid formula with all variables quantified, and by 
a \emph{truth} we mean  a sentence that is true in the domain.
By a proof we mean a valid deduction based on standard rules of inference and
the axioms of the domain, whose final result is the desired statement.  

Hilbert's program asked for  an algorithm that, given a sentence in
number theory, decides whether it is true or false. 
A special case of this is Hilbert's  10th problem, which asked for
an algorithm to decide whether a Diophantine equation (equation with only integer coefficients) has a nonzero integer solution. 
G\"{o}del showed that Hilbert's general program was not achievable.
The tenth problem remained unresolved until 1970, at which point Matiyasevich
showed its impossibility as well.

Here is the main idea of G\"{o}del's proof, re-cast in modern language.
For a Turing Machine $M$, 
whether the empty string $\varepsilon$ is in the language $L(M)$ recognized
by $M$ is undecidable.  The idea is to reduce a question
of the form $\varepsilon \in L(M)$ to a question in number theory. 
If there were an algorithm for deciding the truth of number-theoretic
statements, it would give an algorithm for the above Turing machine problem,
which we know does not exist.

The basic idea of the reduction is similar to the one in 
Cook's proof that \emph{SAT} is \cc{NP}-complete.  Namely, $\varepsilon \in L(M)$ iff there is a valid computation of $M$ which accepts $\varepsilon$.  Using Cook's idea, we can use this to get a Boolean formula: 
$$
\begin{array}{l} 
\exists m \exists \mbox{ a valid computation of } M \mbox{ with configurations of size } \leq m  \\ 
\qquad \mbox{ s.t. the computation accepts }  \varepsilon.
\end{array}
$$
Then we use G\"{o}del numbering -- which assigns a unique number to each sentence in number theory -- to translate this formula to a sentence in number theory.  The details of this should be familiar.

The key point here is: to show that truth is undecidable in number theory
(a negative statement), we show that {there exists} 
a computable reduction from $\varepsilon \stackrel{?}{\in} L(M)$ to number theory (a positive statement). 
This is the essence of the G\"{o}delian flip, which is analogous to -- and in fact was the original motivation for -- the GCT flip.

\section{More details of the GCT approach}
To begin with,  GCT  associates  to each complexity class such as 
\cc{P} and \cc{NP} a projective algebraic variety $\chi_{P}$, $\chi_{NP}$, etc.
 \cite{GCT1}.  
In fact, it associates a family of varieties $\chi_{NP}(n,m)$: 
one for each input length $n$ and circuit size $m$,
but for simplicity we suppress this here.  The languages $L$ in the associated complexity class will be points on these varieties, and the set of such points is dense in the variety.  These varieties are thus called \emph{class varieties}.  To show that $\cc{NP} \nsubseteq \cc{P}$ in characteristic zero,
it suffices to show that $\chi_{NP}$ cannot be imbedded in $\chi_P$.

These class varieties are in fact  $G$-varieties. That is, they 
have an action of the group $G=GL_n(\C)$ on them. 
This action  induces an action on the homogeneous coordinate ring of
the variety, given by $(\sigma f)(\mathbf{x}) = f(\sigma^{-1} \mathbf{x})$ for all $\sigma \in G$.  Thus the coordinate rings $R_P$ and $R_{NP}$ of $\chi_P$
and $\chi_{NP}$  are $G$-algebras, i.e., algebras with $G$-action.
Their degree $d$-components $R_P(d)$ and $R_{NP}(d)$
are thus finite dimensional $G$-representations.

For the sake of contradiction, suppose $\cc{NP} \subseteq \cc{P}$ in characteristic 0.  Then there must be an embedding of $\chi_{NP}$ into $\chi_P$ as a
$G$-subvariety, which in turn gives rise (by standard algebraic geometry arguments) to a surjection $R_P \twoheadrightarrow R_{NP}$ of the coordinate rings.  This implies
(by standard representation-theoretic arguments) that $R_{NP}(d)$ can be
embedded as  a 
$G$-sub-representation of $R_P(d)$.
The following diagram summarizes the implications.

$$
\xymatrix{
\mbox{\parbox{1in}{\centering complexity \\ classes}} & \mbox{\parbox{1in}{\centering class \\ varieties}} & \mbox{\parbox{1in}{\centering coordinate \\ rings}} & \mbox{\parbox{1in}{\centering representations \\ of $GL_n(\C)$}} \\
NP \ar@{^{(}->}[d] \ar@{~>}[r]& \chi_{NP} \ar@{^{(}->}[d] \ar@{~>}[r]& R_{NP} \ar@{~>}[r]& R_{NP}(d) \ar@{^{(}->}[d] \\
P \ar@{~>}[r]& \chi_P \ar@{~>}[r]& R_P \ar@{->>}[u] \ar@{~>}[r]& R_P(d) \\
}
$$

Weyl's theorem--that all finite-dimensional representations of $G=GL_n(\C)$
are   completely reducible, i.e. can be written  as a
direct sum of irreducible representations--implies that
both $R_{NP}(d)$ and $R_P(d)$
can be written as direct sums of irreducible $G$-representations.
An \emph{obstruction} \cite{GCT2} of degree $d$ is defined to be 
an irreducible $G$-representation occuring (as a subrepresentation) in
 $R_{NP}(d)$ but not in $R_P(d)$. Its existence implies that 
$R_{NP}(d)$ cannot be embedded as a subrepresentation of $R_P(d)$,
and hence, $\chi_{NP}$ cannot be embedded in $\chi_P$
as a $G$-subvariety; a contradiction.

We actually have  a \emph{family} of varieties $\chi_{NP}(n,m)$: 
one for each input length $n$ and circuit size $m$. 
Thus if an obstruction of some degree exists for all $n \rightarrow \infty$,
assuming $m=n^{\log n}$ (say), then  $\cc{NP} \neq \cc{P}$ in
characteristic zero.

\begin{conj} \cite{GCTflip1} There is a polynomial-time algorithm for 
constructing such obstructions. \end{conj}

This is the GCT flip: to show that 
{no polynomial-time algorithm exists} for an \cc{NP}-complete  problem,
we hope to show that 
{there is a polynomial time algorithm} for finding  obstructions.
This task then is further reduced to finding 
polynomial time algorithms
for other decision problems in algebraic geometry and representation theory.

Mere existence of  an obstruction for all $n$ would actually suffice here.
For this, it suffices to show that there is an algorithm which, given $n$, outputs an obstruction showing that $\chi_{NP}(n,m)$ cannot be imbedded
in $\chi_P(n,m)$, when $m=n^{\log n}$.
But the conjecture is not just that there is an algorithm,
but that there is a {polynomial-time} algorithm.

The  basic principle here
is that the complexity of the proof of existence of an object (in this case, an obstruction)
is very closed tied to the computational complexity of finding that  object,
and hence, techniques underneath an easy (i.e. polynomial time) time algorithm
for deciding existence may yield an easy (i.e. feasible) proof of existence.
This is supported by much anecdotal evidence:
\begin{itemize}
\item An obstruction to planar embedding (a forbidden Kurotowski minor) can be
found in polynomial, in fact, linear 
time by variants of the usual planarity testing algorithms, and the 
underlying techniques, in retrospect, yield an algorithmic proof of
Kurotowski's theorem that every nonplanar graph contains a forbidden
minor. 
\item Hall's marriage theorem, which 
characterizes the existence of perfect matchings, in retrospect, follows
from the  techniques underlying 
 polynomial-time algorithms for finding perfect matchings.
\item The proof that a graph is Eulerian iff all vertices have even degree is,
 essentially, a polynomial-time algorithm for finding an Eulerian circuit.
\item In contrast, we know of no Hall-type theorem
for  Hamiltonians paths,  essentially, because 
finding such a path is  computationally difficult (\cc{NP}-complete).
\end{itemize}

Analogously the goal is to find a polynomial time algorithm for deciding
if there exists an obstruction for given $n$ and $m$, and then use
the underlying techniques to show that an obstruction always exists for
every large enough $n$ if $m=n^{\log n}$. The main mathematical work in GCT
takes steps  towards this goal.

\chapter{Representation theory of reductive groups}
\begin{center}{\Large   Scribe: Paolo Codenotti} \end{center}

\noindent {\bf Goal:} Basic notions in representation theory.

\noindent {\em References:} \cite{FulH,YT}

In this lecture we  review the  basic representation theory of reductive
groups as needed in this course. Most of the proofs 
will be omitted, or just sketched. For  complete proofs, see
the books by Fulton and Harris, and Fulton \cite{FulH, YT}.
The underlying field throughout this course is $\C$.

\section{Basics of Representation Theory}

\subsection{Definitions}

\begin{defi}
A \emph{representation} of a group $G$, also called a $G$-module,
is a vector space $V$ with an associated homomorphism $\rho:G\rightarrow
GL(V)$. We will refer to a representation by $V$.
\end{defi}

The map $\rho$ induces a natural action of $G$ on $V$, defined by $g\cdot v = (\rho(g))(v)$.

\begin{defi}
A map $\varphi:V\rightarrow W$ is $G$-\emph{equivariant} if the following diagram commutes:
\[\begin{CD}
V @>\varphi>> W\\
@VVgV @VVgV\\
V @>\varphi>> W
\end{CD}\]
That is, if $\varphi(g\cdot v) = g\cdot \varphi(v)$.
A $G$-equivariant map is also called $G$-invariant or a $G$-homomorphism.
\end{defi}

\begin{defi}
A subspace 
$W\subseteq V$ is said to be a \emph{subrepresentation}, or a $G$-submodule
of a representation $V$ over a group $G$ if $W$ is
$G$-\emph{equivariant}, that is if $g\cdot w\in W$ for all $w\in W$.
\end{defi}

\begin{defi}
A representation $V$ of a group $G$ is said to be \emph{irreducible} if it has no proper non-zero
$G$-subrepresentations.
\end{defi}

\begin{defi}
A group $G$ is called \emph{reductive} if every finite dimensional representation $V$ of $G$ is a direct sum of
irreducible representation.
\end{defi}

Here are some examples of reductive groups:
\begin{itemize}
\item
finite groups;
\item
the $n$-dimensional torus $(\mathbb{C}*)^n$;
\item 
linear groups:
\begin{itemize}
\item the general linear group $GL_n(\mathbb{C})$, 
\item the special linear group $SL_n(\mathbb{C})$, 
\item the orthogonal group $O_n(\mathbb{C})$ (linear transformations that preserve a symmetric form), 
\item and the symplectic group $Sp_n(\mathbb{C})$ (linear transformations that preserve a skew symmetric form);
\end{itemize}
\item
Exceptional Lie Groups
\end{itemize}

Their reductivity is a nontrivial fact. It will be proved later in this 
lecture for   finite groups, 
and the general and special linear groups.
In some sense, the  list above is complete:
all reductive groups can be constructed by basic operations from 
the components which are either in this  list or are
related to them in a simple way.

\subsection{New representations from old}

Given representations $V$ and $W$ of a group $G$,  we can construct new representations in several ways, some of which are described below.

\begin{itemize}
\item
Tensor product: $V\otimes W$. $g\cdot (v\otimes w) = (g\cdot v) \otimes (g\cdot w)$.
\item
Direct sum: $V\oplus W$.
\item
Symmetric tensor representation: The subspace 
$Sym^n(V) \subset V\otimes \dots \otimes V$ spanned by elements of the form
\[\sum_\sigma (v_1\otimes \dots \otimes v_n)\cdot \sigma 
=\sum_\sigma v_{\sigma(1)} \otimes \cdots v_{\sigma(n)},\]
where $\sigma$ ranges over all permutations in the symmetric group $S_n$.
\item
Exterior  tensor representation: The subspace $\Lambda^n(V)
\subset V\otimes \dots \otimes V$ spanned by elements of the form
 \[\sum_\sigma sgn(\sigma)(v_1\otimes \dots \otimes v_n) \cdot \sigma
=\sum_\sigma sgn(\sigma) v_{\sigma(1)} \otimes \cdots v_{\sigma(n)}.\]
\item
Let $V$ and $W$ be representations, then $\Hom(V, W)$ is also a representation, where $g\cdot \varphi$ is
defined so that the following diagram commutes:
\[\begin{CD}
V @>\varphi>> W\\
@VVgV @VVgV\\
V @>g\cdot\varphi>> W
\end{CD}\]

More precisely, \[(g\cdot \varphi)(v) = g\cdot(\varphi(g^{-1}\cdot v)).\]
\item
In particular, $V^*:V\rightarrow\C$ is a representation, and is called the \emph{dual representation}.

\item
Let $G$ be a finite group. Let $S$ be a finite $G$-set (that is, a finite set with an associated action of $G$ on its elements). We construct a vector space over any field $K$ (we will be mostly concerned with the case $K=\C$), with a basis vector associated to each element in $S$. More specifically, consider the set $K[S]$ of formal sums $\sum_{s\in S} \alpha_s e_s$, where $\alpha_s\in K$, and $e_s$ is a vector associated with $S\in s$. Note that this set has a vector space structure over $K$, and there is a natural induced action of $G$ on $K[S]$, defined by:
\[g\cdot\sum_{s\in S} \alpha_s e_s = \sum_{s\in S}\alpha_s e_{g\cdot s}.\]
This action gives rise to a representation of $G$. 
\item
In particular, $G$ is a $G$-set under the action of left multiplication. The representation we obtain in the manner described above from this $G$-set is called the \emph{regular representation}.

\end{itemize}

\section{Reductivity of  finite groups}

\begin{prop}
Let $G$ be a finite group. If $W$ is a subrepresentation of a representation $V$, then there exists a
representation $W^\bot$ s.t. $V=W\oplus W^\bot$.
\end{prop}
\begin{proof}
Choose any Hermitian form $H_0$ of $V$, and construct a new Hermitian form $H$ defined as:
\[H(v, w) = \sum_{g\in G} H_o(g\cdot v, g\cdot w).\]
Averaging is a useful trick that is used very often in representation theory, because it ensures $G$-invariance. In fact, $H$ is $G$-invariant, that is,
\[H(v, w)=\sum_{g\in G} H_o(g\cdot v, g\cdot w)=H(h\cdot v, h\cdot w)\]
Let $W^\bot$ be the perpendicular complement to $W$ with respect to the Hermitian form $H$. Then $W^\bot$ is
also $G$-invariant, and therefore it is a $G$-submodule.
\end{proof}

\begin{cor}
Every representation of a finite group is a direct sum of irreducible representations.
\end{cor}

\begin{lemma}(Schur)
If $V$ and $W$ are irreducible representations over $\C$, and $\varphi: V\rightarrow W$ is a homomorphism (i.e.
a $G$-invariant map), then:
\begin{enumerate}
\item Either $\varphi$ is an isomorphism or $\varphi=0$.
\item If $V=W$, $\varphi=\lambda I$ for some $\lambda \in \C$.
\end{enumerate}
\end{lemma}

\begin{proof}

\begin{enumerate}
\item
Since $\Ker(\varphi)$, and $\im{\varphi}$ are $G$-submodules,
either $\im(\varphi)=V$ or
$\im(\varphi)=0$.
\item Let $\varphi: V\rightarrow V$. Since 
$\C$ algebraically closed, there exists an eigenvalue $\lambda$ of $\varphi$.
 Look at the map $\varphi - \lambda I:V\rightarrow V$. By ($1$),
$\varphi -\lambda I=0$ (it can't be an isomorphism because something maps to $0$). So $\varphi=\lambda I$.
\end{enumerate}
\end{proof}

\begin{cor}
Every representation is a \emph{unique} direct sum of irreducible representations. More precisely, given two
decompositions into irreducible representations,
\[V=\bigoplus V_i^{a_i}\\
  =\bigoplus W_j^{b_j},\]
there is a one to one correspondence between the $V_i$'s and $W_j$'s, and the multiplicities correspond.
\end{cor}
\begin{proof}
exercise (follows from Schur's lemma).
\end{proof}

\section{Compact Groups and $GL_n(\C)$ are reductive}
Now we prove reductivity of compact groups.

\subsection{Compact groups}

Examples of compact groups: 
\begin{itemize}
\item $U_n(\C)\subseteq GL_n(\C)$, the unitary groups (all rows are normal and orthogonal).
\item $SU_n(\C)\subseteq SL_n(\C)$, the special unitary group.
\end{itemize}

Given a compact group, a \emph{left-invariant Haar measure} is a measure that is invariant under the left action of the group.
In other words, multiplication by a
group element does not change the area of a small  region (i.e.,
the group action is an isometry, see figure \ref{fig:compact}).

\begin{figure}
    \begin{center}
      \includegraphics[scale=0.5]{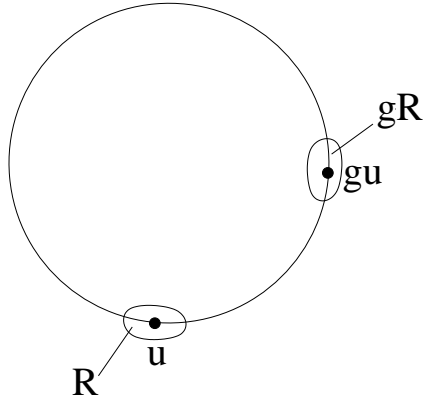}
      \caption{{\small Example of a left Haar measure for the circle
      ($U_1(\C)$). Left action by a group element $g$ on a small region
      $R$ around $u$ does not change the area.}}
      \label{fig:compact}
    \end{center}
\end{figure}

\begin{thm}
Compact groups are reductive
\end{thm}
\begin{proof}
We use the averaging trick again. In fact the proof is the same as in the case of finite groups, using integration instead of summation for the averaging trick. Let $H_0$ be any Hermitian form on V. Then define $H$ as:
\[H(v, w)=\int_G H(gv, gw) dG\]
where $dG$ is a left-invariant Haar measure. Note that $H$ is $G$-invariant. Let $W^\bot$ be the perpendicular complement to $W$. Then $W^\bot$ is 
$G$-invariant. Hence  it is a $G$-submodule. 
\end{proof}

The same proof as before then gives us Schur's lemma for compact groups, 
from which follows:

\begin{thm}
If $G$ is compact, then every finite dimensional representation
of $G$ is a unique direct sum of irreducible representations.
\end{thm}

\subsection{Weyl's unitary trick and $GL_n(\C)$}

\begin{thm} (Weyl)
$GL_n(\C)$ is reductive
\end{thm}
\begin{proof}(general idea)

Let $V$ be a representation of $GL_n(\C)$. Then $GL_n(\C)$ acts on $V$:
\[GL_n(\C)\hookrightarrow V.\]
But $U_n(\C)$ is a subgroup of $GL_n(\C)$. Therefore we have an induced action of $U_n(\C)$ on $V$, and we can look at $V$ as a representation of $U_n(\C)$. As a representation of $U_n(\C)$, $V$ breaks into irreducible representations of $U_n(\C)$ by the theorem above. To summarize, we have:
\[U_n(\C)\subseteq GL_n(\C)\hookrightarrow V = \oplus_i V_i, \]
where the $V_i$'s are irreducible representations of $U_n(\C)$.
Weyl's unitary trick uses Lie algebra to show that every finite dimensional representation of $U_n(\C)$ is also a representation of $GL_n(\C)$, and irreducible representations of $U_n(\C)$ correspond to irreducible representations of $GL_n(\C)$. Hence each $V_i$ above is an irreducible representation of $GL_n(\C)$.
\end{proof}

Once we know these groups are reductive, the goal is to construct and 
classify their 
 irreducible finite dimensional representations. This will be done
in the next lectures: 
Specht modules for $S_n$, and Weyl modules for $GL_n(\C)$.

\chapter{Representation theory of reductive groups (cont)}
\begin{center}{\Large   Scribe: Paolo Codenotti} \end{center} 

%\pagestyle{myheadings}
%\markboth{GCT Lecture 3}{GCT Lecture 3}

\noindent{\bf Goal:} Basic representation theory, continued from the last
lecture.

In this lecture we  continue our introduction to representation theory. Again we refer the reader to the
book by Fulton and Harris for full details \cite{FulH}. Let $G$ be a finite group, and $V$ a finite-dimensional
$G$-representation given by a homomorphism $\rho:G \to GL(V)$. We define the {\em character} of the
representation $V$ (denoted $\chi_V$) by $\chi_V(g) = Tr(\rho(g))$.

Since $Tr(A^{-1}BA) = Tr(B)$, $\chi_V(hgh^{-1}) =
\chi_V(g)$. This means  characters are constant on conjugacy
classes (sets of the form $\{hgh^{-1}| h\in G\}$, for any $g\in G$).
We call such functions {\em class functions}.

Our goal for this lecture is to prove the following two facts:
\begin{enumerate}
\item[Goal 1] A finite dimensional representation is completely determined by its character.
\item[Goal 2] The space of class functions is spanned by the characters of irreducible representations. In fact, these characters form
an orthonormal basis of this  space.
\end{enumerate}

First, we  prove some useful lemmas about characters.
\begin{lemma}
$\chi_{V \oplus W} = \chi_V + \chi_W$
\end{lemma}
\begin{proof}
Let $g\in G$, and let $\rho,\sigma$ be homomorphisms from $G$ into $V$ and $W$,
respectively.
Let $\lambda_1,\dots,\lambda_r$ be the eigenvalues of $\rho(g)$, and $\mu_1,\dots,\mu_s$  the eigenvalues of $\sigma(g)$.
Then $(\rho \oplus \sigma)(g) = (\rho(g),\sigma(g))$, so the eigenvalues of $(\rho \oplus \sigma)(g)$ are just the eigenvalues
of $\rho(g)$ together with the eigenvalues of $\sigma(g)$.

Then $\chi_V(g) = \sum_i \lambda_i$, $\chi_W(g) = \sum_i \mu_i$, and $\chi_{V \oplus W} = \sum_i \lambda_i + \sum_i \mu_i$.
\end{proof}

\begin{lemma}
$\chi_{V \otimes W} = \chi_V \chi_W$
\end{lemma}
\begin{proof}
Let $g\in G$, and let $\rho,\sigma$ be homomorphisms into $V$ and $W$,
respectively.
Let $\lambda_1,\dots,\lambda_r$ be the eigenvalues of $\rho(g)$, and $\mu_1,\dots,\mu_s$  the eigenvalues of $\sigma(g)$.
Then $(\rho \otimes \sigma)(g)$ is the Kronecker product of the matrices $\rho(g)$ and $\sigma(g)$. So its eigenvalues are
all $\lambda_i \mu_j$ where $1\le i \le r$, $1 \le j \le s$.

Then, $\Tr((\rho \otimes \sigma)(g)) = \sum_{i,j} \lambda_i \mu_j = \left( \sum_i \lambda_i \right) \left(\sum_j \mu_j\right)$,
which is equal to $\Tr(\rho(g)) \Tr(\sigma(g))$.
\end{proof}

\section{Projection Formula}
In this section, we derive a projection formula needed for Goal 1 
that allows us to determine the multiplicity of an irreducible
representation in another representation. Given a $G$-module $V$,
let  $V^G = \{v| \forall g\in G, g\cdot v = v \}$. We will call these
elements $G$-invariant. Let
\begin{equation}\label{equ:def}
\phi = \frac{1}{|G|} \sum_{g\in G} g \in \End(V),
\end{equation}
where each $g$, via $\rho$ is considered an element of $\End(V)$. 

\begin{lemma} 
The map  $\phi: V\rightarrow V$ is a
$G$-homomorphism;  i.e., $\phi \in Hom_G(V,V) = (Hom(V,V))^G$.
\end{lemma} 
\begin{proof} 
The set $\End(V)$ is a $G$-module, as we saw in last class,
via the following commutative diagram: for any
$\pi\in\End(V)$, and  $h\in G$:

\[\begin{CD}
V @>\pi>> V\\
@VVhV @VVhV\\
V @>h \cdot \pi>> V.
\end{CD}\]

Therefore $\pi\in Hom_G(V,V)$ (i.e., $\pi$ is a $G$-equivariant morphism) iff
 $h \cdot \pi = \pi$ for all $h\in G$.

When $\phi$ is defined as in equation (\ref{equ:def}) above,
\[h\cdot \phi = \frac{1}{|G|}\sum_g hgh^{-1} = \frac{1}{|G|}\sum_g g = \phi.\]
Thus \[h\cdot\phi = \phi,\ \forall h\in G,\] and $\phi:V\rightarrow V$ is a $G$-equivariant morphism, i.e.
$\phi\in Hom_G(V, V)$. \end{proof}

\begin{lemma}
The map $\phi$ is a $G$-equivariant projection of $V$ onto $V^G$
\end{lemma}
\begin{proof}
For every $w\in W$, let
\[v = \phi(w) = \frac{1}{|G|}\sum_{g\in G} g\cdot w.\]
Then \[h\cdot v = h\cdot \phi(w) = \frac{1}{|G|} \sum_{g\in G} hg \cdot w = v,\ \textrm{for any}\ h\in G.\] So
$v\in V^G$. That is, $\im(\phi) \subseteq V^G$. But if $v\in V^G$, then \[\phi(v) = \frac{1}{|G|} \sum_{g\in G}
g\cdot v = \frac{1}{|G|} |G| v = v.\] So $V^G \subseteq \im(\phi)$, and $\phi$ is the identity on $V^G$. This
means that $\phi$ is the projection onto $V^G$.
\end{proof}

\begin{lemma}\label{lemm:dim}
\[\dim(V^G)= \frac{1}{|G|}\sum_{g\in G} \chi_V(g).\]
\end{lemma}
\begin{proof}
We have: $dim(V^G) = \Tr(\phi)$, because $\phi$ is a projection ($\phi = \phi|_{V^G} \oplus \phi|_{Ker(\phi)}$).
Also, \[\Tr(\phi) = \frac{1}{|G|} \sum_{g\in G} \Tr_V(g) = \frac{1}{|G|} \sum_{g\in G} \chi_V(g).\]
\end{proof}

This gives us a formula for the multiplicity of the trivial representation (i.e., $dim(V^G)$) inside $V$.

\begin{lemma} \label{lschurcon}
Let $V,W$ be $G$-representations. If $V$ is irreducible, $dim(\Hom_G(V,W))$ is the multiplicity of $V$ inside
$W$. If $W$ is irreducible, $dim(\Hom_G(V,W))$ is the multiplicity of $W$ inside $V$.
\end{lemma}
\begin{proof}
By Schur's Lemma.
\end{proof}

Let $C_{class}(G)$ be the space of class functions on $(G)$, and let $(\alpha, \beta) = \frac{1}{|G|} \sum_g
\overline{\alpha}(g)\beta(g)$ be the Hermitian form on $C_{class}$

\begin{lemma}\label{lemm:onezero}
If $V$ and $W$ are irreducible $G$ representations, then
\begin{eqnarray}
(\chi_V, \chi_W)=\frac{1}{|G|} \sum_{g\in G} \overline{\chi_V}(g) \chi_W(g)=\begin{cases} 1 &
\textrm{if}\ V\cong W\\
0 & \textrm{if}\ V\ncong W.
\end{cases}
\end{eqnarray}
\end{lemma}
\begin{proof}
Since $\Hom(V,W) \cong V^* \otimes W$,
$\chi_{\Hom(V,W)} = \chi_{V^*} \chi_W = \overline{\chi_V} \chi_W$.
Now the result  follows from   Lemmas~\ref{lemm:dim} and \ref{lschurcon}.
\end{proof}

\begin{lemma}\label{lemm:orthon}
The characters of the irreducible representations form an orthonormal set.
\end{lemma}
\begin{proof}
Follows from Lemma \ref{lemm:onezero}.
\end{proof}

If $V$,$W$ are irreducible, then $\langle \chi_V, \chi_W \rangle$ is $0$ if $V\ne W$ and $1$ otherwise.

This implies that:

\begin{thm}[Goal 1]
 A representation is determined completely by its character.
\end{thm}
\begin{proof}
Let $V = \bigoplus_i V_i^{\oplus a_i}$.
So $\chi_V = \sum_i a_i \chi_{V_i}$, and 
 $a_i = (\chi_V, \chi_{V_i})$.
This gives us a formula for the multiplicity of an irreducible representation in
another representation, solely in terms of their characters.
Therefore, a representation is completely determined by its character.
\end{proof}

\section{The characters of  irreducible representations form a basis}
In this section, we address Goal 2.

Let $R$ be the regular representation of $G$, $V$ an irreducible representation of $G$.
\begin{lemma}
\[R = \bigoplus_V End(V,V),\] where $V$ ranges over all irreducible representations of $G$.
\end{lemma}
\begin{proof}
$\chi_R(g)$ is $0$ if $g$ is not the identity and $|G|$ otherwise.

\[(\chi_R,\chi_V) = \frac{1}{|G|} \sum_{g\in G} \overline{\chi_R}(g) \chi_V(g) = \frac{1}{|G|} |G| \chi_V(e) = \chi_V(e) =
dim(V)\]
\end{proof}

Let $\alpha: G \to \C$. For any $G$-module $V$, let $\phi_{\alpha,V} = \sum_g \alpha(g) g : V \to V$

\begin{ex}
$\phi_{\alpha,V}$ is $G$ equivariant (i.e. a $G$-homomorphism) iff $\alpha$ is a class function.
\end{ex}

\begin{prop}\label{prop:zero}
Suppose $\alpha : G \to \C$ is a class function, and
 $(\alpha, \chi_V)=0$ for all irreducible representations $V$.
Then $\alpha$ is identically $0$.
\end{prop}
\begin{proof}
If $V$ is irreducible, then, by Schur's lemma, since $\phi_{\alpha,V}$ is a
 $G$-homomorphism, and $V$ is
irreducible, $\phi_{\alpha,V} = \lambda  \Id$, where $\lambda = \frac{1}{n} \Tr(\phi_{\alpha,V})$, $n =
dim(V)$. We have:

\[\lambda = \frac{1}{n} \sum_g \alpha(g) \chi_V(g) = \frac{1}{n} |G| (\alpha, \chi_{V^*}).\]

Now $V$ is irreducible iff $V^*$ is irreducible. 
So $\lambda = \frac{1}{n} |G| 0 = 0$.
Therefore, $\phi_{\alpha, V} = 0$ for any irreducible representation, and hence for any representation.

Now let $V$ be the regular representation. Since $g$ as endomorphisms of $V$
 are
linearly independent, $\phi_{\alpha,V}=0$ implies that  $\alpha(g)=0$.
\end{proof}

\begin{thm}
Characters form an orthonormal basis for the space of class functions.
\end{thm}
\begin{proof} Follows from Proposition \ref{prop:zero}, and Lemma \ref{lemm:orthon}
\end{proof}

If $V = \bigoplus_i V_i^{\oplus a_i}$, and $\pi_i : V \to V_i^{\oplus a_i}$ is the projection operator. We have
a formula $\pi = \frac{1}{|G|} \sum_g g$ for the trivial representation. Analogously:

\begin{ex}
$\pi_i = \frac{dim V_i}{|G|} \sum_g \overline{\chi_{V_i}}(g) g$.
\end{ex}

\section{Extending to Infinite Compact Groups}
In this section, we extend the preceding results to
infinite compact groups. We must take some facts as given, since
these theorems are much more complicated than those for finite groups.

Consider compact $G$, specifically $U_n(\C)$, the unitary subgroup of $G_(\C)$.
$U_1(\C)$ is the circle group. Since $U_1(\C)$ is abelian, all its representations are one-dimensional.

Since the group $G$ is infinite, we can no longer sum over it. The idea is to replace the sum
$\frac{1}{|G|}\sum_g f(g)$ in the previous setting with $\int_G f(g)d\mu$, where $\mu$ is a left-invariant Haar
measure on $G$. In this fashion,
we can derive analogues of the preceding results  for compact groups.
 We need to normalize, so we
set $\int_G d\mu = 1$.

Let $\rho : G \to GL(V)$, where $V$ is a finite dimensional $G$-representation. Let $\chi_V(g) = \Tr(\rho(g))$.
Let $V = \bigoplus_i V_i^{a_i}$ be the complete decomposition of $V$ into irreducible representations.

We can again create a projection operator $\pi: V \to V^G$, by letting $\pi = \int_G \rho(g)d\mu$.

\begin{lemma} \label{lcompact1}
We have: \[dim(V^G) = \int_G\chi_V(g)d\mu.\]
\end{lemma}
\begin{proof}
This result is analogous to Lemma \ref{lemm:dim} for finite groups.
\end{proof}

For class functions $\alpha, \beta$, define an inner product \[(\alpha, \beta) = \int_G
\overline{\alpha}(g)\beta(g)d\mu.\]

Lemma~\ref{lcompact1}  applied to $\Hom_G(V,W)$ gives

\[(\chi_V, \chi_W)=\int_G \overline{\chi_V} \chi_W d\mu = dim(\Hom_G(V,W)).\]
\begin{lemma}
If $V,W$ are irreducible, $(\chi_V, \chi_W)=1$ if $V$ and $W$ are isomorphic, and $(\chi_V, \chi_W)=0$
otherwise.
\end{lemma}
\begin{proof}
This result is analogous to Lemma \ref{lemm:onezero} for finite groups.
\end{proof}

\begin{lemma}
The irreducible representations are orthonormal, 
just as in Lemma \ref{lemm:orthon} in the case of finite groups.
\end{lemma}

If $V$ is reducible, $V = \bigoplus_i V_i^{\oplus a_i}$, then
\[a_i = (\chi_V, \chi_{V_i}) = \int_G \overline{\chi_{V}}\chi_{V_i} d\mu.\]
Hence
\begin{thm}
A finite dimensional representation is completely determined by its character.
\end{thm}

This achieves Goal $1$ for  compact groups. Goal $2$ is much harder:

\begin{thm}[Peter-Weyl Theorem]
(1) The characters of the irreducible representations of $G$ 
span a dense subset of
the space of continuous class functions.

(2) The coordinate functions of all irreducible matrix representations of $G$ span a dense subset of all
continuous functions on $G$.

\end{thm}
By a coordinate function of a representation $\rho: G \rightarrow GL(V)$, we 
mean the function on $G$ 
corresponding to a fixed  entry of the matrix form of $\rho(g)$.
For $G=U_1(\C)$, (2) gives the Fourier series expansion on the circle.
Hence, the Peter-Weyl theorem constitutes a far reaching generalization 
of the harmonic analyis from the circle to general $U_n(\C)$.

\chapter{Representations of the symmetric group}
\begin{center} {\Large  Scribe: Sourav Chakraborty} \end{center}

\textbf{Goal: } To determine the irreducible 
representations of the Symmetric group
$S_n$ and their characters.

\

\textbf{Reference: } \cite{FulH,YT}

\

\subsection*{Recall} Let $G$ be a reductive group. Then
\begin{enumerate}
\item Every finite dimensional representation of $G$ is completely
reducible, that is, can be written as a direct sum of irreducible
representations.

\item Every irreducible representation is determined by its
character.
\end{enumerate}

Examples of reductive groups:
\begin{itemize} \item Continuous:  algebraic torus $(\mathbb{C}^*)^m$,
general linear group $GL_n(\C)$, 
special linear group $Sl_n(\mathbb{C})$, symplectic group
$Sp_n(\mathbb{C})$, orthogonal  group $O_n(\mathbb{C})$.

\item Finite: alternating group $A_n$, symmetric group $S_n$,
$Gl_n(\mathbb{F}_p)$, simple lie groups of finite type.
\end{itemize}

\section{Representations and characters of $S_n$}
The number of irreducible representations of $S_n$ is the same as the
the number of conjugacy classes in $S_n$ since the irreducible characters form 
a basis of the space of class functions.
 Each permutation can be
written uniquely as a product of disjoint cycles. The collection of
lengths of the cycles in a permutation is called the cycle type of
the permutation. So a cycle type of a permutation on $n$ elements is
a partition of $n$. And in $S_n$ each conjugacy class is
determined by the cycle type, which, in turn, 
is determined by the partition of
$n$. So the number of conjugacy class is same as the number of
partitions of $n$. Hence:
\begin{equation} \mbox{Number of irreducible representations of
$S_n$ = Number of partitions of $n$}
\end{equation}

Let $\lambda = \{ \lambda_1 \geq \lambda_2 \geq \dots \}$ be a
partition of $n$; i.e., the size $|\lambda|=\sum \lambda_i$ is $n$.
The Young diagram corresponding to $\lambda$ is a table shown in
Figure 1. It is like an inverted staircase. The top row has
$\lambda_1$ boxes, the second row has $\lambda_2$ boxes and so on.
There are exactly $n$ boxes.

\

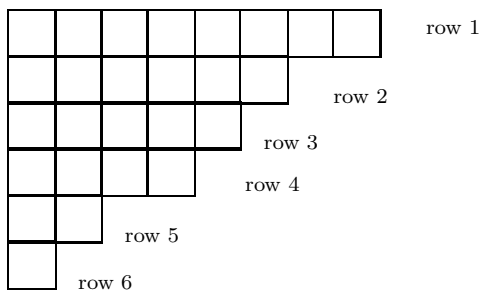
\begin{figure}[tbh]
\begin{center}
\setlength{\unitlength}{4144sp}%
\begingroup\makeatletter\ifx\SetFigFont\undefined%
\gdef\SetFigFont#1#2#3#4#5{%
  \reset@font\fontsize{#1}{#2pt}%
  \fontfamily{#3}\fontseries{#4}\fontshape{#5}%
  \selectfont}%
\fi\endgroup%
\begin{picture}(2524,1692)(439,-1291)
\thinlines
{\color[rgb]{0,0,0}\put(451,112){\framebox(2220,277){}}
}%
{\color[rgb]{0,0,0}\put(451,-166){\framebox(1665,278){}}
}%
{\color[rgb]{0,0,0}\put(451,-443){\framebox(1387,277){}}
}%
{\color[rgb]{0,0,0}\put(451,-998){\framebox(277,277){}}
}%
{\color[rgb]{0,0,0}\put(451,-721){\framebox(277,1110){}}
}%
{\color[rgb]{0,0,0}\put(1006,-721){\framebox(277,1110){}}
}%
{\color[rgb]{0,0,0}\put(1283,-721){\framebox(278,1110){}}
}%
{\color[rgb]{0,0,0}\put(1838,-166){\framebox(278,555){}}
}%
{\color[rgb]{0,0,0}\put(2393,112){\framebox(278,277){}}
}%
{\color[rgb]{0,0,0}\put(451,-1276){\framebox(277,278){}}
}%
{\color[rgb]{0,0,0}\put(728,-998){\framebox(278,277){}}
}%
{\color[rgb]{0,0,0}\put(451,-443){\framebox(1387,0){}}
}%
{\color[rgb]{0,0,0}\put(728,-721){\framebox(278,1110){}}
}%
\put(2948,250){\makebox(0,0)[lb]{\smash{{\SetFigFont{8}{9.6}{\rmdefault}{\mddefault}{\updefault}{\color[rgb]{0,0,0}row 1}%
}}}}
\put(1145,-998){\makebox(0,0)[lb]{\smash{{\SetFigFont{8}{9.6}{\rmdefault}{\mddefault}{\updefault}{\color[rgb]{0,0,0}row 5}%
}}}}
\put(867,-1276){\makebox(0,0)[lb]{\smash{{\SetFigFont{8}{9.6}{\rmdefault}{\mddefault}{\updefault}{\color[rgb]{0,0,0}row 6}%
}}}}
\put(2393,-166){\makebox(0,0)[lb]{\smash{{\SetFigFont{8}{9.6}{\rmdefault}{\mddefault}{\updefault}{\color[rgb]{0,0,0}row 2}%
}}}}
\put(1977,-443){\makebox(0,0)[lb]{\smash{{\SetFigFont{8}{9.6}{\rmdefault}{\mddefault}{\updefault}{\color[rgb]{0,0,0}row 3}%
}}}}
\put(1866,-693){\makebox(0,0)[lb]{\smash{{\SetFigFont{8}{9.6}{\rmdefault}{\mddefault}{\updefault}{\color[rgb]{0,0,0}row 4}%
}}}}
\end{picture}%
\end{center}
\caption{Row $i$ has $\lambda_i$ number of boxes}
\end{figure}

\

For a given partition $\lambda$, we want to
construct an irreducible representation $S_{\lambda}$, called the
Specht-module of $S_n$ for the partition  $\lambda$, and calculate the
character of $S_{\lambda}$. We shall give three constructions of
$S_{\lambda}$.

\subsection{First Construction}
A numbering $T$ of a Young diagram is a filling of the boxes in its
table with distinct numbers from $1, \dots, n$. A numbering of a
Young diagram is also called a tableau. It is called  a standard tableaux
if  the numbers are strictly increasing in each row and column. By
$T_{ij}$ we mean  the value in the tableaux at $i$-th row
and $j$-th column.
We  associate with each tableaux $T$ a polynomial  in
$\mathbb{C}[X_1, X_2, \dots, X_n]$:
$$f_T = \Pi_j \Pi_{i< i'} (X_{T_{ij}} - X_{T_{i'j}}).$$

Let $S_{\lambda}$ be the subspace of $\mathbb{C}[X_1, X_2, \dots,
X_n]$ spanned by $f_T$'s, where $T$ ranges  over all tableaux of shape
$\lambda$. It is a representation of $S_n$. 
Here $S_n$ acts on $\mathbb{C}[X_1, X_2, \dots, X_n]$ as:
$$(\sigma. f) (X_1, X_2, \dots, X_n)  = f(X_{\sigma(1)}, X_{\sigma(2)},
\dots, X_{\sigma(n)})$$

\begin{thm} \begin{enumerate} \item $S_{\lambda}$ is irreducible.
\item $S_{\lambda} \not\approx S_{\lambda '}$ if $\lambda \neq
\lambda '$
\item The set $\{f_T\}$, where $T$ ranges over standard tableau of shape
 $\lambda$,
is a basis of $S_{\lambda}$.
\end{enumerate}
\end{thm}

\subsection{Second Construction}

Let $T$ be a numbering of a Young diagram with distinct numbers from
$\{1, \dots, n\}$. An element $\sigma$ in $S_n$ acts on $T$ in
the usual way by permuting the numbers.
Let $R(T), C(T) \subset S_n$ be the sets of permutations that fix
the rows and columns of $T$, respectively.
We have $R(\sigma T) = \sigma R(T) \sigma^{-1}$ and $C(\sigma T) = \sigma
R(T)\sigma^{-1}$.
We say $T\equiv T'$ if the rows of $T$ and $T'$ are the same up to
ordering. The equivalence class of $T$, called the tabloid, is denoted by
 $\{ T\}$. Its orbit is isomorphic to $S_n/R(T)$.

Let  $\mathbb{C}[S_n]$ be the group algebra of $S_n$. 
Representations of $S_n$ are  the same as the representations of
 $\mathbb{C}[S_n]$.
The element $a_T = \sum_{p \in R(T)}p$ in $\mathbb{C}[S_n]$
 is called the row symmetrizer,  $b_T = \sum_{q \in
C(T)}sign(q)q$  the
column symmetrizer, and  $c_T = a_T b_T$  the Young symmetrizer.

Let

$$V_T = b_T.\{T\} = \sum_{q\in C(T)} sign(q) \{q T\}.$$

Then $\sigma . V_T = V_{\sigma T}$. Let $S_{\lambda}$ be the
span of all $V_T$'s, where $T$ ranges over all numberings of shape
$\lambda$.

\begin{thm} \begin{enumerate} \item $S_{\lambda}$ is irreducible

\item $S_{\lambda} \not\approx S_{\lambda '}$ if $\lambda \neq \lambda
'$.

\item The set $\{V_T | T \mbox{ standard}\}$ forms a basis for $S_{\lambda}$.
\end{enumerate}
\end{thm}

\

\subsection{Third Construction}

Let $T$ be a canonical numbering of shape $\lambda = (\lambda_1,
\dots, \lambda_k)$. By this, we mean
the first row is numbered by  $1, \dots, \lambda_1$, the second row
by $\lambda_1 + 1, \dots \lambda_1 + \lambda_2$, and so on,  and
the rows are increasing. Let $a_{\lambda} = a_T$, $b_{\lambda} =
b_T$, and $c_\lambda=c_T = a_T. b_T$.

Then $S_{\lambda} = \mathbb{C}[S_n]. c_{\lambda}$ is a
representation of $S_n$ from the left.

\begin{thm} \begin{enumerate} \item $S_{\lambda}$ is irreducible

\item $S_{\lambda} \not\cong S_{\lambda '}$ if $\lambda \neq \lambda
'$.

\item  The basis: an exercise.
\end{enumerate}
\end{thm}

\subsection{Character of $S_{\lambda}$ [Frobenius character
formula]}

Let $i = (i_1, i_2, \dots, i_k)$ be such that $\sum_j 
j i_{j} = n$. Let $C_i$ be the conjugacy class consisting of
permutations with $i_j$ cycles of length $j$.
Let $\chi_{\lambda}$ be the character of $S_{\lambda}$. The goal
is to find $\chi_{\lambda}(C_i)$.

Let $\lambda: \lambda_1\ge  \dots \ge \lambda_k$ be a partition of length $k$.
Given $k$ variables $X_1, X_2, \dots, X_k$, let
$$P_j(X) = \sum X_i^j,$$
be the power sum, and 
$$\Delta(X) = \Pi_{i<j}(X_j - X_i)$$
the discriminant.
Let $f(X)$ be a formal power series on $X_i$'s. Let
$[f(X)]_{\ell_1, \ell_2, \dots, \ell_k}$  denote the coefficient
of $X_1^{\ell_1}X_2^{\ell_2}\dots X_k^{\ell_k}$ in $f(X)$.
Let $\ell_i = \lambda_i + k - i$.

\begin{thm}[Frobenius Character Formula]

$$\chi_{\lambda}(C_i) = \left[ \Delta(X)\cdot \Pi_jP_j(X)^{i_j}\right]_{\ell_1, \ell_2, \dots,
\ell_k}$$
\end{thm}

\section{The first decision problem in GCT} 
Now we can state the first hard decision problem in representation theory
that arises in the context of the flip.
Let $S_{\alpha}$ and $S_{\beta}$ be two
Specht modules of $S_n$. Since $S_n$ is reductive,
$S_{\alpha}\otimes S_{\beta}$ decomposes as $$S_{\alpha}\otimes
S_{\beta} = \bigoplus k_{\alpha \beta}^{\lambda} S_{\lambda}$$ Here
$k_{\alpha \beta}^{\lambda}$ is called the Kronecker coefficient.

\begin{problem} {\bf (Kronecker problem)} 
Given $\lambda$, $\alpha$ and $\beta$ decide if
$k_{\alpha \beta}^{\lambda} > 0$.
\end{problem}

\begin{conj}[GCT6] This can be done in polynomial time; i.e. in
time   polynomial in the bit lengths of 
the inputs $\lambda$, $\alpha$ and $\beta$.
\end{conj}

\newcommand{\lecturenum5}{5}

\chapter{Representations of $GL_n(\C)$}
\begin{center} {\Large Scribe: Joshua A. Grochow} \end{center}

%\pagestyle{myheadings}
%\markboth{GCT Lecture \lecturenum5}{GCT Lecture \lecturenum5}

\noindent {\bf Goal:} To determine the irreducible representations of
$GL_n(\C)$ and their characters.

\noindent {\em References:} \cite{FulH,YT}

The goal of today's lecture is to classify all irreducible representations of $GL_n(\C)$ and compute their characters.  We will go over two approaches, the first due to Deruyts and the second due to Weyl. 

A \emph{polynomial representation} of $GL_n(\C)$ is a representation $\rho:GL_n(\C) \to GL(V)$ such that each entry in the matrix $\rho(g)$ is a polynomial in the entries of the matrix $g \in GL_n(\C)$.  

The main result is that the polynomial irreducible representations of $GL_n(\C)$ are in bijective correspondence with Young diagrams $\lambda$ of height at most $n$, i.e. $\lambda_1 \geq \lambda_2 \geq \cdots \geq \lambda_n \geq 0$.  Because of the importance of Weyl's construction (similar constructions can be used on many other Lie groups besides $GL_n(\C)$), the irreducible representation corresponding to $\lambda$ is known as the \emph{Weyl module} $V_\lambda$.

\section{First Approach [Deruyts]}
Let $X=(x_{ij})$ be a generic $n \times n$ matrix with variable entries $x_{ij}$.  Consider the polynomial ring $\C[X] = \C[x_{11},x_{12},\dots,x_{nn}]$.  Then $GL_n(\C)$ acts on $\C[X]$ by $(A \circ f)(X) = f(A^T X)$ (it is easily checked that this is in fact a left action).

Let $T$ be a tableau of shape $\lambda$.  To each column $C$ of $T$ of length $r$, we associate an $r \times r$ minor of $X$ as follows: if $C$ has the entries $i_1,\dots,i_r$, then take from the first $r$ columns of $X$ the rows $i_1,\dots,i_r$.  Visually:
$$
C = \left(\begin{array}{c}
i_1 \\
\vdots \\
i_r
\end{array}\right) 
\longrightarrow 
e_C = 
\begin{array}{ccccccl}
 & & & 1 & & \cdots & r\\
  & & & \downarrow & & & \downarrow\\
\begin{array}{c}
\\
i_1 \rightarrow \\
\\
i_2 \rightarrow \\
\vdots \\
i_r \rightarrow \\
\\
\end{array} & 
\multicolumn{6}{l}{\left(\begin{array}{ccccc}
\\
\mathbf{x_{i_1,1}} & \cdots & \mathbf{x_{i_1,r}} & \cdots & x_{i_1,n} \\
\\
\mathbf{x_{i_2,1}} & \cdots & \mathbf{x_{i_2,r}} & \cdots & x_{i_2,n} \\
\vdots & & \vdots & & \vdots  \\
\mathbf{x_{i_r,1}} & \cdots & \mathbf{x_{i_r,r}} & \cdots & x_{i_r,n} \\
\\ 
\end{array}\right)}

\end{array}
$$

(Thus if there is a repeated number in the column $C$, $e_C=0$, since the same row will get chosen twice.)  Using these monomials $e_C$ for each column $C$ of the tableau $T$, we associate a monomial to the entire tableau, $e_T = \prod_{C} e_C$.  (Thus, if in any column of $T$ there is a repeated number, $e_T=0$.  Furthermore, the numbers must all come from $\{1,\dots,n\}$ if they are to specify rows of an $n \times n$ matrix.  So we restrict our attention to numberings of $T$ from $\{1,\dots,n\}$ in which the numbers in any given column are all distinct.)

Let $V_\lambda$ be the vector space generated by the set $\{ e_T\}$,
where $T$ ranges over all such   numberings  of shape  $\lambda$.
  Then $GL_n(\C)$ acts on $V_\lambda$: for $g \in GL_n(\C)$, each row of $g X$ is a linear combination of the rows of $X$, and since $e_C$ is a minor of $X$, $g \cdot e_C$ is a linear combination of minors of $X$ of the same size, i.e. $g( e_C ) = \sum_D a^{g}_{C,D} e_D$ (this follows from standard linear algebra).  Then 
\begin{eqnarray*}
g(e_T) & = & g(e_{C_1} e_{C_2} \cdots e_{C_k}) \\
 & = & \left(\sum_D a^g_{C_1,D} e_D \right) \cdots \left( \sum_D a^g_{C_k,D} e_D\right)
\end{eqnarray*}
If we expand this product out, we find that each term is in fact $e_{T'}$ for some $T'$ of the appropriate shape.  We then have the following theorem:

\begin{thm} \label{tweyl}  \begin{enumerate}
\item $V_\lambda$ is an irreducible representation of $GL_n(\C)$.
\item The set $\{e_T | T \mbox{ is a semistandard tableau of shape } \lambda \}$ is a basis for $V_\lambda$.  (Recall that a semistandard tableau is one whose numbering is weakly increasing across each row and strictly increasing down each column.)
\item Every polynomial irreducible representation of $GL_n(\C)$ of degree $d$ is isomorphic to $V_\lambda$ for some partition $\lambda$ of $d$ of height at most $n$.
\item Every rational irreducible representation of $GL_n(\C)$ (each entry of $\rho(g)$ is a rational function in the entries of $g \in GL_n(\C)$) is isomorphic to $V_\lambda \otimes \det^k$ for some partition $\lambda$ of height at most $n$ and for some integer $k$ (where $\det$ is the determinant representation).
\item (Weyl's character formula)
Define the character $\chi_\lambda$ of $V_\lambda$ by 
$\chi_\lambda(g)=\mbox{Tr}(\rho(g))$, where 
$\rho: GL_n(\C) \rightarrow GL(V_\lambda)$
is the representation map. Then,
for $g \in GL_n(\C)$ with eigenvalues $x_1,\dots,x_n$, 
$$\chi_\lambda(g) = 
S_\lambda(x_1,\dots,x_n) := \frac{\left|x_{j}^{\lambda_i + n-i}\right|}{\left| x_{j}^{n-i}\right|}$$ (where $|y_{j}^i|$ is the determinant of the $n \times n$ matrix whose entries are $y_{ij} = y_{j}^i$, so, e.g., the determinant in the denominator is the usual van der Monde determinant, which is equal to $\prod_{i < j} (x_i - x_j)$). 
Here $S_\lambda$ is a polynomial, called the {\em Schur  polynomial}.
\end{enumerate} \end{thm}

NB: It turns out that all holomorphic representations of $GL_n(\C)$ are rational, and, by part (4) of the theorem, the Weyl modules classify all such representations up to scalar multiplication by powers of the determinant.

We'll give here a very brief introduction to the Schur polynomial
 introduced in the above theorem, and explain why
the Schur polynomial $S_\lambda$ 
associated to $\lambda$ gives the character of $V_\lambda$.

Let $\lambda$ be a partition, and  $T$ a semistandard tableau of shape $\lambda$.  Define $x(T) = \prod_i x_i^{\mu_i(T)} \in \C[x_1,\dots,x_n]$,
 where $\mu_i(T)$ is the number of times $i$ appears in $T$.  Then
it can be shown \cite{YT} that 
$$
S_\lambda(x_1,\dots,x_n) = \sum_{T} x(T),
$$
where the sum is taken over all semistandard tableau of shape $\lambda$.

\begin{prop} $S_\lambda(x_1,\dots,x_n)$ is the character of $V_\lambda$, where $x_1,\dots,x_n$ denote the eigenvalues of an element of $GL_n(\C)$. \end{prop}

\begin{proof} It suffices to show this
diagonalizable $g \in GL_n(\C)$, since the diagonalizable matrices are 
dense in $GL_n(\C)$.

So let $g \in GL_n(\C)$ be diagonalizable with 
eigenvalues $x_1,\dots,x_n$. We can assume that $g$ is diagonal.
If not, let $A$ be a matrix that diagonalizes $g$. So
$AgA^{-1}$ is diagonal with $x_1,\dots,x_n$ as its diagonal entries.
If $\rho: GL_n(\C) \to GL(V_\lambda)$ is the representation corresponding to the module $V_\lambda$, then conjugate $\rho$ by $A$ to get $\rho': GL_n(\C) \to GL(V_\lambda)$ defined by $\rho'(h)=A\rho(h)A^{-1}$.  In particular, since trace is invariant under conjugation, $\rho$ and $\rho'$ have the same character.  The module corresponding to $\rho'$ is simply $A \cdot V_\lambda$, which is clearly isomorphic to $V_\lambda$ since $A$ is invertible.  Thus to compute the character $\chi_\lambda(g)$, it suffices to compute the character of $g$ under $\rho'$, i.e., when  $g$ is diagonal, as we shall assume now.

We will show that $e_T$ is an eigenvector of $g$ with eigenvalue $x(T)$, i.e. $g(e_T)=x(T)e_T$.  Then since  $\{e_T | $T$ \mbox{ is a semistandard tableau of shape } \lambda\}$ is a basis for $V_\lambda$, the trace of $g$ on $V_\lambda$ 
will  just be $\sum_T x(T)$, where the sum is over semistandard $T$ of shape $\lambda$;  this is exactly $S_\lambda(x_1,\dots,x_n)$.

We reduce to the case where $T$ is a single column.  Suppose the claim is true for all columns $C$.  Then since $e_T$ is a product of $e_C$ where $C$ is a column, the corresponding eigenvalue of $e_T$ will be $\prod_C x(C)$ (where the product is taken over the columns $C$ of $T$), which is exactly $x(T)$. 

So assume $T$ is a single column, say with entries $i_1,\dots,i_r$. 
Then $e_T$ is simply the above-mentioned $r \times r$ minor of the generic $n \times n$ matrix $X=(x_{ij})$ (do not confuse the double-indexed entries of the matrix $X$ with the single-indexed eigenvalues of $g$).  Since $g$ is diagonal,
 $g^t=g$. So $(g \circ e_T)(X) = e_T(g^t X) = e_T(gX)$.  Thus $g$ multiplies the $i_j$-th column by $x_{i_j}$. Thus  its effect on $e_T$ is simply to multiply it by $\prod_{j=1}^r x_{i_j}$, which is exactly $x(T)$. \end{proof}

\subsection{Highest weight vectors}
The subgroup $B \subset GL_n(\C)$ of lower triangular invertible matrices,
called the {\em Borel subgroup}, is solvable. So every irreducible representation of $B$ is one-dimensional.  A \emph{weight vector} for $GL_n(\C)$ is a vector $v$ which is an eigenvector
for every matrix $b \in B$.
In other words, there is a function $\lambda: B \to \C$ such that 
$b \cdot v = \lambda(b) v$ for all $b \in B$.
The restriction  of $\lambda$  to the subgroup of diagonal matrices in
$B$ is known as the \emph{weight} of $v$.

As we showed in the proof of the above proposition, 
$$
\left(\begin{array}{ccc}
x_1 & & \\
 & \ddots & \\
 & & x_n 
\end{array}\right)
e_T = x(T)e_T = x_1^{\lambda_1} \cdots x_n^{\lambda_n} e_T.$$
So $e_T$ is a weight vector with weight $x(T)$.  
Thus  Theorem~\ref{tweyl} (2) gives
a basis consisting entirely of weight vectors.
We abbreviate the weight $x(T)$ by the sequence of exponents $(\lambda_1,\dots,\lambda_n)$.  We say $e_T$ is a \emph{highest weight vector} if its weight is the highest in the lexicographic ordering (using the above sequence notation for the weight).  

Each $V_\lambda$ has a unique (up to scalars) $B$-invariant vector, 
which turns out to be the highest weight vector:
 namely $e_T$, where $T$ is canonical.  For example, for $\lambda=(5,3,2,2,1)$, the canonical $T$ is:
$$
T = \young(11111,222,33,44,5)
$$
Note that the weight of such $e_T$ is $(\lambda_1,\dots,\lambda_n)$, so that the highest weight vector uniquely determines $\lambda$, and thus the entire representation $V_\lambda$.  (This is a general feature of highest weight vectors in the representation theory of Lie algebras and Lie groups.)  Thus the irreducible representations $GL_n(\C)$ are in bijective correspondence wih the highest weights of $GL_n(\C)$, i.e. the sequences of exponents of the eigenvalues of the $B$-invariant eigenvectors.

\section{Second Approach [Weyl]}
Let $V=\C^{n}$ and consider the $d$-th tensor power $V^{\otimes d}$.  The group $GL_n(\C)$ acts on $V^{\otimes d}$ on the left by the diagonal action
$$
g(v_1 \otimes \cdots \otimes v_d) = gv_1 \otimes \cdots \otimes gv_d \mbox{ ($g \in GL(V)$) }
$$
while the symmetric group $S_d$ acts on the right by 
$$
(v_1 \otimes \cdots \otimes v_d) \tau = v_{1\tau} \otimes \cdots \otimes v_{d \tau}\mbox{ ($\tau \in S_d$) }.
$$
These two actions commute, so $V^{\otimes d}$ is a representation of $H=GL_n(\C) \times S_d$.  Every irreducible representation of $H$ is of the form $U \otimes W$ for some irreducible representation $U$ of $GL_n(\C)$ and some irreducible representation $W$ of $S_d$.  Since both $GL_n(\C)$ and $S_d$ are reductive (every finite-dimensional representation is a direct sum of irreducible representations), their product $H$ is reductive as well.  So there are some partitions $\alpha$ and $\beta$ and integers $m_{\alpha\beta}$ such that $$V^{\otimes d} = \bigoplus (V_\alpha \otimes S_\beta)^{m_{\alpha \beta}},$$ where $V_\alpha$ are Weyl modules and $S_\beta$ are Specht modules (irreducible representations of the symmetric group $S_d$).

\begin{thm} $V^{\otimes d} = \bigoplus_{\lambda} V_\lambda \otimes S_\lambda$, where the sum is taken over partitions $\lambda$ of $d$ of height at most $n$.  (Note that each summand appears with multiplicity one, so this is a ``multiplicity-free'' decomposition.)  \end{thm}

Now, let $T$ be any standard tableau of shape $\lambda$, and recall the Young symmetrizer $c_T$ from our discussion of the irreducible representations of $S_d$.  Then $V^{\otimes d}c_T$ is a representation of $GL_n(\C)$ from the left (since $c_T \in \C[S_d]$ acts on the right, and the left action of $GL_n(\C)$ and the right action of $S_d$ commute.)

\begin{thm} \label{tembedweyl} $V^{\otimes d}c_T \cong V_\lambda$ \end{thm}

Thus
$$
V^{\otimes d} = \bigoplus_{\lambda : |\lambda|=d} \bigoplus_{\mbox{\parbox{1in}{\centering std. tableau \\ $T$ of shape $\lambda$}}} V^{\otimes d} c_T,
$$
where $|\lambda|=\sum \lambda_i$ denotes the size of $\lambda$.
In particular, $V_\lambda$ occurs in $V^{\otimes d}$ with multiplicity $\dim(S_\lambda)$.  

Finally, we construct a basis for $V^{\otimes d} c_T$.  A \emph{bitableau} of shape $\lambda$ is a pair $(U,T)$ where $U$ is a semistandard tableau of shape $\lambda$ and $T$ is a standard tableau of shape $\lambda$.  (Recall that the the semistandard tableau of shape $\lambda$ are in natural bijective correspondence with a basis for the Weyl module $V_\lambda$, while the standard tableau of shape $\lambda$ are in natural bijective correspondence with a basis for the Specht module $S_\lambda$.)

To each bitableau we associate a vector $e_{(U,T)}=e_{i_1} \otimes \cdots \otimes e_{i_d}$ where $i_j$ is defined as follows.  Each number $1,\dots,n$ appears in $T$ exactly once.  The number $i_j$ is the entry of $U$ in the same location as the number $j$ in $T$; pictorially:
\newcommand{\ij}{\ensuremath{i_j}}
$$
\begin{array}{cc}
 U & T \\
\begin{Young}
 & & & \cr
 & $i_j$ & \cr
 \cr
 \cr
\end{Young}
 & 
\begin{Young}
 & & & \cr
 & $j$ & \cr
 \cr
 \cr
\end{Young}
\end{array}
$$
  Then:

\begin{thm} The set $\{e_{(U,T)} c_T\}$ is a basis for $V^{\otimes d} c_T$.  \end{thm}

%\begin{thebibliography}{10}

%\bibitem{fulton} Fulton, W.  \emph{Young tableaux.}  London Mathematical Society Student Texts, 35.  Cambridge University Press: Cambridge, 1997.

%\bibitem{FulH} Fulton, W. and Harris, J.  \emph{Representation theory: a first course.}  Graduate Texts in Mathematics, 129.   Springer-Verlag: New York, 1991.

%\end{thebibliography}

%\setlength{\headsep}{.25in} \addtolength{\textheight}{0.25in}

\chapter{Deciding nonvanishing of Littlewood-Richardson coefficients}
\begin{center} {\Large  Scribe: Hariharan Narayanan} \end{center}

{\bf Goal:} \, To show that nonvanishing of 
Littlewood-Richardson coefficients can be decided in polynomial time.
\\\\

\noindent {\em References:} \cite{DM2,GCT3,honey}

\section{Littlewood-Richardson coefficients}
First we  define Littlewood-Richardson coefficients,
which  are basic quantities
encountered in representation theory. 
Recall that the irreducible representations $V_\l$ of $GL_n(\C)$, the Weyl
modules, 
 are indexed by partitions $\l$, and:

\begin{thm}[Weyl]
Every finite dimensional representation of $GL_n(\C)$ is completely
reducible.
\end{thm}

Let $G = GL_n(\C)$. Consider the diagonal embedding of $G
\hookrightarrow G \times G$. This is a group homomorphism. Any $G
\times G$ module, in particular, $V_\a \otimes V_\b$ can also be
viewed as a $G$ module via this homomorphism.
It then splits into irreducible
$G$-submodules:

\beq\label{split} V_\a \otimes V_\b = \oplus_\g c_{\a\b}^{\g}
V_\l.\eeq  Here $c_{\a\b}^{\g}$ is the multiplicity of $V_\g$ in
$V_\a \otimes V_\b$ and is known as the {\it Littlewood-Richardson}
coefficient.

The character of $V_\l$ is the Schur polynomial $S_\l$. Hence,
it follows from ~(\ref{split}) that the Schur polynomials satisfy
the following relation: \beq \label{eqschur1} 
S_\a S_\b = \oplus_\g c_{\a\b}^{\g}
S_\l.\eeq

\begin{thm}
 $c_{\a\b}^{\g}$ is in PSPACE.
\end{thm}
{\em Proof:} This easily follows from
eq.(\ref{eqschur1}) and the definition of Schur polynomials. \qed

As a matter of fact, a stronger result holds:
\begin{thm}
$c_{\a\b}^{\g}$ is in \#P.
\end{thm}
Recall that 
 $\#P \subseteqq$ PSPACE.

\noindent {\bf Proof:}  This is an immediate consequence of the 
following Littlewood-Richardson
rule (formula) for  $c_{\a\b}^{\g}$.
To state it, we need a few definitions.

Given partitions $\gamma$ and $\alpha$, a skew Young diagram of 
shape $\gamma / \alpha$ is 
the difference between the Young diagrams for $\gamma$ and
$\alpha$, with their top-left corners aligned; cf. Figure~\ref{fig2}.
 A skew tableau of shape
$\gamma / \alpha$ is a numbering of the boxes in this diagram.
It is  called semi-standard (SST) if the
entries in each column are strictly increasing top to bottom and the
entries in each row are weakly increasing left to right; see
Figures \ref{fig1} and \ref{fig2}.
The row word row($T$) of  a skew-tableau $T$ is the sequence of
numbers obtained by reading  $T$  left to right, bottom to
top; e.g. row($T$) for Figure~\ref{fig2} is $13312211$.
It is called a reverse lattice  word, if when read right to left,
for each $i$, the number of $i$'s encountered at any point is at least 
the number of $i+1$'s encountered 
till that point; thus the row word for Figure~\ref{fig2} is a reverse
lattice word. We say that $T$ is an LR
tableau for given $\a, \b, \g$  of shape $\g / \a$ and content $\b$ if
\begin{enumerate}
\item $T$ is an  SST,
\item row($T$) is a reverse lattice word,
\item $T$ has shape $\g/\a$, and 
\item the content of $T$ is $\b$, \ie the number of $i$'s in $T$
is $\b_i$. 
\end{enumerate}

For example,  Figure~\ref{fig2} shows an LR tableau with 
$\a = (6, 3, 2)$, $\b =
(4, 2, 2)$ and $\g = (8, 6, 3, 2)$.

\vspace{.1in}

{\bf The Littlewood-Richardson rule} \cite{YT,FulH}:
$c_{\a\b}^{\g}$ is equal to  the
number of LR skew tableaux of shape $\g/\a$ and content $\b$.

\begin{figure}[tbp1]
    \[
        \young(1125,223,4)
    \]
     \caption{semi-standard Young tableau}
    \label{fig1}
\end{figure}
\begin{figure}[tbp2]
    \[
        \young(::::::11,:::122,::3,13)
    \]
     \caption{
     Littlewood--Richardson skew tableau}
    \label{fig2}
\end{figure}

\vspace{.1in}

\noindent {\em Remark:} It may be noticed that the Littlewood-Richardson
rule depends only on the partitions 
$\alpha,\beta$ and $\gamma$ and not on $n$, the rank
of $GL_n(\C)$ (as long as it is greater than or equal to the maximum
height of $\alpha,\beta$ or $\gamma$). 
For this reason, we can assume without loss of generalitity that 
$n$ is the maximum of the heights of $\alpha,\beta$  and $\gamma$,
as we shall henceforth.

Now we  express $c_{\a \b}^\g$ as the number of integer points
in some polytope $P_{\a\b}^\g$ using the Littlewood-Richardson rule:

\begin{lemma} \label{lemmaP}
There exists  a polytope $P = P_{\a\b}^\g$ of dimension polynomial in $n$
 such that the number of
integer points in it is $c_{\a\b}^\g$.
\end{lemma}

\noindent {\em Proof:} Let $r_{j}^i(T)$, $i\le n$, $j\le n$,
denote the number of $j$'s in the $i$-th row of $T$. 
If $T$ is an LR-tableau of shape $\g /\a$ with content $\b$ then these 
integers satisfy the following constraints:

\begin{enumerate}
\item Nonnegativity: $r^i_j \ge 0$.
\item Shape constraints: For $i \le n$,
\[  \alpha_i + \sum_j r^i_j = \gamma_i. \]
\item Content constraints: For $j\le n$:
\[ \sum_i r^i_j=\beta_j.\]
\item Tableau constraints:
\[ \alpha_{i+1}+\sum_{k \le j} r^{i+1}_k \le \alpha_i + \sum_{k' <  j} r_{k'}^i.\]
\item Reverse lattice word constraints:
$r^i_j =0$ for $i<j$, and for $i\le n$, $1<j\le n$:
\[ \sum_{i'\le i} r^{i'}_j \le \sum_{i' < i} r^{i'}_{j-1}.\]
\end{enumerate}

Let $P_{\a\b}^\g$ be the polytope defined by these constraints. Then
$c_{\a\b}^{\g}$ is the number of integer points in this polytope.
This proves the lemma.  $\hfill \Box$

The membership function of the polytope $P_{\a\b}^\g$ is clearly computable in 
time that is polynomial in the bitlengths of $\a,\b$ and $\g$. 
Hence $c_{\a\b}^{\g}$ belongs to $\#P$. This proves the theorem.
 $\hfill \Box$

The complexity-theoretic content of the Littlewood-Richardson rule
is that it puts a quantity, which is a priori only in PSPACE, in
$\#P$. 
We also have:
\begin{thm}[\cite{hari}]
$c_{\a\b}^{\g}$ is \#P-complete.
\end{thm}

Finally, the main complexity-theoretic result that we are interested in:

\begin{thm}[GCT3, Knutson-Tao, De Loera-McAllister]
The problem of deciding nonvanishing of $c_{\a\b}^{\g}$ is in $P$, \ie, it 
 can be solved  in time that is polynomial  in the bitlengths of $\a, \b$ 
and $\g$. In fact, it can solved in strongly polynomial time \cite{GCT3}.
\end{thm}

Here, by a strongly polynomial time algorithm, we mean
that the number of arithmetic
steps $+, -, *, \leq, \dots$ in the algorithm 
is polynomial in the number of parts
of $\a, \b$ and $\g$ regardless of their bitlengths,
and the bit-length of each intermediate operand is
polynomial in the bitlengths of $\a, \b$ and $\g$.

\noindent {\bf Proof:} 
Let $P= P_{\a\b}^\g$ be the polytope as in Lemma~\ref{lemmaP}.
All vertices of $P$ have rational coefficients. Hence, for some
positive integer $q$, the scaled polytope $qP$ has an integer point.
It follows that, for this $q$, $c_{q\alpha,q\beta}^{q\gamma}$ is
positive. The saturation Theorem \cite{knutson} says that, in this
case, $c_{\alpha,\beta}^{\gamma}$ is positive. Hence, $P$ contains
an integer point. This implies:

\begin{lemma}
If $P \neq \emptyset$ then $c_{\a\b}^\g > 0$.
\end{lemma}

By this lemma,
to decide if $c_{\a\b}^\g > 0$, it suffices to test if $P$ is nonempty.
The polytope $P$ is given by $Ax \leq b$ where the entries of $A$ are
$0$ or $1$--such linear programs are called combinatorial.
Hence, this can be done in strongly polynomial time using Tardos' algorithm
\cite{lovasz} 
for combinatorial linear programming. 
This proves the theorem.
 $\hfill \Box$

The integer programming problem 
is NP-complete, in general. However,  linear programming works for
the specific 
integer programming problem here because of the saturation property
\cite{knutson}.

\noindent {\bf Problem}:
Find a genuinely combinatorial poly-time
algorithm for deciding non-vanishing of $c_{\a\b}^\g$.

\chapter{Littlewood-Richardson coefficients (cont)}
\begin{center}  {\Large Scribe: Paolo Codenotti} \end{center}

\noindent{\bf Goal:} We continue our study of Littlewood-Richardson 
coefficients and define Littlewood-Richardson coefficients 
for the orthogonal group $O_n(\C)$. 

\noindent {\em References:} \cite{FulH,YT} 

\subsection*{Recall}
Let us first recall some definitions and results from the last
class. Let $c_{\alpha, \beta}^\gamma$ denote the
Littlewood-Richardson coefficient for $GL_n(\C)$.
\begin{thm}[last class]
Non-vanishing of $c_{\alpha, \beta}^\gamma$ can be decided in
poly$(\bitlength{\alpha}, \bitlength{\beta}, \bitlength{\gamma})$ time,
 where $\langle \ \rangle$ denotes the bit length.
\end{thm}

The positivity hypotheses which hold here are:
\begin{itemize}
\item
$c_{\alpha, \beta}^\gamma \in \# P$, and more strongly,
\item \textbf{Positivity Hypothesis 1 (PH1):}
There exists a  polytope $P_{\alpha, \beta}^\gamma$ of dimension polynomial
in the heights of $\alpha,\beta$ and $\gamma$  such that  $c_{\alpha,
\beta}^\gamma=\varphi(P_{\alpha, \beta}^\gamma)$, where $\varphi$ indicates the number of integer points.
\item \textbf{Saturation Hypothesis (SH):} If
$c_{k\alpha, k\beta}^{k\gamma}\neq 0$ for some $k\geq 1$, then $c_{\alpha, \beta}^{\gamma}\neq 0$ [Saturation
Theorem].
\end{itemize}

\begin{proof}(of theorem)

PH$1$ + SH + Linear programming.
\end{proof}
This is the general form of algorithms in GCT. The main principle is that linear programming works for 
integer programming  when PH1 and SH hold.

\section{The stretching function}

We define $\widetilde{c}^\gamma_{\alpha, \beta} (k) = c_{k\alpha, k\beta}^{k\gamma}$.

\begin{thm}[Kirillov, Derkesen Weyman \cite{Der, Ki}]\label{thm:lrc}
$\widetilde{c}^\gamma_{\alpha, \beta} (k)$ is a polynomial in $k$.
\end{thm}

Here we  prove a weaker result. For its statement,
we will quickly review the theory of Ehrhart
quasipolynomials (cf. Stanley \cite{Sta}).

\begin{defi}(\textbf{Quasipolynomial})
A function $f(k)$ is called 
a \emph{quasipolynomial} if
 there exist polynomials $f_i$, $1\leq i\leq \ell$, for some $\ell$
 such that 
\[f(k)=f_i(k)\ \textrm{if}\ k\equiv i\ \textrm{mod}\ \ell.\]
We  denote such a quasipolynomial $f$ by $f=(f_i)$. Here
$\ell$ is called the period of $f(k)$ (we can assume it is the
smallest such period). The degree of a quasipolynomial $f$ is the
max of the degrees of the $f_i$'s.
\end{defi}

Now let $P\subseteq \R ^m$ be a polytope given by $Ax\leq b$. Let $\varphi(P)$ be the number of integer points
inside $P$. We define the stretching function $f_P(k)=\varphi(kP)$, where $kP$ is the dilated polytope defined
by $Ax\leq kb$.

\begin{thm}(Ehrhart) The stretching function 
$f_P(k)$ is a quasipolynomial. 
Furthermore, $f_P(k)$ is a polynomial if $P$ is an
integral polytope (i.e. all vertices of $P$ 
are integral).
\end{thm}

In view of this result, $f_P(k)$ is called the Ehrhart quasi-polynomial of $P$.
Now  $\widetilde{c}^\gamma_{\alpha, \beta} (k)$ 
is just the Ehrhart quasipolynomial of $P^\gamma_{\alpha,
\beta}$, and 
 $c^\gamma_{\alpha, \beta}=\varphi(P_{\alpha, \beta}^\gamma)$, the number of
integer points in $P_{\alpha,\beta}^\gamma$.
Moreover  $P_{\alpha,
\beta}^\gamma$ is defined by the inequality $Ax\leq b$, where $A$ is constant, and $b$ is a homogeneous linear
form in the coefficients of $\alpha$, $\beta$, and $\gamma$.

However, $P_{\alpha, \beta}^\gamma$ need not be integral. Therefore
Theorem (\ref{thm:lrc}) does not follow
from Ehrhart's result. Its proof needs representation theory.

\begin{defi}
A quasipolynomial $f(k)$ is said to be \emph{positive} if all the coefficients of $f_i(k)$ are nonnegative. In
particular, if $f(k)$ is a polynomial, then it's positive if all its coefficients are nonnegative.
\end{defi}

The Ehrhart quasipolynomial of a polytope is  positive
only in exceptional cases. In this context: 

\noindent {\bf PH$2$}  (positivity hypothesis 
$2$) \cite{KTT}: The polynomial  $\widetilde{c}^\gamma_{\alpha, \beta}(k)$ is
positive.

There is considerable computer evidence for this.

\begin{prop}
PH$2$ implies SH.
\end{prop}
\begin{proof} Look at:
\[c(k)=\widetilde{c}^\gamma_{\alpha, \beta}(k)=\sum a_i k^i.\]
If all the coefficients $a_i$ are nonnegative (by PH$2$), and $c(k)\neq 0$, then $c(1)\neq 0$.
\end{proof}

SH has a proof involving algebraic geometry \cite{Bl}. Therefore we suspect 
that the stronger  PH$2$ is a deep phenomenon related to 
algebraic geometry.

\section{$O_n(\C)$}
So far we have  talked about $GL_n(\C)$. Now we move on 
to the orthogonal group  $O_n(\C)$.
Fix $Q$, a symmetric bilinear form on $C^n$; for example,
$Q(V, W)= V^T W$. 

\begin{defi}
The orthogonal group $O_n(\C)\subseteq GL_n(\C)$ is the group
consisting of all $A\in GL_n(\C)$ s.t.  $Q(AV, AW) = Q(V, W)$
for all $V$ and $W\in \C^n$.
The subgroup $SO_n(\C)\subseteq SL_n(\C)$, where $SL_n(\C)$ is the set of matrices with determinant $1$, is defined similarly.
\end{defi}

\begin{thm}[Weyl]
The group $O_n(\C)$ is reductive
\end{thm}
\begin{proof}
The proof is similar to the reductivity of $GL_n(\C)$, based on Weyl's
unitary trick.
\end{proof}

The next step is to classify all irreducible polynomial
representations of $O_n(\C)$. Fix a partition $\lambda = 
(\lambda_1\geq \lambda_2\geq \dots)$ of length at most $n$.
 Let $|\lambda|= d = \sum \lambda_i$ be its size.
Let $V = \C^n$, 
$V^{\otimes d}= V \otimes \dots \otimes V\ d\ \textrm{times}$, and 
 embed the Weyl module
$V_\lambda$ of $GL_n(\C)$  in $V^{\otimes d}$ 
as per Theorem~\ref{tembedweyl}.
Define a contraction map
\[\varphi_{p,q}:V^{\otimes d}\rightarrow V^{\otimes(d-2)}\]
for $1\leq p\leq q \leq d$ by:

\[\varphi_{p,q}(v_{i_1}\otimes \dots \otimes v_{i_d}) = Q(v_{i_p},
v_{i_q})(v_{i_1}\otimes \dots \otimes \widehat{v_{i_p}}\otimes \dots
\otimes \widehat{v_{i_q}}\otimes \dots \otimes {v_{i_d}}),\] 
where $\widehat{v_{i_p}}$ means omit $v_{i_p}$.

It is  $O_n(\C)$-equivariant, i.e. the following diagram commutes: 
\[\begin{CD}
V^{\otimes d} @>\varphi_{p,q}>> V^{\otimes d-2}\\
@VV\sigma\in O_n(\C)V @VV\sigma\in O_n(\C)V\\
V^{\otimes d} @>\varphi_{p,q}>> V^{\otimes d-2}
\end{CD}\]

Let $$V^{[d]}=\bigcap_{pq} ker(\varphi_{p,q}).$$ Because the maps are equivariant, each kernel is an
$O_n(\C)$-module, and $V^{[d]}$ is an $O_n(\C)$-module.
Let $V_{[\lambda]}=V^{[d]}\bigcap V_\lambda$, where $V_\lambda
\subseteq V^{\otimes d}$ is the  embedded
 Weyl module as above. Then  $V_{[\lambda]}$ is an
$O_n(\C)$-module.

\begin{thm}[Weyl]

$V_{[\lambda]}$ is an irreducible representation of $O_n(\C)$. Moreover, the following two conditions hold:

\begin{enumerate}
\item If $n$ is odd, then $V_{[\lambda]}$ is non-zero if and only if the sum of the lengths of the first two
columns of $\lambda$ is $\leq n$ (see figure \ref{fig}).
\begin{figure}
    \begin{center}
      \psfragscanon
      \psfrag{L}{$\lambda$}
      \epsfig{file=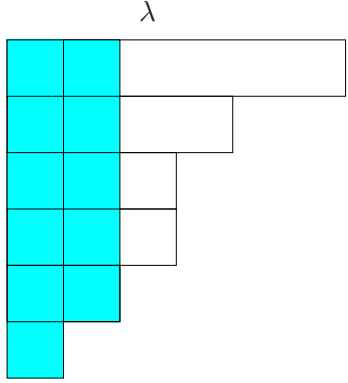, scale=.5}
      \caption{{\small The first two columns of the partition $\lambda$ are highlighted.}}
      \label{fig}
    \end{center}
\end{figure}
\item
If $n$ is odd, then each polynomial irreducible representation is isomorphic to $V_{[\lambda]}$ for some
$\lambda$.
\end{enumerate}
\end{thm}

Let
\[V_{[\lambda]} \otimes V_{[\mu]} = \otimes_\gamma d_{\lambda, \mu}^\gamma V_{[\gamma]}\]

be the decomposition of $V_{[\lambda]}\otimes V_{[\mu]}$ into irreducibles.
Here $d_{\lambda, \mu}^\gamma$ is called the Littlewood-Richardson
coefficient of type B. The types of various connected 
reductive groups are defined as follows:
\begin{itemize}
\item $GL_n(\C)$: type A
\item $O_n(\C)$, $n$ odd: type B
\item $Sp_n(\C)$: type C
\item $O_n(\C)$, $n$ even: type D
\end{itemize}
The  Littlewood-Richardson coefficient can be defined for
any type in a similar fashion. 

\begin{thm}[Generalized Littlewood-Richardson rule]
The Littlewood-Richardson coefficient
$d_{\lambda, \mu}^\gamma \in \#P$. This also holds for any type.
\end{thm}
\begin{proof}
The most transparent proof of this theorem comes through the theory of
quantum groups \cite{Kas}; cf.  Chapter~\ref{ctowards}.
\end{proof}

As in type $A$ this leads to:
\begin{hyp}[PH$1$]
There exists a polytope $P_{\lambda, \mu}^\gamma$ of dimension
polynomial in the heights of $\lambda,\mu$ and $\gamma$  such that:
\begin{enumerate}
\item
$d_{\lambda, \mu}^\gamma=\varphi(P_{\lambda, 
\mu}^\gamma)$, the  number of integer points in 
$P_{\lambda, \mu}^\gamma$,   and
\item $\widetilde{d}_{\lambda,
  \mu}^\gamma(k)=d_{k\lambda, k\mu}^{k\gamma}$ is the Ehrhart 
quasipolynomial of $P_{\lambda, \mu}^\gamma$. 
\end{enumerate}
\end{hyp}

There are several
choices for such polytopes; e.g. the BZ-polytope \cite{BZ}. 

\begin{thm}[De Loera, McAllister \cite{DM2}]
The stretching function $\widetilde{d}_{\lambda, \mu}^\gamma(k)$ is a
quasipolynomial of degree at most $2$; so also for types $C$ and $D$.
\end{thm}

A verbatim translation of the saturation property fails here
\cite{Z}): there exist
$\lambda, \mu$ and $\gamma$ such that $d_{2\lambda,
  2\mu}^{2\gamma} \neq 0$ but $d_{\lambda,  \mu}^\gamma=0$.
Therefore we  change the definition of saturation:

\begin{defi}
Given a quasipolynomial $f(k)=(f_i)$,  $index(f)$ is
the smallest $i$ such that $f_i(k)$ is not an
identically zero polynomial. If $f(k)$ is identically zero,  $index(f)=0$.
\end{defi}

\begin{defi}
A quasipolynomial $f(k)$ is \emph{saturated} if $f(index(f))\neq 0$.
In particular, if $index(f)=1$, then $f(k)$ is saturated if
$f(1)\neq 0$.
\end{defi}

A positive quasi-polynomial is clearly saturated.

\vspace{.1in}

\noindent \textbf{Positivity Hypothesis 2 (PH2)} \cite{DM2}: The stretching
  quasipolyomial $\widetilde{d}_{\lambda, \mu}^\gamma(k)$ is positive.

  \vspace{.1in}

There is
considerable evidence for this.

  \vspace{.1in}

\noindent \textbf{Saturation Hypothesis (SH)}: The stretching
quasipolynomial $\widetilde{d}_{\lambda, \mu}^\gamma(k)$ is saturated.

  \vspace{.1in}

PH2 implies SH.

\begin{thm} \cite{GCT5}
Assuming SH (or PH$2$),
positivity of the Littlewood-Richardson coefficient 
$d_{\lambda, \mu}^\gamma$ of type $B$  can be decided in
$poly(\bitlength{\lambda}, \bitlength{\mu}, \bitlength{\gamma})$ time.
\end{thm}

This  is also true for all types.

\begin{proof}
next class.
\end{proof}

\chapter{Deciding nonvanishing of Littlewood-Richardson coefficients for 
$O_n(\C)$} \label{cgct5}
\begin{center} {\Large Scribe: Hariharan Narayanan} \end{center}

\noindent {\bf Goal:} A polynomial time algorithm for 
deciding nonvanishing of Littlewood-Richardson coefficients for
the orthogonal group assuming SH.

\noindent {\em Reference:} \cite{GCT5}

Let $d_{\lambda,\mu}^\nu$ denote the Littlewood-Richardson coefficient of
type $B$ (i.e. for the orthogonal group $O_n(\C)$, $n$ odd) 
as defined in the earlier lecture.
In this lecture we describe a polynomial time algorithm for
deciding nonvanishing of $d_{\lambda,\mu}^\nu$ assuming the 
following positivity hypothesis PH2. Similar result also holds for
all types, though we shall only concentrate on type B in this lecture.

Let $\tilde d_{\lambda,\mu}^\nu(k)=
d_{k \lambda, k\mu}^{k \nu}$ denote the associated stretching function.
It is known to be a quasi-polynomial of period at most two 
\cite{DM2}. This means there
are polynomials $f_1(k)$ and $f_2(k)$ such that 
$$d_{k\lambda, k\mu}^{k\nu} = \left\{
                               \begin{array}{ll}
                                 f_1(k), & \hbox{if k is odd;} \\
                                 f_2(k), & \hbox{if k is even.}
                               \end{array}
                             \right.$$

\noindent {\bf Positivity Hypothesis (PH2)} \cite{DM2}:
The stretching quasi-polynomial $\tilde d_{\lambda,\mu}^\nu(k)$ is positive.
This means the coefficients of $f_1$ and $f_2$ are all non-negative.

The main result in this lecture is:

\begin{thm} \cite{GCT5} \label{tgct5}
If PH2 holds, then 
the problem of deciding the positivity (nonvanishing)  of 
$d_{\lambda\mu}^\nu$  belongs to $P$. That is, this problem
can be solved in time polynomial in the bitlengths of $\lambda,\mu$ and
$\nu$.
\end{thm}

We need a few lemmas for the proof.

\begin{lemma}\label{reduc}
If PH2 holds,  the following are
equivalent:
\begin{enumerate}
\item[(1)] $d_{\lambda \mu}^\nu \geq 1$.
\item[(2)]  There exists an odd integer $k$ such that $d_{k\lambda \, k\mu}^{k\nu} \geq 1$.
\end{enumerate}
\end{lemma}
{\bf Proof:} Clearly $(1)$ implies $(2)$. By PH2, 
there exists a polynomial $f_1$
with non-negative coefficients such that
$$\forall \text{ odd } k,\, f_1(k) = d_{k\lambda \, k\mu}^{k\nu}.$$
Suppose that for some odd $k$, $d_{k\lambda \, k\mu}^{k\nu} \geq 1.$
Then $f_1(k) \geq 1$. Therefore $f_1$ has at least one non-zero
coefficient. Since all coefficients of $f_1$ are nonnegative,
$d_{\lambda \mu}^{\nu} = f_1(1) > 0$. Since $d_{\lambda \mu}^{\nu}$ is
an integer, $(1)$ follows.
{\hfill $\Box$}\\

\begin{defi}
Let $\Z_{<2>}$ be  the subring of \,$\Q$ obtained by localizing $\Z$ at $2$:
$$\Z_{<2>} := \left\{\frac{p}{q} \mid p, \frac{q-1}{2} \in \Z \right\}.$$
This ring consists of all fractions whose denominators are odd.
\end{defi}

\begin{lemma}\label{aff}
Let $P \in \R^d$ be a convex polytope specified by $Ax \leq
B$, $x_i \geq 0$ for all $i$, where $A$ and $B$ are integral. 
Let $\af(P)$ denote its affine span. The
following are equivalent:
\begin{enumerate}
\item[(1)] $P$ contains a point in $\Z_{<2>}^d$.
\item[(2)] $\af(P)$ contains a point in $\Z_{<2>}^d$.
\end{enumerate}
\end{lemma}
{\bf Proof:} Since $P \subseteq \af(P)$, $(1)$ implies $(2)$. 
Now suppose $(2)$ holds. We have to show $(1)$. 
Let $z \in \ztn \cap \af(P)$.

First, consider the case when $\af(P)$ is  one dimensional. In this
case, $P$ is the line segment joining two points $x$ and $y$ in
$\Q^d$. The point $z$ can be expressed as an affine linear combination, $z =
ax + (1-a)y$ for some $a \in \Q$. There exists $q \in \Z \text{ such
that } qx \in \ztn \text{ and } qy \in \ztn.$ Note that $$\{z +
\lambda(qx - qy) \mid \lambda \in \z2\} \subseteq  \af(P) \cap \ztn.$$ 
Since $\z2$
is a dense subset of $\Q$, the l.h.s. and hence the r.h.s.
is a dense subset of $\af(P)$. Consequently, $P \cap \ztn \neq \emptyset$.

Now consider the general case.
Let $u$ be any point in the interior of $P$ with rational coordinates,
and  $L$ the line through $u$ and $z$. 
By restricting to $L$, the lemma reduces to
the preceding one dimensional case.
 {\hfill $\Box$}

\begin{lemma}\label{smith}
Let  $$P = \left\{x \mid Ax \leq
B,(\forall i) x_i \geq 0\right\} \subseteq \R^d$$
be a convex polytope  where $A$ and $B$ are integral.
Then, it is possible to determine in polynomial time whether or not
$\af(P) \cap \ztn = \emptyset$.
\end{lemma}
{\bf Proof:} Using Linear Programming \cite{Kha79, Kar84}, a
presentation of the form $Cx = D$ can be obtained for $\af(P)$ in
polynomial time, where $C$ is an integer matrix and $D$ is a vector
with integer coordinates. We may assume that $C$ is square since
this can be achieved by padding it with $0$'s if necessary, and
extending $D$. The Smith Normal Form over $\Z$ of $C$ is a matrix
$S$ such that $C = USV$  where $U$ and $V$ are unimodular and $S$
has the form
$$\left(
  \begin{array}{cccc}
    s_{11} & 0 & \dots & 0 \\
    0 & s_{22} & \dots & 0 \\
    \vdots & \vdots & \ddots & 0 \\
    0 & 0 & \dots & s_{dd} \\
  \end{array}
\right)$$ where for $1 \leq i \leq d-1$, $s_{ii}$ divides $s_{i+1\,
i+1}$. It  can be computed  in polynomial
time \cite{KB79}. The question now reduces to whether $USVx = D$ has
a solution $x \in \ztn$. Since $V$ is unimodular, its inverse has
integer entries too, and $y:=Vx \in \ztn \Leftrightarrow x \in
\ztn$. This is equivalent to whether $Sy = U^{-1}D$ has a solution
$y \in \ztn$. Since $S$ is diagonal, this can be answered in
polynomial time simply by checking each coordinate. {\hfill $\Box$}

{\bf Proof of Theorem~\ref{tgct5}:}
By \cite{BZ}, there exists a polytope $P=P_{\lambda,\mu}^\nu$ such that
the Littlewood-Richardson coefficient  $d_{\lambda\mu}^\nu$ is equal to
the number of integer points in $P$.
This polytope is such that 
the number of integer points in the dilated polytope $kP$ is
$d_{k\lambda \,k\mu}^{k\nu}$. Assuming PH2,
we know from Lemma~\ref{reduc} that
$$P \cap \Z^d \neq \emptyset \Leftrightarrow (\exists
\text{ odd } k), kP \cap \Z^d \neq \emptyset.$$ The latter is
equivalent to $$P \cap \ztn \neq \emptyset.$$ The theorem  now
follows from Lemma~\ref{aff} and Lemma~\ref{smith}. {\hfill $\Box$}

In combinatorial optimization, LP works if the polytope is
integral. In our setting, this is not necessarily the case
\cite{loera}: the denominators of the coordinates of 
 the vertices of $P$  can be $\Omega(l)$, where $l$ is the total height
of $\lambda,\mu$ and $\nu$.
LP works here nevertheless because of PH2; it can be checked that
SH is also sufficient.

\ignore{

}

\chapter{The plethysm problem}
\begin{center} {\Large  Scribe: Joshua A. Grochow} \end{center} 

\noindent {\bf Goal:} 
In this lecture we describe the  general \emph{plethysm problem},
state  analogous positivity and saturation hypotheses for it,
and state the results from GCT 6 which imply a polynomial time algorithm
for deciding positivity of a plethysm constant assuming these
hypotheses.

\noindent {\bf Reference:} \cite{GCT6} 

\subsection*{Recall}
Recall that  a function $f(k)$ is \emph{quasipolynomial} if there are functions $f_i(k)$ for $i=1,\dots,\ell$ such that $f(k) = f_i(k)$ whenever $k \equiv i \mbox{ mod } \ell$.  The number $\ell$ is then the \emph{period} of $f$.  The \emph{index} of $f$ is the least $i$ such that $f_i(k)$ is not identically zero.  If $f$ is identically zero, then the index of $f$ is zero by convention.  We say $f$ is \emph{positive} if all the coefficients of each $f_i(k)$ are nonnegative.  We say $f$ is \emph{saturated} if $f(\mbox{index}(f)) \neq 0$.  If $f$ is positive, then it is saturated.

Given any function $f(k)$, we  associate to it the rational series $F(t) = \sum_{k \geq 0} f(k) t^k$.

\begin{prop} \label{pstan} \cite{Sta} The following are equivalent:
\begin{enumerate}
\item $f(k)$ is a quasipolynomial of period $\ell$.
\item $F(t)$ is a rational function of the form $\frac{A(t)}{B(t)}$ where $\deg A < \deg B$ and every root of $B(t)$ is an $\ell$-th root of unity.
\end{enumerate}
\end{prop}

\section{Littlewood-Richardson Problem [GCT 3,5]}
Let $G = GL_n(\C)$ and  $c_{\alpha,\beta}^\gamma$
the Littlewood-Richardson coefficient -- i.e. the multiplicity of the Weyl module $V_\gamma$ in $V_\alpha \otimes V_\beta$.  We saw 
that the positivity of $c_{\alpha,\beta}^\gamma$ can be decided
in $poly(\bitlength{\alpha},\bitlength{\beta},\bitlength{\gamma})$ time,
where $\bitlength{\cdot}$ denotes the bit-length.  Furthermore, we saw that the stretching function 
$\widetilde{c}_{\alpha,\beta}^\gamma(k) = c_{k\alpha,k\beta}^{k\gamma}$ is a polynomial, and the analogous stretching function for type $B$ is
a quasipolynomial of period at most 2.

\section{Kronecker Problem [GCT 4,6]}
Now we study the   analogous problem  for the
representations of the symmetric group (the Specht modules), called
the Kronecker problem.

Let $S_\alpha$ be the Specht module of the symmetric group $S_m$ associated to the partition $\alpha$.  Define the Kronecker coefficient $\kappa_{\lambda,\mu}^\pi$ to be  the multiplicity of $S_\pi$ in $S_\lambda \otimes S_\mu$ (considered as an $S_m$-module via the diagonal action).  In other words, write $S_\lambda \otimes S_\mu = \bigoplus_\pi \kappa_{\lambda,\mu}^\pi S_\pi$. 
We have 
$ \kappa_{\lambda,\mu}^\pi =(\chi_\lambda\chi_\mu, \chi_\pi)$, where 
$\chi_\lambda$ denotes the character of $S_\lambda$. 
By the Frobenius character formula, this can be computed in PSPACE.
More strongly, analogous to the Littlewood-Richardson problem:

\begin{conj} \cite{GCT4,GCT6} 
The Kronecker coefficient 
$\kappa_{\lambda,\mu}^\pi \in \cc{\# P}$.
In other words, there is a positive $\#P$-formula
for $\kappa_{\lambda,\mu}^\pi$.  \end{conj}

This  is a fundamental problem in representation theory.
More concretely, it can be phrased as asking for a set of combinatorial objects $I$ and a characteristic function $\chi:\{I\} \to \{0,1\}$ such that $\chi \in \cc{FP}$ and $\kappa_{\lambda,\mu}^\pi = \sum_I \chi(I)$. 
Continuing our analogy:

\begin{conj} \cite{GCT6} 
The problem of  deciding  positivity of $\kappa_{\lambda,\mu}^\pi$ belongs
 to \cc{P}. \end{conj}

\begin{thm} \cite{GCT6} 
The stretching function $\widetilde{\kappa}_{\lambda,\mu}^\pi(k)=\kappa_{k\lambda,k\mu}^{k\pi}$ is a quasipolynomial.
 \end{thm}
Note that $\kappa_{k\lambda,k\mu}^{k\pi}$ is a Kronecker coefficient for $S_{km}$.

There is also a dual definition  of the Kronecker coefficients.
Namely, consider the embedding 
\[H=GL_n(\C) \times GL_n(\C) \hookrightarrow G= GL(\C^n \otimes \C^n),\]
 where $(g,h)(v \otimes w) = (gv \otimes hw)$.  Then

\begin{prop} \cite{FulH} The Kronecker coefficient 
$\kappa_{\lambda,\mu}^\pi$ is the multiplicity of the tensor product of Weyl modules $V_\lambda(GL_n(\C)) \otimes V_\mu (GL_n(\C))$ (this is an irreducible $H$-module) in the Weyl module $V_\pi(G)$ considered as an $H$-module via the embedding above. \end{prop}

\section{Plethysm Problem [GCT 6,7]}
Next  we consider the more general plethysm problem.

Let $H=GL_n(\C)$,  $V=V_\mu(H)$  the Weyl module of $H$ corresponding to
a partition  $\mu$,
 and  $\rho: H \to G=GL(V)$ the corresponding representation map.
 Then the Weyl module $V_\lambda(G)$ of $G$ for a given partition $\lambda$
can be considered an $H$-module via  the map $\rho$.  By complete reducibility, we may decompose this $H$-representation as
$$
V_\lambda(G) = \bigoplus_\pi a_{\lambda,\mu}^\pi V_\pi(H).
$$
The coefficients $a_{\lambda,\mu}^\pi$ are known as \emph{plethsym constants}
(this definition can easily be generalized to any reductive group $H$). 
The Kronecker coefficient  is a special case of the plethsym constant 
\cite{Ki}.

\begin{thm}[GCT 6] The plethysm constant 
$a_{\lambda,\mu}^\pi \in \mbox{PSPACE}$. \end{thm}

This is based on a parallel algorithm to compute the plethysm constant
using  Weyl's character  formula. Continuing in our previous trend:

\begin{conj} \label{cgct6pl} \cite{GCT6}  $a_{\lambda,\mu}^\pi \in \cc{\# P}$ and 
the problem of deciding 
positivity of $a_{\lambda,\mu}^\pi$ belongs to  $\cc{P}$. \end{conj}

For the stretching function, we need to be a bit careful.  Define $\widetilde{a}_{\lambda,\mu}^\pi = a_{k\lambda,\mu}^{k\pi}$. Here
the subscript $\mu$ {is not stretched}, since that  would change $G$, 
while stretching $\lambda$ and $\pi$ only alters the  representations of $G$.

As in  the beginning of the lecture, we can associate
a function $A_{\lambda,\mu}^\pi(t) = \sum_{k \geq 0} \widetilde{a}_{\lambda,\mu}^\pi(k) t^k$ to the plethysm constant. Kirillov conjectured that $A_{\lambda,\mu}^\pi(t)$ is rational. 
In view of Proposition~\ref{pstan}, this follows from the  following stronger result:

\begin{thm} [GCT 6] \label{thm:qp} The stretching function
$\widetilde{a}_{\lambda,\mu}^\pi(k)$ is a quasipolynomial.  \end{thm}

This is the main result of GCT 6, which  in some sense
allows  GCT to go forward.  Without it, there would be  little hope for 
proving  that the positivity of plethysm constants can be decided 
in polynomial time.  
Its  proof  is essentially algebro-geometric.
The basic idea is to show that the stretching function is the Hilbert function of some algebraic variety with nice (i.e. ``rational'') singularities.
Similar results are shown  for 
the   stretching functions in the algebro-geometric problems arising in GCT.

The main complexity-theoretic result in
\cite{GCT6}  shows that, under the following
positivity and saturation hypotheses (for which there is much experimental evidence), the positivity of the plethysm constants can indeed  be decided in
polynomial time (cf.  Conjecture~\ref{cgct6pl}).

The first positivity hypothesis is suggested by Theorem~\ref{thm:qp}:
 since the stretching function is a quasipolynomial, we may suspect that it is captured by some polytope:

\vspace{.1in}

\textbf{Positivity Hypothesis 1 (PH1).}  
There exists a polytope $P=P_{\lambda,\mu}^\pi$ such that:
\begin{enumerate} 
\item  $a_{\lambda,\mu}^\pi = \varphi(P)$, where $\varphi$
denotes the number of integer points inside the polytope,
\item The stretching quasipolynomial (cf. Thm. \ref{thm:qp}) 
$\widetilde{a}_{\lambda,\mu}^\pi(k)$ is equal to
the Ehrhart quasipolynomial $f_P(k)$ of $P$,
\item The dimension of $P$ is polynomial 
in $\bitlength{\lambda}, \bitlength{\mu}$, and $\bitlength{\pi}$,
\item the membership in $P_{\lambda,\mu}^\pi$ can be decided in
$\poly(\bitlength{\lambda},\bitlength{\mu},\bitlength{\pi})$ 
 time, and there is a polynomial time separation oracle \cite{lovasz} for $P$.
\end{enumerate}

Here  (4) does {not} imply that the polytope $P$
has only polynomially many constraints. In fact,
in the plethysm problem there may be a super-polynomial number of constraints.

\vspace{.1in}

\textbf{Positivity Hypothesis 2 (PH2).} The stretching quasipolynomial $\widetilde{a}_{\lambda,\mu}^\pi(k)$ is positive.

\vspace{.1in}

This implies:

\vspace{.1in}

\textbf{Saturation Hypothesis (SH).}  The stretching quasipolynomial is saturated.  

\vspace{.1in}

Theorem \ref{thm:qp} is essential to state these hypotheses,
since positivity and saturation are properties that only apply
to quasipolynomials.
Evidence for PH1, PH2, and SH can be found in GCT 6.  

\begin{thm} \cite{GCT6} Assuming PH1 and SH (or PH2), 
positivity of the plethysm constant $a_{\lambda,\mu}^\pi$
 can be decided in
$\poly(\bitlength{\lambda},\bitlength{\mu},\bitlength{\pi})$ time. \end{thm}
This follows from the polynomial time algorithm for saturated integer 
programming described in the next class.
As with Theorem~\ref{thm:qp}, this  also holds for more general problems in
algebraic geometry.  

\ignore{It follows from the 
The proof of this theorem is essentially just showing that, under the stated hypotheses, a variant of linear programming suffices for the integer program that decides if the polytope $P_{\lambda,\mu}^\pi$ contains an integer point: although general integer programming is NP-complete, saturated (positive) integer programming is in P.  PH1 and SH imply that this integer program is in fact a saturated integer program (PH1 and PH2 imply it is a positive integer program).  
}

\chapter{Saturated and positive integer programming}
\begin{center} {\Large  Scribe: Sourav Chakraborty} \end{center}

\textbf{Goal : } A polynomial time algorithm for 
saturated integer programming and its 
application to  the plethysm problem. 

\noindent {\em Reference:} \cite{GCT6}

\

\textbf{Notation : } In this class we denote by $\langle a \rangle$
the bit-length of the $a$.

\section{Saturated, positive integer programming}
Let $Ax \leq b$ be a set of inequalities. The number of constraints
can be exponential. Let $P\subset \mathbb{R}^n$ be the polytope
defined by these inequalities. The bit length of P is defined to be $\langle
P\rangle = n + \psi$, where $\psi$ is  the maximum
bit-length of a constraint in the set of inequalities. We assume
that P is given by a separating oracle. This means membership in $P$ 
can be decided  in
poly$(\langle P\rangle)$ time, and  if $x\not\in P$ then a separating
hyperplane is given as a proof as in \cite{lovasz}.

Let $f_P(k)$ be the Ehrhart quasi-polynomial of P. Quasi-polynomiality 
means there exist polynomials $f_i(k)$, $1 \le i \le l$, $l$ the period,
so that $f_P(k)=f_i(k)$ if $k=i$ modulo $l$. 
Then
$$\mbox{Index($f_P$) = min}\{i | f_i(k) \mbox{not identically $0$ as a polynomial}\}$$

The integer programming problem  is called {\em positive}
 if $f_P(k)$ is positive
whenever $P$ is non-empty, and 
{\em saturated} if $f_P(k)$ is saturated whenever $P$ is non-empty.

\begin{thm}[GCT6] \label{tgct6sat}

\begin{enumerate} \item Index$(f_P)$ can be computed
in time polynomial in the bit length $\langle P\rangle$ of $P$
assuming that the separation oracle works in poly-$\langle
P\rangle$-time.

\item Saturated and hence positive integer programming problem can be
solved in  poly-$\langle P\rangle$-time.
\end{enumerate}
\end{thm}

The second statement follow from the first.

\begin{proof} Let $Aff(P)$ denote the affine span
of P. By \cite{lovasz} we can compute the specifications $Cx = d$, $C$ and $d$
integral,   of
$Aff(P)$ in poly$(\langle P\rangle)$ time. Without loss of
generality, by padding, we can assume that $C$ is square. By
 \cite{KB79} we find the Smith-normal form of $C$ in polynomial time.
 Let it be
$\bar{C}$. So,
$$\bar{C} = ACB$$ where $A$ and $B$ are unimodular, and $\bar{C}$
is a  diagonal matrix, where the diagonal entries $c_1,c_2, \dots$
are such that with $c_i | c_{i+1}$.

Clearly $Cx = d$ iff $\bar{C}z = \bar{d}$ where $z = B^{-1}x$ and
$\bar{d} = Ad$.

So all equations here are of form 
\begin{equation}\label{eq} \bar{c_i}z_i
= d_i\end{equation} 

 Without loss of generality we can assume that
$c_i$ and $d_i$ are relatively prime. Let $\tilde{c} = lcm (c_i)$.

\

\begin{claim} $Index(f_P) = \tilde{c}$.
\end{claim}

From this claim the theorem clearly follows.

\

\begin{proof}[Proof of the claim] Let $f_P(t) = \sum_{k\ge  0}f_P(k)t^k$ be the
Ehrhart Series of  $P$.

Now $kP$ will not have an integer point unless $\tilde{c}$ divides
$k$ because of  (\ref{eq}).

Hence $f_P(t) = f_{\bar{P}}(t^{\tilde{c}})$ where $\bar{P}$ is the
stretched polytope $\tilde{c}P$ and $f_{\bar{P}}(s)$ is the Ehrhart
series of $\bar{P}$. From this it follows that
$$Index(f_P) = \tilde{c}Index(f_{\bar{P}})$$

Now we show that $Index(f_{\bar{P}}) = 1$.

The equations of $\bar{P}$ are of the form $$z_i =
\frac{\tilde{c}}{c_i}d_i$$ where each $\frac{\tilde{c}}{c_i}$ is an
integer. Therefore without loss of generality we can ignore these
equations and assume the $\bar{P}$ is full dimensional.

Then it suffices to show that $\bar{P}$ contains a rational point
whose denominators are all 1 modulo $\ell(\bar{P})$, the
period of the quasi-polynomial $f_{\bar{P}}(s)$.

This follows from a simple density argument that we saw earlier (cf.
the proof of Lemma~\ref{aff}).

From this the claim follows.
\end{proof}

\end{proof}

\section{Application to the plethysm problem}

Now we can prove the result stated in the last class:

\begin{thm} Assuming  PH1 and SH, positivity of 
the plethysm constant $a^{\pi}_{\lambda,
\mu}$ can be decided  in time polynomial in 
$\langle\lambda
\rangle,\langle\mu \rangle$ and $\langle\pi \rangle$.
\end{thm}

\begin{proof} 
Let $P=P_{\lambda,\mu}^\pi$ be the polytope as in PH1  such that
$a^{\pi}_{\lambda, \mu}$ is the number of integer points in $P$. 
The goal is to decide if $P$ contains an integer point. This integer 
programming problem is saturated because of SH. 
Hence the result follows from Theorem~\ref{tgct6sat}. 
\end{proof}

\chapter{Basic algebraic geometry}
\begin{center}  {\Large  Scribe: Paolo Codenotti} \end{center}

\noindent {\bf Goal:} 
So far we have focussed on purely representation-theoretic aspects of GCT. 
Now we have to bring in algebraic geometry.
In this lecture we review the   basic definitions and results
in algebraic geometry that will be needed for this purpose.
The proofs will be omitted or only sketched. For details,
see the  books by Mumford \cite{Mum} and Fulton \cite{YT}.

\section{Algebraic geometry definitions}

Let $V=\C^n$, and $v_1, \dots, v_n$  the coordinates of $V$.

\begin{defi}
\begin{itemize}
\item
$Y$ is an \emph{affine algebraic set} in $V$ if $Y$ is the set of simultaneous zeros of a set of polynomials in  $v_i$'s.
\item
An algebraic set that cannot be written as the union of two proper algebraic sets $Y_1$ and $Y_2$ is called \emph{irreducible}.
\item
An irreducible affine algebraic set is called an \emph{affine variety}.
\item
The ideal of an affine algebraic set $Y$ is $I(Y)$, the set of all polynomials that vanish on $Y$.
\end{itemize}
\end{defi}

For example, $Y=(v_1-v_2^2+v_3, v_3^2-v_2+4v_1)$ is an irreducible affine algebraic set (and therefore an affine variety).

\begin{thm}[Hilbert]
$I(Y)$ is finitely generated, i.e. there exist polynomials $g_1, \dots, g_k$ that generate $I(Y)$. This means
 every $f\in I(Y)$ can be written as $f=\sum f_i g_i$ for some polynomials $f_i$.
\end{thm}

Let $\C[V]$, the coordinate ring of $V$, be the ring of polynomials
over the variables $v_1, \dots, v_n$. The coordinate ring
of $Y$ is defined to be  $\C[Y]=\C[V]/I(Y)$. It 
 is the set of polynomial functions over $Y$.

\begin{defi}
\begin{itemize}
\item
$P(V)$ is the \emph{projective space} associated with $V$, i.e. the set of lines through the origin in $V$.
\item
$V$ is called the \emph{cone} of $P(V)$.
\item
$\C[V]$ is called the \emph{homogeneous coordinate ring} of $P(V)$.
\item
$Y\subseteq P(V)$ is a \emph{projective algebraic set} if it is the set of simultaneous zeros of a set of homogeneous forms (polynomials) in the variables $v_1, \dots, v_n$.
It is necessary that the polynomials be homogeneous because a point in $P(V)$ is a line in $V$.
\item
A projective algebraic set $Y$ is \emph{irreducible} if it can not be expressed as the union of two proper algebraic sets in $P(V)$.
\item
An irreducible projective algebraic set is called a \emph{projective variety}.
\end{itemize}
\end{defi}

Let $Y\subseteq P(V)$ be a projective variety, and define $I(Y)$, the
ideal of $Y$ to be the set of all homogeneous forms that vanish on $Y$.
Hilbert's result implies that  $I(Y)$ is finitely generated.

\begin{defi}
The \emph{cone} $C(Y) \subseteq V$ 
of  a projective variety $Y\subseteq P(V)$ is defined
to be the set of all points on the lines in $Y$.
\end{defi}

\begin{defi}
We define the \emph{homogeneous coordinate ring} of $Y$ as $R(Y) = \C[V]/I(Y)$,
 the set of homogeneous polynomial forms on the cone of $Y$.
\end{defi}

\begin{defi}
A \emph{Zariski open subset} of $Y$ is the complement of a projective 
 algebraic
subset of $Y$. It  is called a \emph{quasi-projective variety}.
\end{defi}

Let $G=GL_n(\C)$, and  $V$ a finite dimensional representation of $G$.
Then $\C[V]$ is a $G$-module, with the action of $\sigma \in G$  defined by:
\[(\sigma\cdot f)(v) = f(\sigma^{-1} v),\ v\in V.\]
\begin{defi}
Let $Y\subseteq P(V)$ be a projective variety with ideal $I(Y)$. We say that $Y$ is a \emph{$G$-variety} if $I(Y)$ is a $G$-module, i.e.,
 $I(Y)$ is a $G$-submodule of $\C[V]$.
\end{defi}
If $Y$ is a projective variety, then $R(Y)=\C[V]/I(Y)$ is also a $G$-module. Therefore $Y$ is $G$-invariant, i.e.
\[y\in Y \Implies \sigma y \in Y,\ \forall \sigma \in G.\]
The algebraic geometry of $G$-varieties is called geometric invariant theory
(GIT).

\section{Orbit closures}
We now define  special classes of $G$-varieties called \emph{orbit closures}.
 Let $v\in P(V)$ be a point, and 
$Gv$ the orbit of $v$: 
\[Gv=\{gv|g\in G\}.\] Let the stabilizer of $v$ be 
\[H=G_v=\{g\in G| gv = v\}.\]

The orbit $Gv$ is isomorphic to the space   $G/H$ of cosets, called the
 homogeneous space. This is  a very special kind of algebraic
variety.

\begin{defi}
The   \emph{orbit closure}  of $v$ is defined by:
\[\Delta_V[v]=\overline{Gv}\subseteq P(V).\]
Here  $\overline{Gv}$ is the closure of the orbit $G v$ 
in the complex topology on $P(V)$ (see figure \ref{fig:orbit}).
\end{defi}

\begin{figure}
    \begin{center}
      \psfragscanon
      \psfrag{v}{$v$}
      \psfrag{Gv}{$Gv$}
      \psfrag{Dv}{$\Delta_V[v]$}
      \psfrag{Closure}{Limit points of $Gv$}
      \epsfig{file=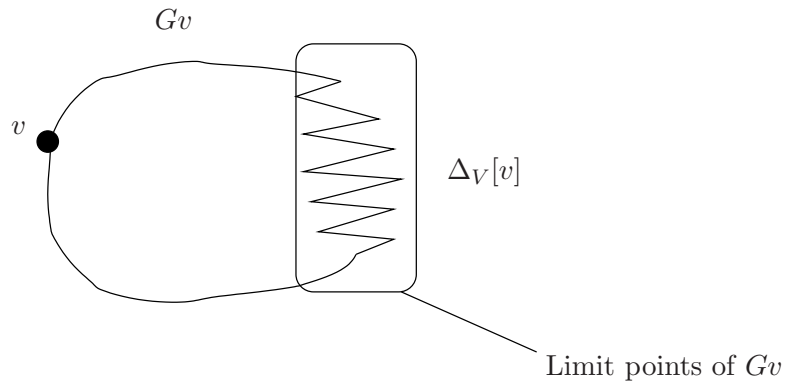, scale=1}
      \caption{{\small The limit points of $Gv$ in $\Delta_V[v]$ can be horrendous.}}
      \label{fig:orbit}
    \end{center}
\end{figure}

A basic fact of algebraic geometry:

\begin{thm}
The orbit closure $\Delta_V[v]$ is a projective $G$-variety
\end{thm} 

It is also called an almost homogeneous space.

Let $I_V[v]$ be the ideal of $\Delta_V[v]$, and $R_V[v]$ 
the homogeneous coordinate ring of $\Delta_V[v]$.
The  algebraic geometry of general orbit closures is hopeless,
since the closures can be horrendous (see figure
\ref{fig:orbit}). Fortunately we shall only be interested in very 
special kinds of orbit closures with good algebraic geometry.

We now define    the
simplest kind of  orbit closure, 
which is obtained  when the orbit itself is closed.
Let $V_\lambda$ be an irreducible Weyl module of $GL_n(\C)$, where $\lambda=(\lambda_1\geq\lambda_2\geq \dots \geq \lambda_n\geq 0)$ is a partition.
Let $v_\lambda$ be the  highest weight point  in $P(V_\lambda)$, i.e., the
point corresponding to the highest weight vector in $V_\lambda$.
This means 
 $bv_\lambda = v_\lambda$ for all $b\in B$,
where $B \subseteq GL_\n(C)$
is the Borel subgroup of lower triangular matrices. 
Recall  that the highest weight vector is unique.

Consider the orbit $G v_\lambda$ of $v_\lambda$.
Basic fact: 

\begin{prop} 
The orbit $Gv_\lambda$ is already closed in $P(V)$.
\end{prop} 

It can be shown that the stabilizer
$P_\lambda=G_{v_\lambda}$ is a
group of block lower triangular matrices, where the block lengths only
depend on $\lambda$ (see figure \ref{fig:block}). 
Such subgroups of $GL_n(\C)$ are called \emph{parabolic
subgroups}, and will be denoted by $P$.
Clearly   $Gv_\lambda \cong  G/P_\lambda=G/P$.

\begin{figure}
    \begin{center}
      \psfragscanon
      \psfrag{A1}{$A_1$}
      \psfrag{A2}{$A_2$}
      \psfrag{A3}{$A_3$}
      \psfrag{A4}{$A_4$}
      \psfrag{A5}{$A_5$}
      \psfrag{m1}{$m_1$}
      \psfrag{m2}{$m_2$}
      \psfrag{m3}{$m_3$}
      \psfrag{m4}{$m_4$}
      \psfrag{m5}{$m_5$}
      \epsfig{file=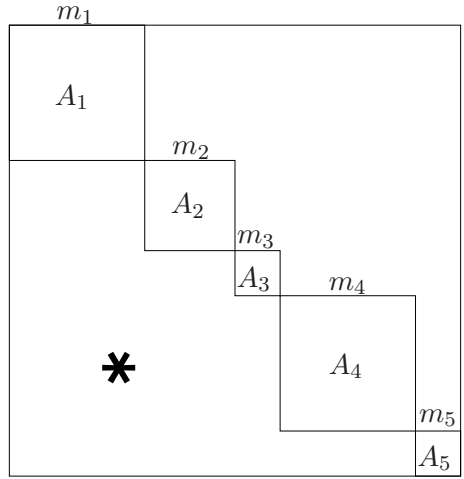, scale=.6}
      \caption{{\small The parabolic subgroup of block lower triangular matrices. The sizes $m_i$ only depend on $\lambda$.}}
      \label{fig:block}
    \end{center}
\end{figure}

\section{Grassmanians}
The simplest examples of $G/P$ are \emph{Grassmanians}.

\begin{defi}
Let $G=Gl_n(\C)$, and $V=\C^n$. 
The Grassmanian  $Gr_d^n$ is
the space  of $d$-dimensional subspaces (containing the
origin) of $V$.
\end{defi}

Examples:
\begin{enumerate}
\item
$Gr_1^2$ is the set of lines in $\C^2$ (see figure \ref{fig:lines}).
\item
More generally, $P(V)=Gr_1^n$.
\end{enumerate}

\begin{figure}
\begin{center}
\epsfig{file=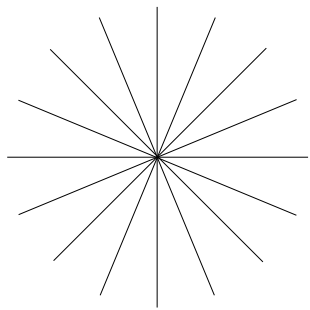, scale=0.2}
\caption{$Gr_1^2$ is the set of lines in $\C^2$.}
\label{fig:lines}
\end{center}
\end{figure}

\begin{prop} \label{pgrass} 
The Grassmanian $Gr_d^n$ is a projective variety (just like $P(V) = Gr_1^n$).
\end{prop}

It is easy to see that $Gr_d^n$ is closed (since the limit of a sequence 
of $d$-dimensional subspaces of $V$ is a $d$-dimensional subspace). 
Hence this follows from:

\begin{prop} \label{pgrass2}
Let $\lambda=(1,\ldots,1)$ be the partition of $d$, whose all
parts are $1$. Then
$Gr_d^n\cong Gv_\lambda \subseteq P(V_\lambda)$.\end{prop}

\proof 
For the given $\lambda$, 
$V_\lambda$ can be identified with
the  $d^{\textrm{th}}$ wedge product
$$\Lambda^d(V)=\textrm{span}\{(v_{i_1}\wedge \dots \wedge
v_{i_d})|i_1, \dots, i_d\ \textrm{are distinct}\}\subseteq
V\otimes\dots\otimes V\ (\textrm{d times}),$$ 
where 
\[(v_{i_1}\wedge \dots \wedge v_{i_d})=\frac{1}{d!}\sum_{\sigma\in
  S_d} sgn(\sigma)(v_{\sigma(i_1)}\otimes \dots \otimes
v_{\sigma(i_d)}).\]

Let $Z$ be a variable $d\times n$ matrix.
 Then  $\C[Z]$ is a $G$-module: given
$f\in \C[Z]$ and $\sigma \in GL_n(\C)$, we define the action of
$\sigma$ by 
\[(\sigma\cdot f)(Z)=f(Z\sigma).\]
Now $\Lambda^d(V)$, as a $G$-module, is isomorphic to the
span in $\C[Z]$ of all $d\times d$ minors of $Z$.

Let $A\in  Gr_d^n$ be  a $d$-dimensional subspace of $V$.
 Take any basis $\{v_1, \dots, v_d\}$ of $A$.
The point  $v_1\wedge \dots \wedge v_d \in \Lambda^d(V)$ depends 
only on the subspace $A$, and 
not on  the  basis, since
the  change of basis does not change the wedge product.
Let $Z_A$ be the $d\times n$ complex matrix
whose rows are the basis vectors  $v_1, \dots,
v_d$ of $A$. The  \emph{Plucker map} 
associates with $A$ the tuple  of all $d\times d$ minors $A_{j_1, \dots,
  j_d}$ of $Z_A$, where $A_{j_1, \dots, j_d}$ denotes the minor of
$Z_A$ formed by the columns $j_1,\dots, j_d$. This depends only
on $A$, and not on the choice of basis for $A$. 

The proposition  follows from: 
\begin{claim}
The Plucker map is a $G$-equivariant map from $Gr_d^n$ to $Gv_\lambda\subseteq P(V_\lambda)$ and $Gr_d^n\approx
Gv_\lambda\subseteq P(V_\lambda)$.
\end{claim}
\begin{proof}
Exercise. Hint: take the usual basis, and note that the highest weight point
 $v_\lambda$ corresponds to $v_1\wedge\dots\wedge v_d$.
\end{proof}

\chapter{The class varieties}
\begin{center} {\Large  Scribe: Hariharan Narayanan} \end{center} 

\noindent {\bf Goal:}   Associate class varieties with the complexity 
classes $\#P$ and $NC$ and reduce the $NC \not = P^{\#P}$ conjecture 
over $\C$ to a conjecture that the class variety for $\#P$ cannot be
embedded in the class variety for $NC$.

\noindent {\em reference:} \cite{GCT1}

The $NC \not = P^{\#P}$ conjecture over $\C$ says that the permanent
of an $n\times n$ complex matrix $X$ cannot be expressed as a 
determinant of an $m\times m$ complex matrix $Y$, $m=\poly(n)$, whose 
entries are (possibly nonhomogeneous) linear forms in the entries of $X$.
This obviously implies the $NC \not = P^{\#P}$ conjecture over $\Z$,
since multivariate
polynomials over $\C^n$ are determined by the values that they take
over the subset $\Z^n$.
The conjecture over $\Z$ 
is implied by the usual
$NC \not = P^{\#P}$ conjecture over a finite field $F_p$,
$p \not = 2$, and hence, has to be proved first anyway. 

For this reason, we  concentrate on the $NC \not = P^{\#P}$ conjecture
over $\C$ in this lecture. The goal is to reduce this conjecture to
a statement in geometric invariant theory.

\section{Class Varieties in GCT}
Towards that end, 
we  associate with the complexity classes $\#P$ and $NC$ certain
projective algebraic varieties, which we call {\em class varieties}.
For this, we need a few definitions.

Let $G =
GL_\ell(\C)$, $V$ a finite dimensional representation of $G$.
Let  $P(V)$ be the associated
projective space, which inherits the group action. Given a point $v \in P(V)$,
let  $\Delta_V[v] =
\overline{Gv} \subseteq P(V)$ be its orbit closure. Here
$\overline{Gv}$ is the closure of the orbit 
$Gv$ in the
complex topology on $P(v)$.
It is a projective $G$-variety; i.e., a projective variety
with the action of $G$.

All class varieties in GCT are orbit closures (or their slight 
generalizations), where
 $v \in P(V)$  corresponds to a complete function for the
class in question. The choice of the  complete function is crucial, since
it determines the algebraic geometry of $\Delta_V[v]$. 

We  now associate a class variety with $NC$.
Let  $g = det(Y)$, $Y$ an $m\times m$ variable matrix. This is a complete
function for  $NC$. 
Let $V = sym^m(Y)$ be the space of
homogeneous forms in the entries of $Y$ of degree $m$. It  is a
$G$-module, $G=GL_{m^2}(\C)$,  with the action of $\sigma \in G$ given by:
$$\sigma:f(Y) \longmapsto f(\sigma^{-1}Y).$$
Here $\sigma^{-1} Y$ is defined thinking of $Y$ as an $m^2$-vector.

Let $\Delta_V[g] =\Delta_V[g,m] = \overline{Gg}$, where we think of $g$ as an
element of $P(V)$. This is the class variety associated
with $NC$. If $g$ is a different function instead of $det(Y)$, the
algebraic geometry of $\Delta_V[g]$ would have been unmanageable.
The main point  is that the algebraic geometry of $\Delta_V[g]$  is
nice, because of the very special nature of the determinant function.

We next associate a class variety with $\#P$.
Let  $h = perm(X)$, $X$ an $n \times n$ variable matrix.
Let 
$W = sym^n(X)$. It is similarly  an $H$-module, $H=GL_{k}(\C)$, $k=n^2$.
Think of $h$ as an element of $P(W)$, and let 
$\Delta_W[h] = \overline{Hh}$ be its orbit closure. It is called the
 class variety associated with
$\#P$.

Now assume that $m>n$, and think of $X$ as a submatrix of $Y$, say the lower 
principal submatrix. Fix a variable entry $y$ of $Y$ outside of $X$. 
Define the map
$\phi: W \rightarrow V$ which takes $w(x) \in W$ to 
$y^{m-n} w(x) \in V$.
This induces a map from $P(V)$ to $P(W)$ which we
call $\phi$ as well. Let 
$\phi(h) = f \in P(V)$ and  $\Delta_V[f, m, n] = \overline{Gf}$ its 
orbit closure. It is called the 
extended class variety associated with $\#P$.

\begin{prop} [GCT 1]
\begin{enumerate}
\item If $h(X) \in W$ can be computed by a circuit (over $\C$) of depth
$\le  \log^c(n)$,  $c$  a constant, then $f = \phi(h) \in
\Delta_V[g,m]$, for  $m = O(2^{\log^c n})$.
\item Conversely if $f \in \Delta_V[g, m]$ for $m = 2^{\log^c n}$, then
$h(X)$ can be approximated infinitesimally closely by a circuit of
depth $\log^{2c} m$. That is, 
$\forall \epsilon > 0$, there exists a function  $\tilde{h}(X)$ that can be
computed by a circuit of depth $\le \log^{2 c} m$ such that
 $\|\tilde{h} - h\| <
\epsilon$ in the usual norm on $P(V)$.
\end{enumerate}
\end{prop}

If the permanent $h(X)$  can be
approximated infinitesimally closely by small depth circuits, then 
every function in $\#P$ can be approximated infinitesimally closely by small
depth circuits. This is not expected. Hence:

\begin{conj}[GCT 1]
Let $h(X) = perm(X)$, $X$ an $n \times n$ variable matrix.
Then $f = \phi(h)
\not\in \Delta_V[g; m]$ if $m = 2^{polylog(n)}$ and $n$ is sufficiently
large.
\end{conj}
This is equivalent to:

\begin{conj}[GCT 1]
The $G$-variety  $\Delta_V[f; m, n]$ cannot be embedded as a $G$-subvariety of
$\Delta_V[g,m]$, symbolically
$$\Delta_V[f; m, n] \not\hookrightarrow \Delta_V[g, m],$$
if $m =
2^{polylog(n)}$ and $n\rightarrow \infty$.
\end{conj}
This is the statement in geometric invariant theory (GIT) that we sought.

\chapter{Obstructions}
\begin{center} {\Large  Scribe: Paolo Codenotti} \end{center} 

\noindent {\bf Goal:} Define an obstruction to the embedding of 
the $\#P$-class variety in the $NC$-class-variety  and describe  why it should
exist. 

\noindent {\em References:} \cite{GCT1,GCT2} 

\subsection*{Recall} 

\begin{figure}
    \begin{center}
      \epsfig{file=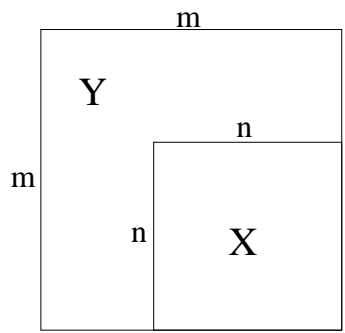, scale=.5}
      \caption{Here $Y$ is a generic $m$ by $m$ matrix, and
      $X$ is an $n$ by $n$ minor.}
      \label{fig:matr}
    \end{center}
\end{figure}

Let us first recall some definitions and results from the last
class. Let $Y$ be a generic $m\times m$  variable 
matrix, and $X$ an $n\times n$ minor of $Y$ (see figure \ref{fig:matr}). Let $g=\det(Y)$, $h=
\perm(X)$, $f=\phi(h)= y^{m-n} \perm(X)$, and 
$V=\Sym^m[Y]$  the set of homogeneous forms of degree $m$ in the entries of 
$Y$. Then $V$ is a $G$-module for
$G=GL(Y)=GL_l(\C)$,  $l=m^2$, with the action of $\sigma \in G$ given by
\[ \sigma: f(Y) \rightarrow f(\sigma^{-1} Y),\] 
where $Y$ is thought of as an $l$-vector, and 
 $P(V)$  a $G$-variety. Let 
$$\Delta_V[f; m, n]=\overline{Gf}\subseteq P(V),$$
and 
$$\Delta_V[g; m]=\overline{Gg}\subseteq P(V)$$
be the class varieties associated with $\#P$ and $NC$. 

\section{Obstructions}

\begin{conj}\cite{GCT1} \label{cembedd1}
There does not exist an embedding $\Delta_V[f; m, n] \hookrightarrow \Delta_V[g; m]$ with
$m=2^{\mbox{polylog}(n)}$, $n\rightarrow\infty$.
\end{conj}
This implies Valiant's conjecture that 
the permanent cannot be computed by circuits of polylog depth.
Now we discuss  how to go about proving the conjecture.

Suppose to the contrary,
\begin{equation}\label{equ:emb}
\Delta[f; m, n]\hookrightarrow \Delta_V[g; m].
\end{equation}

We  denote $\Delta_V[f; m, n]$ by $\Delta_V[f]$, and $\Delta_V[g; m]$ by $\Delta_V[g]$.
Let $R_V[g]$ be the homogeneous coordinate ring of $\Delta_V[g]$. The embedding (\ref{equ:emb}) implies
existence of a surjection:
\begin{equation}\label{equ:surj}
R_V[f]\twoheadleftarrow R_V[g]
\end{equation}
This is a basic fact from algebraic geometry. The reason is that $R_V[g]$ is the set of homogeneous polynomial
functions on the cone $C$ of $\Delta_V[g]$, and any such function
$\tau$ can be restricted to $\Delta_V[f]$ (see
figure \ref{figcone}). Conversely, any polynomial function on $\Delta_V[f]$ can be extended to a homogeneous
polynomial function on the cone $C$.

\begin{figure}
    \begin{center}
      \psfragscanon
      \psfrag{C}{$C$}
      \psfrag{T}{$\tau$}
      \psfrag{D}{$\Delta_V[f]$}
      \epsfig{file=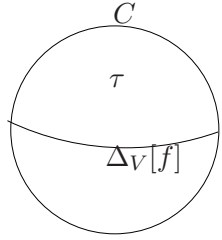, scale=.5}
      \caption{{\small $C$ denotes the cone of $\Delta_V[g]$.}}
      \label{figcone}
    \end{center}
\end{figure}

Let $R_V[f]_d$ and $R_V[g]_d$ be the degree $d$ components of $R_V[f]$ and
$R_V[g]$. These are  $G$-modules since $\Delta_V[f]$ and $\Delta_V[g]$ are
$G$-varieties. The surjection
(\ref{equ:surj}) is degree preserving. So  there is  a surjection
\begin{equation}\label{equ:surd}
R_V[f]_d\twoheadleftarrow R_V[g]_d
\end{equation}
for every $d$.
Since $G$ is reductive, both $R_V[f]_d$ and $R_V[g]_d$ are direct sums
of irreducible $G$-modules. Hence 
the surjection (\ref{equ:surd}) implies that $R_V[f]_d$ can be
embedded as  a $G$ submodule of $R_V[g]_d$.

\begin{defi}
We say that a Weyl-module $S=V_\lambda(G)$ is an \emph{obstruction} for the embedding (\ref{equ:emb}) (or,
equivalently, for the pair $(f, g)$) if $V_\lambda(G)$ occurs in $R_V[f;m, n]_d$, but not in $R_V[g; m]_d$, for
some $d$. Here  occurs means
 the multiplicity of $V_\lambda(G)$ in the decomposition of $R_V[f;m,
n]_d$ is nonzero.
\end{defi}

If an obstruction exists for given $m, n$, then the embedding (\ref{equ:emb}) does not exist.

\begin{conj}[GCT2]
An obstruction exists for the pair $(f, g)$
for all large enough $n$ if $m=2^{\textrm{polylog}(n)}$.
\end{conj}
This implies Conjecture~\ref{cembedd1}.
In essence, this   turns 
a nonexistence problem (of polylog depth circuit for the permanent)
into an existence problem (of an obstruction).

If  we replace   the  determinant here
by any other complete function in $NC$,
an obstruction need not exist. Because,
as we shall see in the next lecture,  the existence of
an obstruction crucially depends on the exceptional nature of the 
class variety constructed from the determinant.
The main  goals of GCT in this context are:
\begin{enumerate}
\item  understand the exceptional nature of the
class varieties for $NC$ and $\#P$, and 
\item  use it to prove the existence of obstructions.
\end{enumerate}

\subsection{Why are the class  varieties  exceptional?}
We now elaborate on the exceptional nature of the class varieties.
Its significance for the existence of obstructions will be discussed in
the next lecture.

Let $V$ be a $G$-module, $G=GL_n(\C)$. Let $P(V)$ be a projective variety over $V$. Let $v\in P(V)$, and recall
$\Delta_V[v]=\overline{Gv}$. Let $H=G_v$ be the stabilizer of $v$, that is, $G_v=\{\sigma\in G|\sigma v = v\}$.
\begin{defi}
We say that $v$ is \emph{characterized by its stabilizer} $H=G_v$ if $v$ is the only point in $P(V)$ such that
$hv=v,\ \forall h\in H$.
\end{defi}
If $v$ is characterized by its stabilizer, then $\Delta_V[v]$ is completely determined by the group triple
$H\hookrightarrow G\hookrightarrow K=GL(V)$.

\begin{defi}
The orbit closure $\Delta_V[v]$, when $v$ is characterized by its stabilizer, is called a \emph{group-theoretic
variety}.
\end{defi}

\begin{prop} \cite{GCT1}\label{prop:char}
\begin{enumerate}
\item The determinant $g=\det(Y)\in P(V)$
 is characterized by its stabilizer. 
Therefore $\Delta_V[g]$ is group theoretic.
\item The permanent $h=\perm(X)\in P(W)$, where $W=\Sym^n(X)$,
 is also characterized by its stabilizer. Therefore
$\Delta_W[h]$ is also group theoretic.
\item Finally, $f=\phi(h)\in P(V)$ is also characterized by its stabilizer. 
Hence $\Delta_V[f]$ is also  group theoretic.
\end{enumerate}
\end{prop}

\begin{proof}
\noindent (1) 
It is a fact in classical representation theory that the stabilizer of $\det(Y)$ in $G=GL(Y)=GL_{m^2}(\C)$ is
the subgroup $G_{\det}$ that consists of linear transformations of the form $Y\rightarrow AY^*B$, where $Y^*=Y$
or $Y^t$, for any $A, B \in GL_m(\C)$. It is clear that linear transformation of this form stabilize the
determinant since:
\begin{enumerate}
\item $\det(AYB) = det(A)det(B)det(Y)= c\det(Y)$, where $c=\det(A)\det(B)$. Note that the constant $c$ doesn't matter
because we get the same point in the projective space.
\item $\det(Y^*)=\det(Y)$.
\end{enumerate}
It is a basic fact in classical invariant theory that
$\det(Y)$ is the only point in $P(V)$ stabilized by $G_{\det}$.
Furthermore, the stabilizer $G_{\det}$ is reductive, since
its  connected part is $(G_{\det})_\circ\approx
GL_m\times GL_m$ with the  natural embedding
\[(G_{\det})_\circ =  GL_m\times GL_m\hookrightarrow 
GL(\C^m \otimes \C^m)=GL_{m^2}(\C)=G.\]

\noindent (2) 
The stabilizer of $\perm(x)$ is the subgroup $G_{\perm} \subseteq GL(X)=GL_{n^2}(\C)$ generated by linear
transformations of the form $X\rightarrow \lambda X^* \mu$,
where $X^*=X or X^t$, and $\lambda$ and $\mu$ are diagonal
(which change the permanent by a constant factor) or permutation matrices (which do not change the permanent).

Finally, the discrete component of $G_{\perm}$ is isomorphic to $S_2\rtimes S_n \times S_n$, where $\rtimes$
denotes semidirect product. The continuous part is $(\C^*)^n \times (\C^*)^n$.
So $G_{\perm}$ is reductive.

\noindent (3)  Similar.

\end{proof}

The main significance of this proposition is the following. Because $\Delta_V[g], \Delta_V[f]$, and
$\Delta_W[h]$ are group theoretic, the algebraic geometric problems concerning these varieties can be ``reduced"
to problems in the theory of quantum groups. So the plan is:
\begin{enumerate}
\item Use the theory of quantum groups to understand the structure of the group triple associated with the
algebraic variety.
\item Translate this understanding to the structure of the algebraic variety.
\item Use this to show the existence of obstructions.
\end{enumerate}

\chapter{Group theoretic varieties}
\begin{center} {\Large  Scribe: Joshua A. Grochow} \end{center}

\noindent {\bf Goal:} In this lecture we continue our discussion of group-theoretic varieties. We describe why obstructions should exist, and 
why the exceptional group-theoretic nature of the class varieties is
crucial for this existence.

\subsection*{Recall} 
Let $G=GL_n(\C)$,  $V$ a $G$-module, and  $\PP(V)$  the associated 
 projective space.  Let $v \in \PP(V)$ be a point characterized by its stabilizer $H=G_v \subset G$. In other words, $v$ is the only point in $\PP(V)$ stabilized by $H$.  Then $\Delta_V[v] = \overline{Gv}$ is called a \emph{group-theoretic variety} because it is completely determined by the group triple $$H \hookrightarrow G \hookrightarrow GL(V).$$

The simplest example of a group-theoretic variety is a variety of the 
form $G/P$ that we described in the earlier lecture.
Let $V=V_\lambda(G)$ be a Weyl module of $G$ and
 $v_\lambda \in \PP(V)$  the highest weight point (recall: the unique point
stabilized by the Borel subgroup $B \subset G$ of lower triangular matrices).  Then the stabilizer of $v_\lambda$ consists of block-upper triangular matrices, where the block sizes are determined by $\lambda$:
$$
P_\lambda := G_{v_\lambda} = \left(\begin{array}{ccc|cc|ccc}
* & * & * & * & * & * & * & * \\
* & * & * & * & * & * & * & * \\
* & * & * & * & * & * & * & * \\ \hline
0 & 0 & 0 & * & * & * & * & * \\
0 & 0 & 0 & * & * & * & * & * \\ \hline
0 & 0 & 0 & 0 & 0 & * & * & * \\
0 & 0 & 0 & 0 & 0 & * & * & * \\
0 & 0 & 0 & 0 & 0 & * & * & *
\end{array} \right)
$$

 The orbit  $\Delta_V[v_\lambda]=Gv_\lambda \cong G/P_\lambda$ is
a group-theoretic variety determined entirely by the  triple
$$
P_\lambda=G_{v_\lambda} \hookrightarrow G \hookrightarrow K=GL(V).
$$

The group-theoretic varieties of main interest in GCT are the
class varieties associated with the various complexity classes.

\section{Representation theoretic data}
The {\bf main principle} guiding GCT 
is that the algebraic geometry of a group-theoretic variety
 ought to be completely determined by the representation theory of the corresponding group triple.  This is a natural extension of work already pursued in mathematics by Deligne and Milne on Tannakien categories \cite{deligneMilne}, showing that an algebraic group is completely determined by its representation theory.
So the  goal is to associate to a group-theoretic variety some representation-theoretic data that will analogously 
 capture the information in the variety completely. 
We shall now illustrate this for the class variety for $NC$.
First a few definitions.

Let $v \in \PP(V)$ be the point as above
characterized by its stabilizer $G_v$. This means the line $\C v \subseteq V$
corresponding to $v$ is a one-dimensional representation of $G_v$.
  Thus $(\C v)^{\otimes d}$ is a one-dimensional degree $d$ representation, i.e. the representation  $\rho: G \to GL(\C v) \cong \C^*$
is  polynomial of degree $d$ in the entries of the matrix of an element 
 in $G$.  Recall that $\C[V]$ is the coordinate ring of $V$, and $\C[V]_d$ is its degree $d$ homogeneous component, so $(\C v)^{\otimes d} \subseteq \C[V]_d$.  

To each $v \in \PP(V)$ that is characterized by its stabilizer, we associate 
a representation-theoretic data, which is 
the set of $G$-modules 
$$\Pi_v = \bigcup_d \Pi_v(d),$$ where $\Pi_v(d)$ is the set of all irreducible $G$-submodules $S$ of $\C[V]_d$ whose duals $S^*$ do \emph{not} contain a $G_v$-submodule isomorphic to $(\C v)^{\otimes d *}$ (the dual of $(\C v)^{\otimes d}$).  The following proposition  elucidates the importance of this data:

\begin{prop} \cite{GCT2} \label{pgct2b}  $\Pi_v \subseteq I_V[v]$ (where $I_V[v]$ is the ideal of the projective variety $\Delta_V[v] \subseteq \PP(V)$). \end{prop}

\begin{proof} Fix $S \in \Pi_v(d)$.  Suppose, for the sake of contradiction, that $S \nsubseteq I_V[v]$.  Since $S \subseteq \C[V]$, $S$ consists of ``functions'' on the variety $\PP(V)$ (actually homogeneous polynomials on $V$).  The coordinate ring of $\Delta_V[v]$ is $\C[V]/I_V[v]$, and since $S \nsubseteq I_V[v]$, $S$ must not vanish identically on $\Delta_V[v]$.  Since the orbit $Gv$ is dense in $\Delta_V[v]$, $S$ must not vanish identically on this single orbit $Gv$.  Since $S$ is a $G$-module, if $S$ were to vanish identically on the line $\C v$, then it would vanish on the entire orbit $Gv$, so $S$ does not vanish identically on $\C v$.

Now $S$ consists of functions of degree $d$.
Restrict them  to the line $\C v$. The dual of this restriction gives an 
injection of $(\C v)^{\otimes d *}$ as a $G_v$-submodule of $S^*$,
contradicting the definition of $\Pi_v(d)$.  \end{proof}

\section{The second fundamental theorem}
We  now ask essentially the reverse question: when does the
representation theoretic data  $\Pi_v$ generate the ideal $I_V[v]$? 
 For if $\Pi_v$ generates $I_V[v]$, then $\Pi_v$ completely captures the coordinate ring $\C[V]/I_V[v]$, and hence the variety $\Delta_V[v]$.

\begin{thm} [Second fundamental theorem of invariant theory for $G/P$]
The $G$-modules in
$\Pi_{v_\lambda}(2)$ generate the ideal $I_V[v_\lambda]$ of the orbit
 $Gv_\lambda \cong G/P_\lambda$, when $V=V_\lambda(G)$.  \end{thm}

This theorem justifies the main principle for $G/P$,
so we can hope that  similar results
hold for  the class varieties in GCT  
(though not always exactly in the same form).

Now, let $\Delta_V[g]$ be the class variety for \cc{NC} (in other words, take $g=\det(Y)$ for a matrix $Y$ of indeterminates).  Based on the main principle, we have the following conjecture, which essentially generalizes the second fundamental theorem of invariant theory for $G/P$ to the class variety for \cc{NC}:

\begin{conj} [GCT 2] \label{cgct2} $\Delta_V[g] = X(\Pi_g)$ where $X(\Pi_g)$ is the zero-set of all  forms in the $G$-modules contained in $\Pi_g$. \end{conj}

\begin{thm} [GCT 2] A weaker version of the above conjecture holds. 
Specifically, assuming that the Kronecker coefficients satisfy a certain
separation property, there exists a $G$-invariant (Zariski) 
open neighbourhood $U \subseteq P(V)$ of 
the orbit $G g$  such that $X(\Pi_g) \cap U = \Delta_V[g] \cap U$. \end{thm}

There is a notion of algebro-geometric complexity called \emph{Luna-Vust complexity} which  quantifies the gap between $G/P$ and class varieties.  The Luna-Vust complexity of $G/P$  is 0. The Luna-Vust complexity of the \cc{NC} 
class variety is $\Omega(\dim(Y))$.
This is analogous to the difference between circuits of constant depth and circuits of superpolynomial depth.  This is why the previous conjecture and theorem turn out to be far harder than the corresponding facts for $G/P$.  

\section{Why should obstructions exist?}
The following proposition explains why obstructions should exist to separate \cc{NC} from \cc{$P^{\# P}$}.  

\begin{prop} [GCT 2]
Let $g=\det(Y)$, $h=\perm(X)$, $f=\phi(h)$, $n=\dim(X)$,
$m=\dim(Y)$.  If Conjecture~\ref{cgct2}  holds and
the permanent cannot be approximated arbitrarily closely by circuits of poly-logarithmic depth (hardness assumption), 
then an obstruction for the pair $(f,g)$ exists for all
large enough $n$, when  $m=2^{\log^c n}$ for some constant $c$.
Hence, under these conditions,
$\cc{NC} \neq \cc{$P^{\# P}$}$ over $\C$.  \end{prop}

This proposition may seem a bit circular at first, since it relies on a hardness assumption.  But we do not plan to prove the existence of obstructions by proving the assumptions of this proposition.  Rather, this proposition should be taken as evidence that obstructions exist (since we expect the hardness assumption
therein  to hold, given that the  permanent is \cc{\# P}-complete), and
we will develop other methods to prove their existence.

\begin{proof} The hardness assumption implies that $f \notin \Delta_V[g]$ if $m=2^{\log^c n}$ [GCT 1].

Conjecture~\ref{cgct2}  says that $X(\Pi_g)=\Delta_V[g]$.  So there exists an irreducible $G$-module $S \in \Pi_g$ such that $S$ does not vanish on $f$. 
So $S$ occurs in $R_V[f]$ as a $G$-submodule.  

On the other hand, since $S \in \Pi_g$, $S \subseteq I_V[g]$ by 
Proposition~\ref{pgct2b}. So $S$ does not occur 
in $R_V[g]=\C[V]/I_V[g]$.  Thus  $S$ is not a $G$-submodule of $R_V[g]$, but it is a $G$-submodule of $R_V[f]$, i.e., $S$ is an obstruction.  \end{proof}

\chapter{The flip}
\begin{center} {\Large  Scribe: Hariharan Narayanan} \end{center} 

\noindent {\bf Goal:} Describe the basic principle of GCT, called the
flip, in the context of the $NC$ vs. $P^{\#P}$ problem over $\C$.
\\ \\

\noindent {\em references:} \cite{GCTflip1,GCT1,GCT2,GCT6}

\subsection*{Recall}
As in the previous lectures,
let  $g = det(Y) \in P(V)$, $Y$ an $m \times m$ variable matrix, $G =
GL_{m^2}(\C)$,  and $\Delta_V(g) = \Delta_V[g;m]=\overline{G}g \subseteq P(V)$ the 
class variety for NC. Let
$h = perm(X)$, $X$ an $n \times n$ variable matrix,
$f = \phi(h)=y^{m-n} h \in P(V)$, and 
$\Delta_V(f) = \Delta_V[f; m, n] = \overline{G}f \subseteq P(V)$ the
class variety for $P^{\#P}$. 
Let $R_V[f;m,n]$ denote the 
homogeneous coordinate ring of $\Delta_V[f;m,n]$, $R_V[g;m]$ the 
homogeneous coordinate ring of $\Delta_V[g;m]$, and 
$R_V[f;m,n]_d$ and $R_V[g;m]_d$  their degree $d$-components. 
A Weyl module  $S = V_\l(G)$ of $G$ is an obstruction of degree
$d$ for the pair $(f,g)$  if $V_\l$ occurs in $R_V[f; m, n]_d$ but not
$R_V[g; m]_d$.

\begin{conj} \cite{GCT2} \label{cobs}
An obstruction (of degree polynomial in $m$)  exists if $m =
2^{\text{polylog}(n)}$ as $n\rightarrow \infty$.
\end{conj}
This implies $NC \not = P^{\#P}$ over $\C$. 

\section{The flip} 
In this lecture we describe an approach to prove 
the existence of such obstructions. It is  based on 
the following complexity theoretic positivity hypothesis:

\vspace{.1in}

\noindent {\bf PHflip} \cite{GCTflip1}:

\begin{enumerate}
\item Given $n, m$ and $d$, whether an obstruction of degree $d$
for $m$ and $n$ exists can be decided in $poly(n, m, \bitlength{d})$ time,
 and if
it exists, the label $\l$ of such an obstruction can be constructed
in $poly(n, m, \bitlength{d})$ time. Here $\bitlength{d}$ denotes the
bitlength of $d$.
\item \begin{enumerate} \item 
 Whether $V_\l$ occurs in $R_V[f; m, n]_d$ can be decided in
$poly(n, m, \bitlength{d}, \bitlength{\l})$ time.
\item Whether $V_\l$ occurs in $R_V[g; m]_d$ can be decided in
$poly(n, m, \bitlength{d}, \bitlength{\l})$ time.
\end{enumerate}
\end{enumerate}

This suggests the following approach for proving Conjecture~\ref{cobs}:

\begin{enumerate}
\item Find polynomial time algorithms sought  in PHflip-2 for the
basic decision problems (a) and (b) therein.
\item Using these find   a polynomial time algorithm sought  in PHflip-1 
for deciding if an obstruction exists.
\item Transform (the techniques underlying)
this ``easy'' (polynomial time)  algorithm for 
deciding if an obstruction exists for given $n$ and $m$ into an
 ``easy" (i.e., feasible)  proof of
existence of an obstruction for every $n\rightarrow \infty$ 
when  $d$ is large enough and $m=2^{\mbox{polylog}(n)}$.
\end{enumerate}

The first step here is the crux of the matter.
The main results of \cite{GCT6} say that the polynomial time
algorithms for the basic decision problems as sought in PHflip-2 
indeed exist assuming natural analogues of PH1 and SH (PH2) that
we have seen earlier in the context of the plethysm problem.
To state them, 
we  need some definitions.

Let $S_d^\l[f]=S_d^\l[f;m,n]$
be the multiplicity of $V_\l=V_\l(G)$ in $R_V[f; m, n]$. The
stretching function $\tilde{S}_d[f] = \tilde{S}_d^\l[f; m, n]$  is
defined  by
$$\tilde{S}_d^\l[f](k) := S_{kd}^{k\l}[f].$$
The stretching function  for $g$, $\tilde{S}^\lambda_d[g]=
\tilde{S}^\lambda_d[g;m]$, is defined
analogously.

The main mathematical result of \cite{GCT6} is:
\begin{thm}\cite{GCT6}
The stretching functions $\tilde{S}_d^\l[g]$ and $\tilde{S}_d^\l[f]$
are quasipolynomials assuming that the singularities of $\Delta_V[f; m,
n]$ and $\Delta_V[g; m]$ are rational.
\end{thm}
Here rational means ``nice''; we shall not worry about the exact 
definition.

The main complexity-theoretic result is:
\begin{thm}\cite{GCT6} \label{tcomplexity}
Assuming the following mathematical 
positivity hypothesis $PH1$ and the saturation hypothesis $SH$ (or the 
stronger positivity hypothesis  $PH2$), PHflip-2 holds.
\end{thm}

\vspace{.1in}

\noindent {\bf PH1:} 
There exists a polytope $P = P_d^\l[f]$ such that

\begin{enumerate}
\item The Ehrhart quasi-polynomial of $P$, $f_P(k)$,  is
$\tilde{S}_d^\l[f](k)$.
\item $dim(P) = poly(n, m, \bitlength{d})$.
\item Membership in $P$ can be answered in polynomial time.
\item There is a polynomial time separation oracle \cite{lovasz} for $P$.
\end{enumerate}

Similarly, there exists a polytope $Q = Q_d^\l[g]$ such that

\begin{enumerate}
\item The Ehrhart quasi-polynomial of $Q$, $f_Q(k)$,  is
$\tilde{S}_d^\l[g](k)$.
\item $dim(Q) = poly( m, \bitlength{d})$.
\item Membership in $Q$ can be answered in polynomial time.
\item There is a polynomial time separation oracle for $Q$.
\end{enumerate}

\vspace{.1in}

\noindent {\bf PH2:} 
The quasi-polynomials  $\tilde{S}_d^\l[g]$ and $\tilde{S}_d^\l[f]$ are
positive. 

\vspace{.1in}

This implies:

\vspace{.1in}

\noindent {\bf SH:}
The quasi-polynomials 
$\tilde{S}_d^\l[g]$ and $\tilde{S}_d^\l[f]$ are saturated.

\vspace{.1in}

PH1 and SH imply that the  decision problems in PHflip-2 can be
transformed into saturated positive integer programming problems.
Hence Theorem~\ref{tcomplexity} follows from the polynomial time algorithm 
for saturated linear programming that we described in an earlier class.

The decision problems in PHflip-2 are  ``hyped" up versions of the
plethysm problem discussed earlier.
The article \cite{GCT6} provides evidence for $PH1$ and $PH2$ 
for the plethysm problem. This constitutes the main evidence for
PH1 and PH2 for the class varieties in view of their group-theoretic
nature; cf. \cite{GCTflip1}.

The following problem is important in the context of PHflip-2:

\begin{problem}
Understand the $G$-module structure of the homogeneous coordinate rings
$R_V[f]_d$ and  $R_V[g]_d$.
\end{problem} 

This is an instance of the following abstract:

\begin{problem} Let $X$ be a projective group-theoretic $G$-variety.
 Let $R = \bigoplus_{d=0}^\infty
R_d$ be its homogeneous coordinate ring. Understand the $G$-module
structure of $R_d$.
\end{problem} 

The simplest  group-theoretic variety is $G/P$.
For it,  a solution to this abstract 
problem is given by the following results:

\begin{enumerate}
\item The Borel-Weil theorem.
\item The Second Fundamental theorem of invariant theory [SFT].
\end{enumerate}

These will be covered in the next class for the simplest case of 
$G/P$, the Grassmanian.

\chapter{The Grassmanian}
\begin{center} {\Large  Scribe: Hariharan Narayanan} \end{center} 
\noindent {\bf Goal:} The Borel-Weil and the second fundamental theorem
of invariant theory for the Grassmanian.

\noindent {\em Reference:} \cite{YT}

\subsection*{Recall}
Let $V = V_\l(G)$ be a Weyl module of $G= GL_n(\C)$ and $v_\l \in P(V)$  the
point corresponding to its  highest weight vector.
The orbit  $\Delta_V[v_\l] := G v_\l$, which is already closed, is 
of the form $G/P$, where $P$ is the parabolic stabilizer of $v_\l$.
 When  $\l$ is  a single column, it is called the 
{\em Grassmannian}.

An alternative description of the Grassmanian is as follows.
Assume that $\lambda$ is a single column of length $d$.
Let $Z$ be a
$d \times n$ matrix of variables $z_{ij}$.
Then  $V=V_\l(G)$ can be identified with  the span of
$d \times d$ minors of $Z$ with the action of $\sigma \in G$ given by:
$$\sigma: f(z) \mapsto f(z\sigma).$$

Let $Gr_d^n$ be the space of all d-dimensional subspaces of  $\C^n$. Let $W$
be a $d$-dimensional subspace of  $\C^n$. Let $B=B(W)$ be a basis of $W$.
Construct the $d \times n$ matrix $z_B$, whose rows are vectors in
$B$. Consider 
the  Pl\"{u}cker map from $Gr_d^n$ to $P(V)$ which maps any $W \in Gr_d^n$ 
to the tuple of $d \times
d$ minors of  $Z_B$. Here  the choice of $B=B(W)$ does not matter, since any
choice gives the same point in $P(V)$.
Then  the image of $Gr_d^n$ is precisely the 
Grassmanian $Gv_\lambda \subseteq P(V)$. 

\section{The second fundamental theorem}
Now we ask:
\begin{question}
 What is the ideal of $Gr_d^n \approx Gv_\l \subseteq
P(V)$? 
\end{question} 

The homogeneous coordinate ring of $P(V)$ is $\C[V]$.
We want an explicit set of generators of this ideal in
$\C[V]$.  This is 
given by the second fundamental theorem of invariant theory,
which we describe next.

The coordinates of $P(V)$ are in one-to-one correspondence with
 the   $d \times d$
minors of the matrix $Z$. Let each minor of $Z$ be indexed by its
columns. Thus for $1 \leq i_1 < \dots < i_d \leq n$, $Z_{i_1, \dots,
i_d}$ is a coordinate of $P(V)$ corresponding to 
the minor of $Z$ formed by
the columns $i_1,i_2,\ldots$. Let  $\Lambda(n, d)$ be  the set of ordered
$d$-tuples of $\{1, \dots, n\}$. The tuple $[i_1,\ldots,i_d]$ in
 this set 
will  be identified with  the coordinate $Z_{i_1,
\dots, i_d}$ of $P(V)$. There is a bijection between the elements of
$\Lambda(n, d)$ and of $\Lambda(n, n-d)$ obtained by associating
complementary sets:
$$\Lambda(n, d) \ni \l \leftrightsquigarrow \l^* \in \Lambda(n, n-d).$$
We define $sgn(\l, \l^*)$ to be the sign of the permutation that
takes $[1, \dots, n]$ to $[\l_1, \dots, \l_d, \l^*_1, \dots,
\l^*_{n-d}]$. 

 Given $s \in \{1. \dots, d\}$, $\alpha \in
\Lambda(n, s-1)$, $\beta \in \Lambda(n, d+1)$, and $\gamma \in
\Lambda(n, d-s)$, we now define the 
Van der Waerden Syzygy $[[\alpha,\beta,\gamma]]$, which 
is an element of the degree two component  $\C[V]_2$ 
of $\C[V]$, as follows:

\[
\begin{array}{l} 
[[\alpha,\beta,\gamma]]= \\
\sum_{\tau \in \Lambda(d+1, s)} sgn(\tau, \tau^*)[\alpha_1, \dots,
\alpha_{s-1}, \beta_{\tau^*_1}, \dots, \beta_{\tau^*_{d+1-s}}] 
[\beta_{\tau_1}, \dots, \beta_{\tau_s}, \gamma_1, \dots, \gamma_{d-s}].
% \in \C[V]_2
\end{array}
\]

It is easy to show that this syzygy  vanishes on 
the Grassmanian $Gr_d^n$: 
because it is an alternating $(d+1)$-multilinear-form,
and hence has to vanish on any $d$-dimensional space $W \in Gr_d^n$.
Thus it belongs to the ideal of the Grassmanian. Moreover:

\begin{thm}  [Second fundamental theorem] \label{tsecondthm1}
The ideal of the Grassmanian $Gr_d^n$ is generated by 
the Van-der-Waerden syzygies.
\end{thm}

An alternative formulation of this result is as follows.
Let $P_\l \subseteq G$ be the stabilizer of $v_\l$.
Let $\Pi_{v_\l}(2)$ be the set of
irreducible $G$-submodules of $\C[V]_2$ whose duals do not contain a
$P_\l$-submodule isomorphic to $\C v_\l^{\otimes 2 *}$ (the dual of 
$\C v_\l^{\otimes 2}$). 
Here $\C v_\l$ denotes the line in $P(V)$ corresponding to $v_\lambda$, which
is a one-dimensional representation of $P_\l$ since it stabilizes 
$v_\l \in P(V)$.
It can be
shown that the span of the $G$-modules in $\Pi_{v_\l}(2) $ is equal to
the span of the Van-der-Waerden syzygies. Hence, 
Theorem~\ref{tsecondthm1} is equivalent to:

\begin{thm}
[Second Fundamental Theorem(SFT)] The $G$-modules in
$\Pi_{v_\l}(2)$ generate the ideal of $Gr_d^n$.
\end{thm}
This formulation of SFT for the Grassmanian looks very similar
to the  generalized conjectural SFT for the $NC$-class variety described in 
the earlier class. This indicates that the 
class varieties in GCT are ``qualitatively similar" to $G/P$.

\section{The Borel-Weil theorem}
We now describe the $G$-module structure of the homogeneous 
coordinate ring $R$ of the
Grassmannian $Gv_\l \subseteq P(V)$, where  $\l$
is a single column of height $d$. The goal is to give 
an explicit basis for $R$.
Let $R_s$ be the degree $s$
component of $R$. Corresponding to any numbering $T$ of the shape
$s\l$, which is a $d \times s$ rectangle, whose columns have
strictly increasing elements top to bottom, we have a monomial $m_T
 = \prod_c Z_c \in \C[V]_s$, were $Z_c$ is the coordinate of $P(V)$
indexed by the $d$-tuple $c$, and $c$ ranges over the $s$ columns of $T$.
We say that $m_T$ is (semi)-standard if the rows of $T$ are nondecreasing, when
read left to right. It is called nonstandard otherwise.

\begin{lemma}
[Straightening Lemma] Each non-standard $m_T$ can be straightened to
a normal form, as a linear combination of standard monomials,
by using  Van der Waerden Syzygies as straightening relations (rewriting
rules).
\end{lemma}

For any numbering $T$ as above, express $m_T$ in a normal form as per the
lemma:
$$m_T = \sum_{\text{(Semi)-Standard
Tableau } S} \alpha(S, T),m_S$$
where $\alpha(S,T) \in \C$.

\begin{thm}
[Borel-Weil Theorem for Grassmannians] Standard monomials $\{m_T\}$
form a basis of $R_s$, where $T$ ranges over all {\it semi-standard}
tableaux of rectangular shape $s\l$.
Hence, $R_s \cong  V_{s\l}^*$, the dual of the Weyl module  $V_{s \l}$.
\end{thm}

This  gives
the $G$-module structure of $R$ completely. It follows that 
the  problem of deciding if 
$V_\beta(G)$ occurs in $R_s$ can be solved  in polynomial time: this
is so if and only if $(s\l)^* = \beta$, where $(s \l)^*$ denotes the
dual partition, whose description is left as an exercise. 

The second fundamental theorem as well as the Borel-Weil theorem
easily follow from the straightening lemma and linear 
independence of the standard monomials (as functions on the Grassmanian). 

\chapter{Quantum group: basic definitions}
\begin{center} {\Large  Scribe: Paolo Codenotti} \end{center}

\noindent {\bf Goal:} 
The basic plan to implement the flip in \cite{GCT6} is to prove 
PH1 and SH via the theory of quantum groups.
We  introduce the basic 
concepts in this theory in this and the next two lectures,
and  briefly show their relevance  in the context of PH1 in the final lecture.

\noindent {\em Reference:} \cite{KS}

\section{Hopf Algebras}

Let $G$ be a group, and $K[G]$ the ring of  functions on $G$ with values in the field $K$, which will be $\C$ in our applications. The group
$G$ is defined by the following operations:

\begin{itemize}
\item multiplication: $G\times G\rightarrow G$,
\item identity $e$: $e\rightarrow G$,
\item inverse: $G\rightarrow G$.
\end{itemize}

In order for $G$ to be a group, the following properties have to hold:
\begin{itemize}
\item $eg=ge=g$,
\item $g_1(g_2g_3)=(g_1g_2)g_3$,
\item $g^{-1}g=gg^{-1}=e$.
\end{itemize}

We now want to translate these properties to properties of $K[G]$. This should be possible since $K[G]$ contains
all the information that $G$ has. In other words, we want to translate the notion of a group in terms of $K[G]$.
This translate is called a Hopf algebra. Thus if $G$ is a group,  $K[G]$ is a Hopf algebra. Let
us first define the dual operations.

\begin{itemize}
\item Multiplication is a map:
\[\cdot:G\times G\rightarrow G.\]
So co-multiplication $\Delta$ will be a map as follows:
\[K[G\times G]=K[G]\otimes K[G]\leftarrow K[G].\]
We want $\Delta$ to be the pullback of multiplication. So 
for a given $f \in K[G]$ we define $\Delta(f) \in K[G] \otimes K[G]$ 
by: \[\Delta(f)(g_1, g_2)=f(g_1g_2).\]
Pictorially:
\[\begin{CD}
G\times G @>\cdot>> G\\
@V\Delta(f)VV @VVfV\\
k @= k
\end{CD}\]

\item
The unit is a map:
\[e\rightarrow G.\]
Therefore we want the co-unit $\epsilon$ to be a map:
\[K\underleftarrow{\epsilon} K[G],\]
defined by: for $f\in K[G]$, $\epsilon(f)=f(e)$.

\item
Inverse is a map:
\[(\ )^{-1}: G\rightarrow G.\]
We want the dual antipode $S$ to be the map:
\[K[G]\leftarrow K[G]\]
defined by: for $f\in K[G]$, $S(f)(g)=f(g^{-1})$.
\end{itemize}

The following are the abstract axioms satisfied by $\Delta, \epsilon$ and
$S$.
\begin{enumerate}
\item
$\Delta$ and $\epsilon$ are algebra homomorphisms.
\[\Delta:K[G]\rightarrow K[G]\otimes K[G]\]
\[\epsilon:K[G]\rightarrow K.\]
\item
co-associativity: Associativity is defined so that the following diagram commutes:
\[\begin{CD}
G\times G @. \times @. G @= G @.\times @. G\times G\\
@V\cdot VV  @. @V\id VV      @VV\id V @. @VV\cdot V\\
G @. \times @. G @. G @. \times @. G\\
@. @V\cdot VV @. @. @VV\cdot V @.\\
@. G @.@=@. G
\end{CD}\]

Similarly, we define co-associativity so that the following dual diagram commutes:
\[\begin{CD}
K[G] \otimes K[G] @. \otimes @. K[G] @= K[G] @.\otimes @. K[G]\otimes K[G]\\
@A\Delta AA  @. @A\id AA      @AA\id A @. @AA\Delta A\\
K[G] @. \otimes @. K[G] @. K[G] @. \otimes @. K[G]\\
@. @A\Delta AA @. @. @AA\Delta A @.\\
@. K[G] @.@=@. K[G]
\end{CD}\]

Therefore co-associativity says: \[(\Delta\otimes\id)\circ \Delta = (\id\otimes \Delta)\circ\Delta.\]
\item
The property $ge=g$ is defined so that the following diagram commutes:
\[\begin{CD}
e @. \times @. G @= G\\
@VeVV  @. @VV\id V  @VVV\\
G @. \times @. G @. \id \\
@. @VV\cdot V @. @VVV\\
@. G @. @= G
\end{CD}\]

We define the co of this property so that the following diagram commutes:
\[\begin{CD}
K @. \times @. K[G] @= K[G]\\
@A\epsilon AA  @. @AA\id A @AAA\\
K[G] @. \times @. K[G] @. \id\\
@. @AA\Delta A @. @AAA\\
@. K[G] @. @= K[G]
\end{CD}\]

That is, $\id=(\epsilon\otimes\id)\circ\Delta.$ Similarly, $ge=g$ translates to:
$\id=(\id\otimes\epsilon)\circ\Delta.$ Therefore we get

\[\id=(\epsilon\otimes\id)\circ\Delta=(\id\otimes\epsilon)\circ\Delta.\]

\item
The last property is $gg^{-1}=e=g^{-1}g$. The first equality is equivalent to requiring that the following
diagram commute:
\[\begin{CD}
 @. G @. @= G\\
@.  @V\textrm{diag}VV @. @VVV\\
G @. \times @. G @. @VVV\\
@V()^{-1} VV @. @VV\id V e\\
G @. \times @. G @. @VVV\\
@. @VV\cdot V @. @VVV\\
@. G @. @= G
\end{CD}\]
Where $\textrm{diag}:G\rightarrow G\times G$ is the diagonal embedding. The co of $\textrm{diag}$ is $m:K[G]\leftarrow K[G]\otimes K[G]$ defined by $m(f_1, f_2)(g)=f_1(g)\cdot f_2(g)$. So the co of this property will hold when the following diagram commutes:

\[\begin{CD}
 @. K[G] @. @= K[G]\\
@.  @AmAA @. @AAA\\
K[G] @. \otimes @. k[G] @. @A\nu AA\\
@ASAA @. @AA\id A K\\
K[G] @. \otimes @. K[G] @. @AAA\\
@. @AA\Delta A @. @A\epsilon AA\\
@. K[G] @. @= K[G]
\end{CD}\]
Where $\nu$ is the embedding of $K$ into $K[G]$. Therefore the last property we want to be satisfied is:
\[m\circ(S\otimes\id)\circ\Delta=\nu\circ\epsilon.\] For $e=g^{-1}g$, we
similarly  get:
\[m\circ(\id\otimes S)\circ\Delta=\nu\circ\epsilon.\]

\end{enumerate}

\begin{defi}[Hopf algebra]
A $K$-algebra $A$ is called a \emph{Hopf algebra} if there exist homomorphisms $\Delta:A\otimes A\rightarrow A$,
$S:A\rightarrow A$, $\epsilon:A\rightarrow K$, and $\nu:A\rightarrow K$ that satisfy $(1)-(4)$ above, with $A$
in place of $K[G]$.
\end{defi}

We have shown that if $G$ is a group, the ring $K[G]$
of functions on $G$ is a (commutative)
Hopf algebra, which is non-co-commutative if $G$ is non-commutative.
Thus for every usual group, we get a 
commutative Hopf algebra. However, in general, Hopf algebras may be
 non-commutative.
\begin{defi}
A \emph{quantum group} is a (non-commutative and non-co-commutative) Hopf algebra.
\end{defi}

A nontrivial example of a quantum group will be constructed in the next lecture.

Next we want to look at what happens to group theoretic notions such as representations, actions, and
homomorphisms, in the context of Hopf algebras. These will correspond to co-representations, co-actions, and
co-homomorphisms.

Let us look closely at the notion of co-representation. A representation is a map $\cdot:G\times V\rightarrow
V$, such that \begin{itemize}\item$(h_1h_2)\cdot v = h_1\cdot(h_2\cdot v)$, and \item$e\cdot v=v$.\end{itemize}
Therefore a (right) co-representation of $A$ will be a linear mapping $\varphi:V\rightarrow V\otimes A$, where
$V$ is a $K$-vector space, and $\varphi$ satisfies the following:
\begin{itemize}
\item The following diagram commutes:
\[\begin{CD}
V\otimes A\otimes A @<\id\otimes\Delta<< V\otimes A\\
@A\varphi\otimes\id AA   @AA\varphi A\\
V \otimes A @<<\varphi < V
\end{CD}\]
That is, the following equality holds:
$$(\varphi\otimes \id)\circ\varphi=(\id\otimes\Delta)\circ\varphi.$$
\item The following diagram commutes:
\[\begin{CD}
V\otimes K@<\id<< V\otimes K\\
@A\id\otimes\epsilon AA   @|\\
V \otimes A @<<\varphi < V
\end{CD}\]
That is, the following equality holds:
$$(\id\otimes\epsilon)\circ\varphi=\id$$
\end{itemize}

In fact all usual group theoretic notions can be ``Hopfified" in this sense [exercise].

Let us look now at an example. Let $$G=GL_n(\C)=GL(\C^n)=GL(V),$$ where $V=\C^n$. Let $M_n$ be the matrix space
of $n\times n$ $\C$-matrices, and ${\cal O}(M_n)$ the coordinate ring of $M_n$, $${\cal O}(M_n)=\C[U]=C[\{u_j^i\}],$$ where
$U$ is an $n\times n$ variable matrix with entries $u_j^i$. 
Let $\C[G]={\cal O}(G)$ be the coordinate ring of $G$
obtained by adjoining $det(U)^{-1}$ to ${\cal O}(M_n)$. That is,
 $\C[G]={\cal O}(G)=\C[U][\det(U)^{-1}]$, which is the
$\C$ algebra generated by  $u_j^i$'s and $\det(U)^{-1}$.

\begin{prop}
\C[G] is a Hopf algebra, with $\Delta$, $\epsilon$, and $S$ as follows.
\begin{itemize}
\item Recall that the axioms of a Hopf algebra require that \[\Delta:\C[G]\rightarrow\C[G]\otimes\C[G],\]
$$\Delta(f)(g_1,g_2)=f(g_1g_2).$$ Therefore we define $$\Delta(u_j^i) = \sum_k u^i_k\otimes u_j^k,$$ where $U$ denotes 
the generic matrix in $M_n$ as above.
\item
Again, it is required that $$\epsilon(f)=f(e).$$ Therefore we define $$\epsilon(u_j^i)=\delta_{ij},$$ where
$\delta_{ij}$ is the Kronecker delta function.
\item
Finally, the antipode is required to satisfy $S(f)(g)=f(g^{-1})$. Let $\widetilde{U}$ be the cofactor matrix of
$U$, $U^{-1}=\frac{1}{\det(U)} \widetilde{U}$, and  $\widetilde{u}_j^i$
 the entries of $\widetilde{U}$.
Then we define $S$ by:$$S(u_j^i)=\frac{1}{\det(U)}\widetilde{u}_j^i = 
(U^{-1})^i_j.$$
\end{itemize}
\end{prop}

\chapter{Standard quantum group}
\begin{center} {\Large  Scribe: Paolo Codenotti} \end{center} 

\noindent {\bf Goal:} 
In this lecture we  construct the standard (Drinfeld-Jimbo) 
quantum group,
which is a  $q$-deformation of the general linear group $GL_\n(\C)$ 
with remarkable properties. 

\noindent {\em Reference:}  \cite{KS} 

Let $G=GL(V)=GL(\C^n)$, and $V=\C^n$. In the earlier 
lecture, we constructed the commutative and non co-commutative
Hopf algebra $\C[G]$. In this lecture we
quantize $\C[G]$ to  get
a non-commutative and non-co-commutative Hopf algebra  $\C_q[G]$,
and then define
the standard quantum group $G_q=GL_q(V)=GL_q(n)$ as the virtual object 
whose coordinate ring is $\C_q[G]$.

We start by defining $GL_q(2)$ and $SL_q(2)$, for $n=2$. Then we will generalize this construction to arbitrary
$n$. 
Let  ${\cal O}(M_2)$ be 
the coordinate ring of $M_2$, the set of $2\times 2$ complex matrices,
$\C[V]$  the coordinate ring of $V$ generated by 
the coordinates $x_1$ and $x_2$ of $V$ which 
satisfy $x_1x_2=x_2x_1$. Let
$$U=\left[
\begin{tabular}{ll}
a & b \\
c & d
\end{tabular}
\right]$$ be the  generic (variable) matrix in $M_2$.
It  acts on $V=\C^2$ from the left and from the
right. Let
$$x=\left[
\begin{tabular}{l}
$x_1$ \\
$x_2$
\end{tabular}
\right].$$ The left action is defined by \[x\rightarrow x^\prime:=Ux.\] Let
$$x^\prime=\left[\begin{tabular}{l}$x_1^\prime$\\$x_2^\prime$\end{tabular}\right].$$ Similarly, the right action
is defined by\[x^T\rightarrow (x^{\prime\prime})^T:=x^TU.\] Let
$$x^{\prime\prime}=\left[\begin{tabular}{l}$x_1^{\prime\prime}$\\$x_2^{\prime\prime}$\end{tabular}\right].$$
The action of $M_2$ on $V$ satisfies\[x_1^\prime x_2^\prime=x_2^\prime x_1^\prime\textrm{,
and}\]\[x_1^{\prime\prime} x_2^{\prime\prime}=x_2^{\prime\prime} x_1^{\prime\prime}.\]

Now instead of $V$, we  take its $q$-deformation $V_q$, a quantum space,
whose coordinates $x_1$ and $x_2$ satisfy
\begin{equation}\label{eq:comm}
x_1x_2=qx_2x_1,
\end{equation}
where $q\in \C$ is a parameter. Intuitively, in quantum physics if $x_1$ and $x_2$ are position and momentum, then
$q=e^{i\hbar}$ when $\hbar$
  is Planck's constant. Let $\C_q[V]$ be the ring generated by
$x_1$ and $x_2$ with the relation ($\ref{eq:comm}$). That is,\[\C_q[V]=\C[x_1, x_2]/<x_1x_2-qx_2x_1>.\]
It is  the coordinate ring of the quantum space $V_q$.
 Now we want to quantize $M(2)$ to get $M_q(2)$, the space
of quantum $2\times 2$ matrices, and $GL(2)$ to $GL_q(2)$, the space of quantum $2\times 2$ nonsingular
matrices. Intuitively, $M_q(2)$ is the space of linear transformations of the quantum space $V_q$ which preserve
the equation (\ref{eq:comm}) under the left and right actions, 
and similarly, $GL_q(2)$ is the space of non-singular
linear transformation that preserve the equation (\ref{eq:comm}) under the
left and right actions. We now formalize
this intuition.

Let $U=\left(
\begin{tabular}{ll}
a & b \\
c & d
\end{tabular}
\right)$ be a quantum matrix whose coordinates do not commute.
The  left and  right actions of $U$ must preserve \ref{eq:comm}.

\vspace{.1in}

\noindent [Left action:] 
Let the left action be $\varphi_L:x\rightarrow Ux$, and $Ux=x^\prime$.
 Then we must have:

\[\left(\begin{tabular}{ll}$a$ & $b$ \\$c$ & $d$\end{tabular}\right)
\left( \begin{tabular}{l} $x_1$ \\ $x_2$ \end{tabular} \right) = \left(\begin{tabular}{l}$ax_1+bx_2$
\\$cx_1+dx_2$\end{tabular}\right) = \left( \begin{tabular}{l} $x_1^\prime$ \\ $x_2^\prime$
\end{tabular} \right).\]

\vspace{.1in}

\noindent [Right action:]
Let the right action be $\varphi_R:x^T\rightarrow x^T U$,
and let $x^{\prime\prime}=(x^T U)^T=U^T x$. Then we must
have:

\[\left( \begin{tabular}{ll} $x_1$ & $x_2$ \end{tabular} \right) \left(\begin{tabular}{ll}$a$ & $b$ \\$c$ & $d$\end{tabular}\right)
= \left(\begin{tabular}{l}$ax_1+cx_2$
\\$bx_1+dx_2$\end{tabular}\right) = \left( \begin{tabular}{l} $x_1^{\prime\prime}$ \\ $x_2^{\prime\prime}$
\end{tabular} \right).\]

The preservation of $x_1x_2=qx_2x_1$ under left multiplication means $$x_1^\prime x_2^\prime=qx_2^\prime x_1^\prime.
$$
That is,
\begin{equation}\label{eq:pres}
(ax_1+bx_2)(cx_1+dx_2)=q(cx_1+dx_2)(ax_1+bx_2).
\end{equation}
The left hand side of (\ref{eq:pres}) is
$$acx_1^2+bcx_2x_1+adx_1x_2+bdx_2^2=acx_1^2+(bc+adq)x_2x_1+bdx_2^2.$$
Similarly, the right hand side of (\ref{eq:pres}) is
$$q(cax_1^2+(da+cbq)x_2x_1+bdx_2^2).$$
Therefore equation (\ref{eq:pres}) implies:
\[\begin{array}{l}
    ac=qca\\
    bd=qdb\\
    bc+adq=da+qcb.
\end{array} \]
That is, 
\[\begin{array}{l}
    ac=qca\\
    bd=qdb\\
    ad-da-qcb+q^{-1}bc=0.
\end{array}\]

Similarly, since $x_1^{\prime\prime}x_2^{\prime\prime}=qx_2^{\prime\prime}x_1^{\prime\prime}$, we get:
\[\begin{array}{l}
    ab=qba\\
    cd=qdc\\
    ad-da-qbc+q^{-1}cb=0.
\end{array}\]

The last equations from each of these sets imply $bc=cb$.

So we  define ${\cal O}(M_q(2))$, the coordinate ring of the space
 of $2\times 2$ quantum matrices $M_q(2)$, to be
the $\C$-algebra with generators $a$, $b$, $c$, and $d$,
 satisfying the relations:
\begin{eqnarray*} ab=qba,
\quad ac=qca, \quad bd=qdb, \quad cd=qdc,\\ \quad bc=cb, \quad ad-da=(q-q^{-1})bc.
\end{eqnarray*}

Let
\begin{eqnarray*}
U = \left(\begin{array}{ll}a & b \\ c & d\end{array}\right) = \left(
\begin{array}{ll}u^1_{1} & u^1_2 \\ u^2_1 & u^2_2 \end{array}\right).\end{eqnarray*}
Define the quantum determinant of $U$ to be 
\[D_q=\det(U)=ad-qbc=da-q^{-1}bc.\]
Define $\C_q[G]={\cal O}(GL_q(2))$, the coordinate ring of the virtual 
quantum group $GL_q(2)$ of invertible $2\times 2$ quantum matrices, to be
\[{\cal O}(GL_q(2))={\cal O}(M_q(2))[D_q^{-1}],\]
where the square brackets indicate adjoining.

\begin{prop}
The coordinate ring
${\cal O}(GL_q(2))$ is a Hopf algebra, with $$\Delta(u^i_j)=\sum_k u_k^i\otimes u_j^k,$$ $$S(u_j^i)=\frac{1}{D_q}
\widetilde{u}_j^i=(U^{-1})^i_j,$$ $$\epsilon(u_j^i)=\delta_{i j},$$ 
where 
$\widetilde{U}=[\tilde u^i_j]$ is the cofactor 
matrix
\begin{eqnarray*}
\widetilde{U} = \left(\begin{array}{ll}d & -q^{-1}b \\ -q c & a\end{array}\right).\end{eqnarray*}
 (defined  so that $U \widetilde{U}=D_q I$) and
$U^{-1}=\tilde U/D_q$ is the inverse of $U$.
\end{prop}
This is a non-commutative and non-co-commutative Hopf algebra.

Now we go to the general $n$. Let $V_q$ be the 
 $n$-dimensional quantum space, the $q$-deformation of $V$,
 with coordinates $x_i$'s which
satisfy \begin{equation}\label{eq:commute}x_ix_j=qx_jx_i \quad \forall i<j.\end{equation} Let $\C_q[V]$ be the
coordinate ring of $V_q$ defined by \[\C_q[V]=\C[x_1, \dots, x_n]/<x_ix_j-qx_jx_i>.\] Let $M_q(n)$ be the space
of quantum $n\times n$ matrices, that is the set of linear transformations on $V_q$ which preserve
(\ref{eq:commute}) under the left as well as the right action. The left action 
is given by:
$$\left[\begin{array}{l}x_1\\...\\x_n \end{array}\right]=x\rightarrow Ux=x^\prime,$$
where $U$ is the  $n\times n$ generic quantum matrix. Similarly,
 the right action is given by: $$x^T\rightarrow
x^TU=(x^{\prime\prime})^T.$$

Preservation of (\ref{eq:commute}) under the left and right actions means:
\[x_i^\prime y_j^\prime=qx_j^\prime x_i^\prime,\quad \textrm{for}\quad i<j \]
\[x_i^{\prime\prime} y_j^{\prime\prime}=qx_j^{\prime\prime} x_i^{\prime\prime},\quad \textrm{for}\quad i<j. \]

After straightforward calculations, 
these yield the following relations on the entries $u_{ij}=u^i_j$ of $U$:
\begin{eqnarray}\label{eq:reltns}
u_{jk}u_{ik}  = q^{-1}u_{ik}u_{jk} & (i<j) \nonumber \\
u_{kj}u_{ki}  = q^{-1}u_{ki}u_{kj} & (i<j) \nonumber \\
u_{jk}u_{i\ell}  = u_{i\ell}u_{jk} & (i<j, k<\ell) \nonumber \\
u_{jl}u_{ik}  = u_{ik}u_{j\ell}-(q-q^{-1})u_{jk}u_{i\ell} & (i<j, k<\ell).
\end{eqnarray}

The quantum determinant is defined as
\[D_q=\sum_{\sigma\in S_n} (-q)^{\ell(\sigma)}u_{j1}^{i_{\sigma(1)}}
\dots u_{jn}^{i_{\sigma(n)}},\]
where $\ell(\sigma)$ denotes the length of the permutation $\sigma$,
 that is, the number of inversions in
$\sigma$. This  determinant formula is the same as the usual formula substituting $(-q)$ for $(-1)$.

We define the coordinate ring of the space $M_q(n)$ 
of quantum $n\times n$ matrices by
\[{\cal O}(M_q(n)) = \C[U]/<(\ref{eq:reltns})>,\ \textrm{and}\]
and the coordinate ring of the virtual quantum group $GL_q(n)$ by 
\[\C_q[G]={\cal O}(GL_q(n))={\cal O}(M_q(n))[D_q^{-1}].\]
We  define the quantum minors and, using these,
 the quantum co-factor matrix $\widetilde{U}$ and the quantum inverse
matrix $U^{-1}=\widetilde{U}/D_q$ in a straightforward fashion
(these constructions are left as exercises).

\begin{thm}
The algebra ${\cal O}(GL_q(n))$ is a Hopf algebra, with
\[\Delta(u^i_j)=\sum_k u_k^i\otimes u_j^k\]
\[\epsilon(u_j^i)=\delta_{ij}\]
\[S(u_j^i)=\frac{1}{D_q}\widetilde{u}_j^i=(U^{-1})^i_j\]
\[S(D_q^{-1})=D_q.\]
\end{thm}

We also denote the quantum group 
$GL_q(n)$ by $G_q$, $GL_q(\C^n)$  or $GL_q(V)$. 
It has to be emphasized that this is only  a virtual
object. Only its coordinate ring $\C_q[G]$ is real.
Henceforth, whenever we say  representation or action of $G_q$, we
actually mean corepresentation or coaction of $\C_q[G]$, and so forth.

\chapter{Quantum unitary group}
\begin{center}  {\Large Scribe: Joshua A. Grochow} \end{center} 

\noindent{\bf Goal:} 
Define the quantum unitary subgroup of the standard quantum group. 

\noindent {\em Reference:} \cite{KS} 

\subsection*{Recall} 
Let $V=\C^n$, $G=GL_n(\C)=GL(V)=GL(\C^n)$, and $\mathcal{O}(G)$ 
the coordinate ring of $G$.
 The quantum group $G_q = GL_q(V)$ is
the virtual object whose coordinate ring is $$\mathcal{O}(G_q) = \C[U]/\langle \mbox{ relations } \rangle,$$ 
where $U$ is the generic $n \times n$ matrix of indeterminates,
and the relations are the  quadratic relations on the coordinates
 $u_{i}^j$ defined in the last class
so as to preserve the non-commuting relations 
among the coordinates of the quantum vector space $V_q$  on which $G_q$ acts.
This coordinate ring is a  Hopf algebra.  

\section{A $q$-analogue of the unitary group}
In this lecture we 
define a $q$-analogue of the unitary subgroup
$U=U_n(\C) =U(V)  \subseteq GL_n(\C)=GL(V)=G$. This  is a
$q$-deformation $U_q=U_q(V) \subseteq G_q$ of $U(V)$. 
 Since $G_q$ is only a virtual object, $U_q$ will also be virtual.
  To define $U_q$, we must determine how to capture the notion of unitarity in the setting of Hopf algebras.  As we shall see, it is captured by the notion of a Hopf $*$-algebra.

\begin{defi} A \emph{$*$-vector space} is a vector space $V$ with an involution $*: V \to V$ satisfying 
$$
\begin{array}{lr}
(\alpha v + \beta w)^* = \overline{\alpha} v^* + \overline{\beta} w^* & (v^*)^*=v
\end{array}
$$
for all $v,w \in V$, and $\alpha,\beta \in \C$.  \end{defi}

We  think of $*$ as a generalization of complex conjugation; and
in fact every complex vector space is a $*$-vector space, where  $*$ is exactly complex conjugation.

\begin{defi} A \emph{Hopf $*$-algebra} is a Hopf algebra $(A,\Delta,\epsilon,S)$ with an involution $*: A \to A$ such that $(A,*)$ is a $*$-vector space, and:
\begin{enumerate}
\item $(ab)^* = b^* a^*$, $1^* = 1$
\item $\Delta(a^*) = \Delta(a)^*$ (where $*$ acts diagonally on the tensor product $A \otimes A$: $(v \otimes w)^* = (v^* \otimes w^*)$)
\item $\epsilon(a^*) = \overline{\epsilon(a)}$
\end{enumerate}
 \end{defi}
There 
is no explicit condition here on how $*$ interacts with the antipode $S$.

Let ${\cal O}(G)=\C[G]$ be the coordinate ring of $G$ as defined earlier.

\begin{prop} Then $\mathcal{O}(G)$ is a Hopf $*$-algebra. \end{prop}

\begin{proof} 
We think of the elements in  $\mathcal{O}(G)$ as $\C$-valued functions
on  $G$ and 
define $*:\mathcal{O}(G) \to \mathcal{O}(G)$ so
that it satisfies the three conditions for a Hopf $*$-algebra, and 
\begin{enumerate}
\item[(4)] For all $f \in \mathcal{O}(G)$ and  $g \in U \subseteq G$, $f^*(g) = \overline{f(g)}$
\end{enumerate}
  Let $u_{i}^j$ be the coordinate functions
 which, together with $D^{-1}$, $D=\det(U)$, generate $\mathcal{O}(G)$. 
Because of the first condition on a Hopf $*$-algebra (relating the involution $*$ to multiplication), specifying $(u_i^j)^*$ and $D^*$
suffices to define $*$ completely.  We define
$$
(u_i^j)^* = S(u_j^i) = (U^{-1})_j^i
$$
and $D^*=D^{-1}$. We can check that this satifies (1)-(4).  Here we will
only  check (4), and leave the  remaining verification  as an exercise. 
Let $g$ be an element of the unitary group $U$.  Then $(u_i^j)^*(g) = S(u_j^i)(g) = (g^{-1})_j^i = (\overline{g})_i^j$, where the last equality follows from the fact that $g$ is unitary (i.e. $g^{-1}=g^\dagger$, where $\dagger$ denotes conjugate transpose). \end{proof}

Thus, we have defined a map $f \mapsto f^*$ purely algebraically in such a way that the restriction of $f^*$ to the unitary group $U$ is the same as taking the complex conjugate $\overline{f}$ on $U$.  

\begin{prop} The coordinate ring $\C_q[G]=\mathcal{O}(G_q)$
of the quantum group $G_q=GL_q(V)$  is also a Hopf $*$-algebra. \end{prop}

\begin{proof} The proof is syntactically identical to the proof for $\mathcal{O}(G)$, except that the coordinate function $u_i^j$ now lives
 in $\mathcal{O}(G_q)$ and the determinant $D$ becomes the $q$-determinant $D_q$.  The definition of $*$ is: $(u_i^j)^* = S(u_j^i)$ and $D_q^* = D_q^{-1}$, essentially the same as in the classical case.  \end{proof}

Intuitively, the ``quantum subgroup'' $U_q$ of $G_q$ is  the virtual object
such that the restriction to $U_q$ of the involution $*$ just defined 
coincides with the complex conjugate.  

\section{Properties of $U_q$}
We would like the nice properties of the classical unitary group to transfer over to the quantum unitary group, and  this is indeed the case.
  Some of the nice properties of  $U$ are:
\begin{enumerate}
\item It is compact, so we can integrate over $U$.
\item we can do harmonic analysis on $U$ (viz. the Peter-Weyl Theorem,
which is an analogue for $U$ of the 
Fourier analysis on the circle $U_1$).
\item Every finite dimensional representation of $U$ has a $G$-invariant Hermitian form, and thus a unitary basis -- we
say that every finite dimensional representation of $U$ is {\em unitarizable}.
\item Every finite dimensional representation $X$ 
of $U$ is completely reducible;
this follows  from (3), since any subrepresentation $W \subseteq X$ has a perpendicular subrepresentation $W^\bot$ under the $G$-invariant Hermitian form.
\end{enumerate}

Compactness is in some sense the key here.
The question is how to define it in the quantum setting.
Following Woronowicz, we define 
compactness to mean that every finite dimensional representation of $U_q$ is
unitarizable.  Let us see what this means formally.

Let $A$ be a Hopf $*$-algebra, and $W$ a corepresentation of $A$.
Let  $\rho: W \to W \otimes A$ be the corepresentation map.
  Let $\{b_i\}$ be a basis of $W$.  Then, under $\rho$, $b_i \mapsto \sum_j b_j \otimes m_i^j$ for some $m_i^j \in A$.  We can thus define the \emph{matrix of the (co)representation} $M=(m_i^j)$
 in the basis $\{b_i\}$.  We define $M^*$ such that $(M^*)_i^j = (M_j^i)^*$.
Thus, in the classical case (i.e. when $q=1$), $M^* = M^\dag$.  

We say that the corepresentation $W$ is \emph{unitarizable} if it has  a
basis $B=\{b_i\}$  such that the corresponding matrix  $M_B$
of corepresentation  satisfies the unitarity condition: $M_B M_B^* = I$. 
In this case, we say $B$ is a unitary basis of the corepresentation $W$.

\begin{defi} A Hopf $*$-algebra $A$ is \emph{compact} if every finite dimensional corepresentation of $A$ is unitarizable.  \end{defi}

\begin{thm}[Woronowicz] The coordinte ring 
$\C_q[G]=\mathcal{O}(G_q)$ is a compact Hopf $*$-algebra.
This implies that every finite dimensional representation of  $G_q$,
by which we mean a finite dimensional coorepresentation of  $\C_q[G]$,
is completely reducible.  \end{thm}

Woronowicz goes further to show that we can $q$-integrate on $U_q$, and
that we can do quantized harmonic analysis on $U_q$; i.e., a quantum
analogue of the Peter-Weyl theorem holds.

Now that we know the finite dimensional representations of $G_q$ are
completely reducible, we can ask what the irreducible representations are.

\section{Irreducible Representations of $G_q$}
We  proceed by analogy with the Weyl modules $V_\lambda(G)$ for $G$.  
Recall that every polynomial irreducible representation of $G=GL_n(\C)$ is of this form.

\begin{thm} \begin{enumerate}
\item For all partitions $\lambda$ of length at most $n$, there exists a $q$-Weyl module $V_{q,\lambda}(G_q)$ which is an irreducible representation of $G_q$ such that
$$
\lim_{q \to 1} V_{q,\lambda}(G_q) = V_\lambda(G).
$$
\item The $q$-Weyl modules give all polynomial irreducible representations of $G_q$.
\end{enumerate}
\end{thm}

\section{Gelfand-Tsetlin basis}
To understand the $q$-Weyl modules better, 
we wish to get an explicit basis for each module $V_{q,\lambda}$. 
 We begin by defining a very useful basis -- the Gel'fand-Testlin basis -- in the classical case for $V_\lambda(G)$, and then describe the $q$-analogue of this basis.

By Pieri's rule \cite{FulH}
$$
V_\lambda(GL_n(\C)) = \bigoplus_{\lambda'} V_{\lambda'}(GL_{n-1}(\C))
$$
where the sum is taken over all $\lambda'$ obtained from $\lambda$ by removing any number of boxes (in a legal way) such that no two removed boxes come
from the same column. This is an orthogonal decomposition (relative to the
$GL_n(\C)$-invariant Hermitian form on $V_\lambda$) and it is also 
multiplicity-free, i.e.,  each $V_{\lambda'}$ appears only once.

Fix a $G$-invariant Hermitian form on $V_\lambda$.  Then the \emph{Gel'fand-Tsetlin} basis for $V_\lambda(GL_n(\C))$, denoted $GT_\lambda^n$, is the unique orthonormal basis for $V_\lambda$ such that
$$
GT_\lambda^n = \bigcup_{\lambda'} GT_{\lambda'}^{n-1},
$$
where the disjoint union is   over the $\lambda'$ as in Pieri's rule,
 and $GT_{\lambda'}^{n-1}$ is defined recursively, the case $n=1$ being 
trivial.

The dimension of $V_\lambda$ is the number of semistandard tableau of shape $\lambda$.  With any tableau $T$ of this shape, one can also explicitly associate 
a basis element $GT(T) \in GT_\lambda^n$; we shall not worry about how.

We can define the Gel'fand-Tsetlin basis  $GT_{q,\lambda}^n$ 
for $V_{q,\lambda}(G_q(\C^n))$ analogously.
We have the $q$-analogue of Pieri's rule:
$$
V_{q,\lambda}(G_q(\C^n)) = \bigoplus_{\lambda'} V_{q,\lambda'}(G_q(\C^{n-1}))
$$
where the decomposition is orthogonal and multiplicity-free, 
and the sum ranges over the same $\lambda'$ as above.  So we can define $GT_{q,\lambda}^n$ to be the unique unitary  basis  of $V_{q,\lambda}$ 
such that
$$
GT_{q,\lambda}^n = \bigcup_{\lambda'} GT_{q,\lambda'}^{n-1}.
$$
With  any semistandard tableau $T$, one can also explicitly 
associate  a basis element $GT_q(T) \in GT_{q,\lambda'}^n$; details omitted.

\chapter{Towards positivity hypotheses 
 via  quantum groups}  \label{ctowards}
\begin{center} {\Large  Scribe: Joshua A. Grochow} \end{center} 

\noindent {\bf Goal:} 
In this final brisk lecture, we indicate  the role of quantum groups in
the context of the positivity hypothesis PH1. Specifically, we sketch  how
the Littlewood-Richardson rule -- the gist of PH1 in the Littlewood-Richardson
problem --  follows from the theory of standard quantum groups. 
We  then  briefly mention analogous (nonstandard) quantum groups
for the Kronecker and plethysm problems defined in \cite{GCT4,GCT7},
and  the  theorems and conjectures for them
that would imply PH1 for these problems.

\noindent{\em References:} \cite{KS,Kas,lusztigbook,GCT4,GCT6,GCT7,GCT8}

Let $V=\C^n$, $G=GL_n(\C)=GL(V)$, $V_\lambda=V_\lambda(G)$
a Weyl module of $G$, 
$G_q=GL_q(V)$ the standard quantum group, 
$V_q$ the $q$-deformation of $V$ on which $GL_q(V)$ acts, 
$V_{q,\lambda}=V_{q,\lambda}(G_q)$ the $q$-deformation of $V_\lambda(G)$, 
and  $GT_{q,\lambda}=GT_{q,\lambda}^n$
the Gel'fand-Tsetlin basis for $V_{q,\lambda}$.

\section{Littlewood-Richardson rule via standard quantum groups}
We now sketch how the Littlewood-Richardson rule falls out of the
standard quantum group machinery, specifically the properties of the
Gelfand-Tsetlin basis.

\subsection{An embedding of the Weyl module} 
For this, we have to embed the $q$-Weyl module $V_{q,\lambda}$ in
$V_q^{\otimes d}$, where 
$d=|\lambda|=\sum \lambda_i$ is the size of $\lambda$.
We  first describe how to embed the Weyl module 
$V_\lambda$ of $G$ in $V^{\otimes d}$ in a standard 
way that can be quantized.

If $d=1$, then $V_\lambda(G)=V=V^{\otimes 1}$.
Otherwise, obtain a Young diagram $\mu$ from $\lambda$ by
removing its  top-rightmost box that can be removed 
to get   a valid Young diagram, e.g.:
\newcommand{\xx}{\mbox{x}}
$$
\begin{array}{ccc}
\young(\hfil \hfil \hfil \xx,\hfil \hfil \hfil,\hfil \hfil,\hfil \hfil,\hfil) & \leadsto & \young(\hfil \hfil \hfil,\hfil \hfil \hfil,\hfil \hfil,\hfil \hfil,\hfil) \\
\lambda & & \mu
\end{array}
$$

In the following, the box must be removed from the second row, since removing from the first row would result in an illegal Young diagram:
$$
\begin{array}{ccc}
\young(\hfil \hfil \hfil,\hfil \hfil \xx,\hfil) & \leadsto & \young(\hfil \hfil \hfil,\hfil \hfil,\hfil) \\
\lambda & & \mu
\end{array}
$$
By induction on $d$, we have a standard embedding $V_\mu(G) \hookrightarrow V^{\otimes d-1}$.  This gives us an embedding $V_\mu(G) \otimes V \hookrightarrow V^{\otimes d}$.  By Pieri's rule \cite{FulH} 
$$
V_\mu(G) \otimes V = \bigoplus_\beta V_\beta(G),
$$
where the sum is over all $\beta$ obtained from $\mu$ by adding one box in a legal way.  In particular,  $V_\lambda(G) \subset V_\mu(G) \otimes V$.
By restricting the above embedding,  we get a standard embedding $V_\lambda(G) \hookrightarrow V^{\otimes d}$.  

Now Pieri's rule also holds in a quantized setting:
$$
V_{q,\mu} \otimes V_q = \bigoplus_\beta V_{q,\beta}(G),
$$
where $\beta$ is as above. Hence, the standard embedding  $V_\lambda
\hookrightarrow V^{\otimes d}$
above can be quantized in a straightforward fashion to get a
standard embedding  $V_{q,\lambda} \hookrightarrow V_q^{\otimes d}$.
We shall denote it by $\rho$.
Here the  tensor product  is meant to be over  $\Q(q)$.
Actually, $\Q(q)$ doesn't quite work. 
We have to allow square roots of elements of $\Q(q)$,
but we won't worry about this.
For a semistandard tableau $b$ of shape $\lambda$, we denote the image of
a Gelfand-Tsetlin basis element $GT_{q,\lambda}(b) 
\in GT_{q,\lambda}$ under
 $\rho$ by $GT_{q,\lambda}^\rho(b) = \rho(GT_{q,\lambda}(b))
\in V_q^{\otimes d}$.  

\subsection{Crystal operators and crystal bases}

\begin{thm}[Crystallization]\cite{date} \label{tdate}
The Gelfand-Tsetlin basis elements crystallize at $q=0$. This means:
\begin{equation}
\lim_{q \to 0} GT_{q,\lambda}^\rho(b)  =  v_{i_1(b)} \otimes \cdots \otimes v_{i_d(b)},
\end{equation} 
 for some integer functions  $i_1(b),\dots,i_d(b)$, and 
\begin{equation}
\lim_{q \to \infty} GT_{q,\lambda}^\rho(b)  =  v_{j_1(b)} \otimes \cdots \otimes v_{j_d(b)},
\end{equation}  for some integer functions  $j_1(b),\dots,j_d(b)$. 
\end{thm}

The phenomenon 
that these limits consists of monomials, i.e., simple tensors is
known as {\em crystallization}.
It is  related to the physical phenomenon of crystallization, hence the name.  
The maps $b \mapsto \overline{i}(b)=(i_1(b),\dots,i_d(b))$ and $b \mapsto \overline{j}(b)=(j_1(b),\dots,j_d(b))$ are computable in $poly(\bitlength{b})$ time (where $\bitlength{b}$ is the bit-length of $b$).  

Now we want to define a special {\em crystal basis} of $V_{q,\lambda}$ 
based on this phenomenon of crystallization.  Towards that end,
consider the following family of $n\times n$ matrices:
$$
E_i = \left(\begin{array}{cccccc}
0 & 0      &        &        & \cdots & 0 \\
  & \ddots & \ddots &        &        & \vdots \\
  &        & 0      & 1      & \cdots & 0 \\
  &        &        & 0      & \cdots & 0 \\
  &        &        &        & \ddots & \vdots \\
  &        &        &        &        & 0 
\end{array}\right), 
$$
where the only nonzero entry is a 1 in the $i$-th row and $(i+1)$-st column. 
Let $F_i=E_i^T$.  
Corresponding to $E_i$ and $F_i$, Kashiwara associates certain
operators $\hat E_i$ and $\hat F_i$ on 
$V_{q,\lambda}(G_q)$. We shall not worry about their actual construction
here (for the readers familiar with Lie algebras: these are closely related 
to the usual operators  in the Lie algebra of $G$
associated with  $E_i$ and $F_i$).

If we let   $\hat E_i$ act on $GT_{q,\lambda}^\rho(b)$,
we get some linear combination 
$$
\hat E_i(GT_{q,\lambda}^\rho(b)) = \sum_{b'} a_{b'}^b(q) GT_{q,\lambda}^\rho(b'),
$$
where $a_{b'}^b(q) \in \Q(q)$ (actually an algebraic extension of $\Q(q)$ as 
mentioned above). Essentially because of crystallization (Theorem~\ref{tdate}),
it turns out that $\lim_{q \to 0}a_{b'}^b(q)$ is always either 0 or 1, and for a given $b$, this limit is 1 for at most {one} $b'$, if any. 
A similar result holds for $\hat F_i(GT_{q,\lambda}^\rho(b))$.  This allows us to
 define the \emph{crystal operators} (due to Kashiwara):
$$
\widetilde{e_i} \cdot b = \left\{\begin{array}{ll}
b' & \mbox{ if } \lim_{q \to 0}a_{b'}^b(q) = 1, \\
0  & \mbox{ if no such $b'$ exists, }
\end{array}\right.
$$
and similarly for $\widetilde{f_i}$.  Although these operators are defined according to a particular embedding $V_{q,\lambda} \hookrightarrow V_q^{\otimes d}$
and a  basis, they can be defined intrinsically, i.e.,
without reference to the embedding or the Gel'fand-Tsetlin basis.

Now, let $W$ be a finite-dimensional representation of $G_q$, and  $R$
 the subring of functions in $\Q(q)$ regular at $q=0$ (i.e. without a
pole at $q=0$). 
 A \emph{lattice} within $W$ is  an $R$-submodule of $W$ such that $\Q(q) \otimes_R L = W$.  
(Intuition behind this definition: $R \subset \Q(q)$ is analogous to $\Z \subset \Q$.  A lattice in $\R^n$ is a $\Z$-submodule $L$ of $\R^n$ 
such that $\R \otimes_{\Z} L = \R^n$.)

\begin{defi} An (upper)  \emph{crystal basis} of a representation $W$ of $G_q$ is a pair $(L,B)$ such that
\begin{itemize}
\item $L$ is a lattice in $W$ preserved by the Kashiwara operators $\hat E_i$ and
$\hat F_i$, i.e. $\hat E_i(L) \subseteq L$ and $\hat F_i(L) \subseteq L$.
\item $B$ is a basis of $L/qL$ preserved by the crystal operators
$\widetilde{e_i}$ and $\widetilde{f_i}$, i.e.,
 $\widetilde{e_i}(B) \subseteq B \cup \{0\}$ and $\widetilde{f_i}(B) \subseteq B \cup \{0\}$.
\item The crystal operators $\widetilde{e_i}$ and $\widetilde{f_i}$ are inverse to each other wherever possible, i.e.,
 for all $b,b' \in B$, if $\widetilde{e_i}(b)=b' \neq 0$ then $\widetilde{f_i}(b')=b$,
and similarly,
 if $\widetilde{f_i}(b)=b' \not = 0$ then $\widetilde{e_i}(b')=b$.
\end{itemize}
\end{defi}

It can be shown that if $W=V_{q,\lambda}(G_q)$, then there exists a unique
 $b \in B$ such that $\widetilde{e_i}(b)=0$ for all $i$;
this corresponds to the highest weight vector of $V_{q,\lambda}$ 
(the weight vectors in $V_{q,\lambda}$ are analogous to the weight vectors
in $V_\lambda$; we do not give their exact definition here).
By the work of Kashiwara and Date et al \cite{Kas,date} above,
the Gel'fand-Tsetlin basis (after appropriate rescaling)
is in fact a crystal basis: just let
\begin{eqnarray*}
L=L_{GT} & = & \mbox{ the $R$-module generated by } GT_{q,\lambda}, and  \\
B_{GT} & = & \overline{GT_{q,\lambda}(b)},
\end{eqnarray*}
where $\overline{GT_{q,\lambda}(b)}$ is the image under
the projection $L \mapsto L/qL$ of the set of 
basis vectors in $GT_{q,\lambda}(b)$.  

\begin{thm}[Kashiwara] \begin{enumerate}
\item Every finite-dimensional $G_q$-module has a unique crystal basis (up to isomorphism).
\item Let $(L_\lambda,B_\lambda)$ be the unique crystal basis corresponding to $V_{q,\lambda}$.  Then $(L_\alpha,B_\alpha) \otimes (L_\beta,B_\beta) = (L_\alpha \otimes L_\beta, B_\alpha \otimes B_\beta)$ is
 the unique crystal basis of $V_{q,\alpha} \otimes V_{q,\beta}$,
where $B_\alpha \otimes B_\beta$ denotes $\{ b_a \otimes b_b | b_a \in B_\alpha, b_b \in B_\beta\}$.
\end{enumerate}
\end{thm}

It can be shown  that every $b \in B_\lambda$ has a weight; i.e., 
it is the image of a weight vector in $L_\lambda$ under the projection
$L_\lambda \rightarrow L_\lambda / q L_\lambda$.

Now let us see how the Littlewood-Richardson rule falls out of the properties
of the crystal bases.
Recall that   the specialization of $V_{q,\alpha}$ at $q=1$ is the
Weyl module $V_\alpha$ of  $G=GL_n(\C)$, and 
\begin{equation} \label{eqlittle} 
V_\alpha \otimes V_\beta = \bigoplus_\gamma c_{\alpha,\beta}^\gamma V_\gamma
\end{equation}
where  $c_{\alpha,\beta}^\gamma$ are the Littlewood-Richardson coefficients.
The Littlewood-Richardson rule now  follows  from the following fact:

$$
c_{\alpha,\beta}^\gamma = \#\{b \otimes b' \in B_\alpha \otimes B_\beta | 
\forall i, \ \widetilde{e_i}(b \otimes b') = 0  \mbox{ and } b \otimes b' \mbox{ has weight } \gamma\}.
$$
Intuitively,  $b \otimes b'$ here   correspond to the
highest weight vectors of the  $G$-submodules
of $V_\alpha \otimes V_\beta$ isomorphic to $V_\gamma$.  

\section{Explicit decomposition of the tensor product}  
The decomposition (\ref{eqlittle}) is only an abstract decomposition
of $V_\alpha \otimes V_\beta$ as a $G$-module.
Next we consider the explicit decomposition problem.
The goal is to find an explicit basis
${\cal B}={\cal B}_{\alpha\otimes \beta}$ 
of $V_\alpha \otimes V_\beta$ 
that is compatible with this abstract decomposition. Specifically, 
we want to construct  an explicit basis ${\cal B}$ of
 $V_\alpha \otimes V_\beta$ 
in terms of suitable explicit bases of $V_\alpha$ and $V_\beta$ 
such that ${\cal B}$ has a filtration
$$
{\cal B} = {\cal B}_0 \supseteq {\cal B}_1 \supseteq \cdots \supseteq \emptyset
$$
where each $\langle {\cal B}_i \rangle / \langle {\cal B}_{i+1} \rangle$ is an irreducible representation of $G$  and
$\langle {\cal B}_i \rangle$ denotes the linear span of ${\cal B}_i$.
Furthermore, each element 
$b\in {\cal B}$ should have a sufficiently explicit 
representation  in terms of the  basis 
${\cal B}_\alpha \otimes {\cal B}_\beta$ of $V_\alpha \otimes V_\beta$.
The explicit decomposition problem 
for the  $q$-analogue $V_{q,\alpha} \otimes V_{q,\beta}$ is similar.

For example, we have already constructed  explicit  Gelfand-Tsetlin bases of 
Weyl modules.
But it is  not known how to construct an explicit  basis ${\cal B}$ with
filtration as above
in terms of the Gelfand-Tsetlin bases of $V_\alpha$ and $V_\beta$ (except
when  the Young diagram of either $\alpha$ or $\beta$ is a single row). 

Kashiwara and Lusztig \cite{Kas,lusztigbook} construct certain
{\em canonical bases}  ${\cal B}_{q,\alpha}$ and ${\cal B}_{q,\beta}$ 
of $V_{q,\alpha}$ and $V_{q,\beta}$, and Lusztig furthermore constructs 
a {\em canonical basis} 
${\cal B}_q={\cal B}_{q,\alpha \otimes \beta}$ of 
$V_{q,\alpha} \otimes V_{q,\beta}$ such that:
\begin{enumerate}
\item ${\cal B}_q$  has a filtration as above,
\item 
Each $b \in {\cal B}_q$ has an expansion of the form
$$
b = \sum_{b_\alpha \in {\cal B}_{q,\alpha}, b_\beta \in {\cal B}_{q,\beta}}
a_{b}^{b_\alpha,b_\beta} b_\alpha \otimes b_\beta,
$$
where each $a_{b}^{b_\alpha,b_\beta}$ is a polynomial in $q$ and $q^{-1}$
with nonnegative integral coefficients,
\item {\em Crystallization}: 
For each $b$, as $q \to 0$, 
exactly one coefficient $a_{b}^{b_\alpha,b_\beta} \to 1$, 
and the remaining all vanish.
\end{enumerate} 
The proof of nonnegativity of the coefficients 
of $a_{b}^{b_\alpha,b_\beta}$ is based
on  the Riemann hypothesis (theorem) over finite fields \cite{weil2},
and  explicit formulae for these coefficients are known
in terms of perverse sheaves \cite{beilinson} 
(which are certain types of algebro-geometric objects).

This then provides a satisfactory solution to the explicit decomposition 
problem, which is far harder and deeper than the abstract decomposition 
provided by the Littlewood-Richardson rule.
By specializing at $q=1$, we also get a solution to the explicit
decomposition problem for $V_\alpha \otimes V_\beta$. This (i.e. via
quantum groups) is the only known solution to the explicit 
decomposition problem even at $q=1$. This may give some idea of the
power of the quantum group machinery.

\section{Towards nonstandard quantum groups
 for the Kronecker and plethysm problems}
Now the goal is to construct  quantum groups which can be used to
derive PH1 and explicit decomposition 
for the Kronecker and plethysm problems just as the
standard quantum group can be used for the same in
 the Littlewood-Richardson problem.

In the Kronecker problem, we let $H=GL(\C^n)$ and $G=GL(\C^n \otimes \C^n)$.
The  Kronecker coefficient
$\kappa_{\alpha,\beta}^\gamma$ is the multiplicity of 
 $V_\alpha(H) \otimes V_\beta(H)$ in $V_\gamma(G)$:
$$
V_\gamma(G) = \bigoplus_{\alpha,\beta} \kappa_{\alpha,\beta}^\gamma V_\alpha(H) \otimes V_\beta(H).
$$
The goal  is to get a positive  \cc{\# P}-formula for
$\kappa_{\alpha,\beta}^\gamma$; this is the gist of PH1 for the
Kronecker problem.

In the plethysm problem, 
we let $H=GL(\C^n)$ and $G=GL(V_\mu(H))$.  The
 plethysm constant
 $a_{\lambda,\mu}^\pi$ is the multiplicity of
$V_\pi(H)$ in $V_\lambda(G)$:
$$
V_\lambda(G) = \bigoplus_\pi a_{\lambda,\mu}^\pi V_\pi(H).
$$
Again, the goal is to get a positive 
\cc{\# P}-formula for the plethysm constant;
this is the gist of PH1 for the plethysm problem.

To apply the quantum group approach, we need a $q$-analogue of the
embedding $H \hookrightarrow G$.  Unfortunately, there is no such $q$-analogue
in the theory of standard quantum groups.
Because there is no nontrivial quantum group homomorphism from
 the standard quantum group $H_q=GL_q(\C^n)$ and to the standard 
quantum group  $G_q$.

\begin{thm}

\noindent (1) \cite{GCT4}:  Let $H$ and $G$ be as in the Kronecker problem.  Then there exists a quantum group $\hat{G}_q$ such that the homomorphism 
$H\rightarrow G$ can be quantized in the form
$H_q \hookrightarrow \hat{G}_q$.
  Furthermore, $\hat{G}_q$ has a unitary quantum subgroup $\hat{U}_q$ which corresponds to the maximal unitary subgroup $U \subseteq G$,
and a $q$-analogue of the Peter-Weyl theorem holds for $\hat G_q$.
The latter implies that every finite dimensional representation of $\hat G_q$ 
is completely decomposible into irreducibles.

\noindent (2) \cite{GCT7} There is an analogous (possibly singular)
quantum group $\hat{G}_q$ when $H$ and $G$ are as in the plethysm problem.
This also holds for general connected reductive (classical) $H$.
\end{thm}
Since the Kronecker problem is a special case of the (generalized)
plethysm problem,
the quantum group in GCT 4 is a special case of the quantum group in GCT 7.
The quantum group in the plethysm problem can be singular, i.e., its
determinant can vanish and hence the antipode need not exist. We still
call it a quantum group because its properties are very similar
to those of the standard quantum group; e.g. $q$-analogue of the
Peter-Weyl theorem, which allows $q$-harmonic analysis on these groups.
We call the quantum group $\hat G_q$  {\em nonstandard}, because 
though it is qualitatively similar to the standard (Drinfeld-Jimbo) 
quantum group $G_q$, it is also, as expected, fundamentally different.

The article \cite{GCT8} gives a conjecturally correct  algorithm to construct 
a canonical basis of an irreducible polynomial representation of $\hat G_q$ 
which generalizes the canonical basis for a polynomial
representation of the standard quantum group as per
 Kashiwara and Lusztig. It also gives a 
conjecturally correct algorithm to construct a canonical basis of
a certain $q$-deformation of the symmetric group algebra  $\C[S_r]$
which generalizes the Kazhdan-Lusztig basis \cite{kazhdan}
of the Hecke algebra (a standard
$q$-deformation of $\C[S_r]$). 
It  is shown in \cite{GCT7,GCT8} that PH1 for the Kronecker and
plethysm problems follows assuming that these canonical bases in the
nonstandard setting have properties akin to the ones in the standard 
setting. For a discussion on SH, see \cite{GCT6}.

\part{Invariant theory with a view towards GCT \\ 
 {\normalsize By Milind Sohoni}}

\newtheorem{theorem}{Theorem}

\newcommand{\ca}[1]{{\cal #1}}
\newcommand{\f}[2]{{\frac {#1} {#2}}}
\newcommand{\embed}[1]{{#1}^\phi}
\newcommand{\stab}{{\mbox {stab}}}
\newcommand{\codim}{{\mbox {codim}}}
\newcommand{\modulo}{{\mbox {mod}\ }}
\newcommand{\Y}{\mathbb{Y}}
\newcommand{\rel}{ \backslash}
\newcommand{\sym}{{\mbox {Sym}}}
\newcommand{\idealof}{{\mbox {ideal}}}
\newcommand{\trace}{{\mbox {trace}}}
\newcommand{\polylog}{{\mbox {polylog}}}
\newcommand{\sign}{{\mbox {sign}}}
\newcommand{\rank}{{\mbox {rank}}}
\newcommand{\formula}{{\mbox {Formula}}}
\newcommand{\circuit}{{\mbox {Circuit}}}
\newcommand{\core}{{\mbox {core}}}
\newcommand{\orbit}{{\mbox {orbit}}}
\newcommand{\cycle}{{\mbox {cycle}}}
\newcommand{\ideal}{{\mbox {ideal}}}
\newcommand{\wt}{{\mbox {wt}}}
\newcommand{\level}{{\mbox {level}}}
\newcommand{\vol}{{\mbox {vol}}}
\newcommand{\vect}{{\mbox {Vect}}}
\newcommand{\val}{{\mbox {wt}}}
\newcommand{\adm}{{\mbox {Adm}}}
\newcommand{\eval}{{\mbox {eval}}}
\newcommand{\invlim}{{\mbox {lim}_\leftarrow}}
\newcommand{\directlim}{{\mbox {lim}_\rightarrow}}
\newcommand{\sformal}{{\cal S}^{\mbox f}}
\newcommand{\san}{{\cal S}^{\mbox an}}
\newcommand{\Oan}{{\cal O}^{\mbox an}}
\newcommand{\vformal}{{\cal V}^{\mbox f}}
\newcommand{\diffl}{\bar {\cal D}}
\newcommand{\diff}{{\cal D}}
\newcommand{\invf}{{\cal F}^{-1}}
\newcommand{\invox}{{\cal O}_X^{-1}}
\newcommand{\tableau}{{\mbox {Tab}}}

\chapter{Finite Groups}

\noindent {\em References:} \cite{FulH,nagata}

\section{Generalities}
Let $V$ be a vector space over $\C$, and let $GL(V)$ denote the
group of all isomorphisms on $V$. For a fixed basis of $V$,
$GL(V)$ is isomorphic to the group $GL_n (\C)$, the group of all
$n \times n$ invertible matrices. 

Let $G$ be a group and $\rho :G \rightarrow GL(V)$ be a
representation. We also denote this by the tuple $(\rho ,V)$ or
say that $V$ is a $G$-module. 
Let $Z\subseteq V$ be a subspace such that $\rho (g)(Z)\subseteq
Z$ for all $g\in G$. Then, we say that $Z$ is an {\bf invariant
subspace}. We
say that $(\rho ,V)$ is {\bf irreducible} if there is no proper
subspace $W\subset V$ such that $\rho (g)(W) \subseteq W$ for all
$g\in G$. We say that $(\rho,V)$ is {\bf indecomposable} is there
is no expression $V=W_1 \oplus W_2 $ such that $\rho (g) (W_i
)\subseteq W_i $, for all $g\in G$. 

For a point $v\in V$, the {\bf orbit} $O(v)$, and the {\bf
stabilizer} $Stab(v)$ are defined as:
\[ \begin{array}{rcl}
  O(v)&=&\{ v' \in V| \exists g\in G \: with \: \rho (g)(v)=v' \}\\
Stab(v)&=&\{ g\in G| \rho (g)(v)=v \} \end{array}
\]
One may also define $v \sim v'$ if there is a $g\in G$ such
that $\rho (g)(v)=v'$. It is then easy to show that 
$[v]_{\sim } =O(v)$. 

Let $V^* $ be the dual-space of $V$. The representation $(\rho ,V)$
induces the {\bf dual} representation $(\rho^* ,V^*)$ 
defined as $\rho^* (v^*
)(v)=v^* (\rho (g^{-1} )(v)) $. It will be convenient for $\rho^*
$ to act on the right, i.e., $((v^* )(\rho^* ))(v)=v^* (\rho
(g^{-1})(v))$. 

When $\rho $ is fixed, we abbrieviate $\rho(g)(v)$ as just $g\cdot
v$. Along with this, there are the
standard constructions of the {\bf tensor} $T^d (V)$, 
the {\bf symmetric power} $Sym^d (V)$ and the {\bf alternating
power} $\wedge^d (V)$ representations. 

Of special significance is $Sym^d (V^* )$, the space of
homogeneous polynomial functions on $V$ of degree $d$. Let $dim(V)=n
$ and let $X_1
,\ldots ,X_n $ be a basis of $V^* $. We define
as follows:
\[ R= \C[X_1 ,\ldots ,X_n ]=\oplus_{d=0}^{\infty } R_d =
\oplus_{d=0}^{\infty}  Sym^d (V^* ) \]

Thus $R$ is the ring of all polynomial functions on $V$ and is
isomorphic to the algebra (over $\C$) of $n$ indeterminates. 
Since $G$ acts on the domain $V$, $G$ also acts on all functions
$f:V\rightarrow \C $ as follows:
\[ (f\cdot g)(v)=f(g^{-1}\cdot v) \]
This action of $G$ on all functions extends the action of $G$ on
polynomial functions above. 
Indeed, for any $g\in G$, the map $t_g :R \rightarrow R$ given by
$f\rightarrow f \cdot g$ is an algebra isomorphism. This is called
the {\bf translation map}.

For an $f\in R$, we say that $f$ is an {\bf invariant} if $f\cdot
g=f $ for all $g\in G$. The following are equivalent:
\begin{itemize}
\item $f\in R$ is an invariant.
\item $Stab(f)=G$.
\item $f(g\cdot v)=f(v)$ for all $g\in G$ and $v\in V$. 
\item For all $v,v'$ such that $v'\in Orbit(v)$, we have $f(v)=f(v')$. 
\end{itemize}

If $W_1 $ and $W_2 $ are two modules of $G$ and $\phi :W_1
\rightarrow W_2 $ is a linear map such that $g\cdot \phi (w_1
)=\phi (g \cdot w_1 )$ for all $g\in G$ and $w_1 \in W_1 $ then we
say that $\phi $ is {\bf $G$-equivariant} or that $\phi $ is a
{\bf morphism of $G$-modules}.

\section{The finite group action}
Let $G$ be a finite group and $(\mu ,W)$ be a representation. 

Recall that a {\bf complex inner product} on $W$ is a map $h:W
\times W \rightarrow \C$ such that:
\begin{itemize}
\item $h(\alpha w+\beta w', w'')=\overline{\alpha
}h(w,w'')+\overline{\beta }h(w',w'') $ for all $\alpha , \beta \in
\C$ and all $w,w',w''\in W$. 
\item $h(w'', \alpha w+\beta w')=\alpha
h(w'',w)+\beta h(w'',w') $ for all $\alpha , \beta \in
\C$ and all $w,w',w''\in W$. 
\item $h(w,w)>0 $ for all $w\neq 0$. 
\end{itemize}
Also recall that if $Z\subseteq W$ is a subspace, then $Z^{\bot }$
is defined as:
\[ Z^{\bot }=\{ w\in W| h(w,z)=0 \: \forall z\in Z \} \]
Also recall that $W=Z\oplus Z^{\bot} $. 

We say that an inner product $h$ is {\bf $G$-invariant} if 
$h(g \cdot w, g\cdot w')=h(w,w')$ for all $w,w'\in W$ and $g\in
G$.

\begin{prop}
Let $W$ be as above, and $Z$ be an invariant subspace of $W$. Then
$Z^{\bot }$ is also an invariant subspace. Thus every reducible
representation of $G$ is also decomposable. 
\end{prop}

\noindent
{\bf Proof}: Let $x\in Z^{\bot}, z\in Z$ and let us examine $(g\cdot
x,z)$. Applying $g^{-1}$ to both sides, we see that:
\[ h(g \cdot x,z)=h(g^{-1} \cdot g \cdot x, g^{-1} \cdot
z)=h(x,g^{-1} \cdot z)=0 \]
Thus, $G$ preserves $Z^{\bot }$ as claimed. $\Box $

Let $h$ be a complex inner product on $W$. We define the inner
product $h^G $ as follows: 
\[ h^G (w,w')=\frac{1}{|G|} \sum_{g' \in G} h(g'\cdot w, g'\cdot w')
\]

\begin{lemma}
$h^G $ is a $G$-invariant inner product.
\end{lemma}

\noindent
{\bf Proof}: First we see that 
\[ h^G (w,w)=\frac{1}{|G|} \sum_{g'\in G} h(w',w')\]
where $w'=g'\cdot w$. Thus $h^G (w,w)>0$ unless $w=0$. Secondly,
by the linearity of the action of $G$, we see that $h^G $ is
indeed an inner product. Finally, we see that:
\[ h^G (g\cdot w,g\cdot w')=\frac{1}{|G|} \sum_{g'\in G} 
h(g'\cdot g\cdot w, g' \cdot g \cdot w') \]
Since as $g'$ ranges over $G$, so will $g'\cdot g$ for any fixed
$g$, we have that $h^G $ is $G$-invariant. $\Box $

\begin{theorem}
\begin{itemize}
\item Let $G$ be a finite group and $(\rho ,V)$ be an indecomposable
representation, then it is also irreducible. 
\item Every representation $(\rho ,V)$ may be decomposed into
irreducible representations $V_i $. Thus $V=\oplus_i V_i $, where
$(\rho_i ,V_i )$ is an irreducible representation.
\end{itemize}
\end{theorem}

\noindent
{\bf Proof}: Suppose that $Z\subseteq V$ is an invariant subspace,
then $ V=Z\oplus Z^{\bot }$ is a non-trivial decomposition of $V$
contradicting the hypothesis. The second part is proved by
applying the first, recursively. $\Box $

We have seen the operation of {\bf averaging over the group} in
going from the inner product $h$ to the $G$-invariant inner
product $h^G $. A similar approach may be used for constructing
invariant polynomials functions. So let $p(X)\in R=\C[X_1 ,\ldots
,X_n ]$ be a polynomial function. We define the function $p^G
:V\rightarrow \C$ as:
\[ p^G (v)=\frac{1}{|G|} \sum_{g\in G} p(g \cdot v) \]
The transition from $p$ to $p^G $ is called the {\bf Reynold's
operator}. 

\begin{prop}
Let $p\in R$ be of degree atmost $d$, then $p^G $ is also a
polynomial function of degree atmost $d$. Next, $p^G $ is an
invariant. 
\end{prop}

Let $R^G $ denote the set of all invariant polynomial 
functions on the space $V$. It is easy to see that $R^G \subseteq R$ is
actually a {\em subring} of $R$. 

Let $Z\subseteq V$ be an arbitrary subset of $V$. We say that $Z$
is $G$-closed if $g\cdot z\in Z$ for all $g\in G$ and $z\in Z$. 
Thus $Z$ is a union of orbits of points in $V$.

\begin{lemma}
Let $p\in R^G $ be an invariant and let $Z=V(p)$ be the variety of
$p$. Then $Z$ is $G$-closed.
\end{lemma}

We have already seen that $O(v)$, the orbit of $v$ arises from the
equivalence class $\sim $ on $V$. Since the group is finite,
$|O(v)|\leq |G|$ for any $v$. Let $O_1 $ and $O_2 $ be disjoint
orbits.  It is essential to determine if elements of $R^G $ can
separate $O_1 $ and $O_2 $. 

\begin{lemma}
Let $O_1 $ and $O_2 $ be as above, and $I_1 $ and $I_2 $ be their
ideals in $R$. Then there are $p_1 \in I_1 $ and $p_2 \in I_2 $ so
that $p_1 +p_2 =1$. 
\end{lemma}

\noindent
{\bf Proof}: This follows from the Hilbert Nullstellensatz. Since
the point sets are finite, there is an explicit construction based
on Lagrange interpolation. $\Box $

Let $G$ be a finite group and $(\rho ,V)$ be a representation as
above. We
have see that this induces an action on $\C [X_1 ,\ldots ,X_n ]$.
Also note that this action is {\bf homogeneous}: for a $g \in G$
and $p\in R_d $, we have that $p\cdot g \in R_d $ as well. Thus
$R^G $, the ring of invariants, is a {\em homogeneous subring} of
$R$. In other words:
\[ R^G =\oplus_{d=0}^{\infty} R^G_d \]
where $R^G_d $ are invariants which are homogeneous of degree $d$.
The existence of the above decomposition implies that every
invariant is a sum of homogeneous invariants. Now, since $R^G_d
\subseteq R_d $ as a vector space over $\C$. Thus 
\[ dim_{\C} (R^G_d
) \leq dim_{\C}(R_d ) \leq {{n+d-1}\choose{n-1}} \]

We define the {\bf hilbert function} $h(R^G )$ of  $R^G $ (or for
that matter, of any homogeneous ring) as:
\[ h(R^G )=\sum_{d=0}^{\infty}  dim_{\C} (R^G_d ) z^d \]

We will see now that $h(R^G )$ is actually a rational function
which is easily computed. We need a lemma.

Let $(\rho ,W)$ be a representation of the finite group $G$. Let
\[ W^G =\{ w\in W| g\cdot w=w \} \]
be the set of all vector invariants in $W$, and this is a subspace
of $W$. . 

\begin{lemma}
Let $(\rho ,W)$ be as above. We have:
\[ dim_{\C }(W^G )=\frac{1}{|G|} \sum_{g \in G} trace(\rho (g)) \]
\end{lemma}

\noindent
{\bf Proof}: Define $P=\frac{1}{|G|} \sum_{g \in G} \rho (g)$, as
the average of the representation matrices. We see that $\rho
(g)\cdot P=P \cdot \rho (g)$ and that $P^2 =P$. 
Thus $P$ is diagonalizable and the eigenvalues of $P$ are 
in the set $\{ 1,0 \}$. Let $W_1 $ and $W_0 $ be the corresponding
eigen-spaces. It is clear that $W^G \subseteq W_1 $ and that $W_1
$ is fixed by each $g\in G$. We now argue that every $w\in W_1 $
is actually an invariant. For that, let $w_g =g \cdot w$. We then
have that $Pw=w$ implies that 
\[ w=\frac{1}{|G|} \sum_{g \in G} w_g \]
Note that a change-of-basis does not affect the hypothesis nor the
assertion. We may thus assume that 
each $\rho (g)$ is unitary, we have that $w_g =w$ for all
$g\in G$. Now, the claim follows by computing $trace(P)$. $\Box $

We are now ready to state {\bf Molien's Theorem}:
\begin{theorem}
Let $(\rho ,W)$ be as above. We have:
\[ h(R^G )=\frac{1}{|G|} \sum_{g \in G} \frac{1}{det(I-z\rho (g))}
\]
\end{theorem}

\noindent
{\bf Proof}: Let $dim_{\C}(W)=n$ and let $\{ X_1 ,\ldots ,X_n \}$
be a basis of $W^* $. 
Since $R^G =\sum_d R^G_d $ and each $R^G_d \subseteq
\C[X_1 ,\ldots ,X_n ]_d $. Note that each $C[X_1 ,\ldots ,X_n ]_d
$ is also a representation $\rho_d $ of $G$. Furthermore, it is
easy to see that if $\{ \lambda_1 ,\ldots ,\lambda_n \}$ are the
eigenvalues of $\rho (g)$, then the eigen-values of the matrix
$\rho_d (g)$ are precisely (including multiplicity) 
\[ \{ \prod_i \lambda_i^{d_i }| \sum_i d_i =d \} \] 
Thus 
\[ trace(\rho_d (g))=\sum_{\overline{d}: |\overline{d}|=d} 
\prod_i \lambda_i^{d_i } \]

We then have:
\[ \begin{array}{rcl}
h(R^G )&=&\sum_d z^d dim_{\C} (R^G_d ) \\
&=& \sum_d z^d [\frac{1}{|G|} \sum_g trace(\rho_d (g)) ]\\
&=& \frac{1}{|G|} \sum_g \frac{1}{(1-\lambda_1 (g) z)\ldots
(1-\lambda_n (g) z)} \\
&=& \frac{1}{|G|} \sum_g \frac{1}{det(I-z\rho(g))} \end{array} \]

This proves the theorem. $\Box $

\section{The Symmetric Group}

$S_n $ will denote the symmetric group of all bijections on the
set $[n]$. The standard representation of $S_n $ is obviously on
$V=\C^n $ with 
\[ \sigma \cdot (v_1 ,\ldots ,v_n )=(v_{\sigma (1)},\ldots
,v_{\sigma (n)} )\]

Thus, regarding $V$ as column vectors, and $S_n $ as the group of
$n\times n$-permutation matrices, we see that the action of
permutation $P$ on vector $v$ is given by the matrix
multiplication $P\cdot v$. 

Let $X_1 ,\ldots ,X_n $ be a basis of $V^* $. $S_n $ acts on
$R=\C[X_1 ,\ldots ,X_n ]$ by $X_i \cdot \sigma =X_{\sigma (i)} $. 
The orbit of any point $v=(v_1 ,\ldots ,v_n )$ is the collection of
all permutation of the entries of the vector $v$ and thus the size
of the orbit is bounded by $n!$. 

The invariants for this action are given by the elementary
symmetric polynomials $e_k (X)$, for $k=1,\ldots ,n$, where 
\[ e_k (X)=\sum_{i_1 <i_2 <\ldots <i_k } X_{i_1} X_{i_2} \ldots
X_{i_k } \]

Given two vector $v$ and $w$, if $w\not \in O(v)$, then there is a
$k$ such that $e_k (v) \neq e_k (w)$. This follows from the theory
of equations in one variable. 

The ring $R^G $ equals $\C [e_1 ,\ldots ,e_n ]$ and has no
algebraic dependencies. The hilbert function of $R_G $ may then be
expressed as:
\[ h(R^G )=\frac{1}{(1-z)(1-z^2 )\ldots (1-z^n )} \]
It is an exercise to verify that Molien's expression agrees with
the above.

A related action of $S_n $ is the {\bf diagonal action}:
Let $X=\{ X_1 ,\ldots ,X_n \}$, $Y=\{ Y_1 ,\ldots ,Y_n \}$, and so
on upto $W=\{ W_1 ,\ldots ,W_n \}$ be a family of $r$ (disjoint) 
variables. Let $B=\C [X,Y,\ldots ,W]$ be the ring of polynomials in the
variables of the disjoint union. 

We define the action of $S_n $ on $X\cup Y\cup \ldots \cup W$ as 
$X_i \cdot \sigma =X_{\sigma (i)} $, 
$Y_i \cdot \sigma =Y_{\sigma (i)} $, and so on. 

The matrix equivalence of this action is the action of the
permutation matrices on $n\times r$ matrices $A$, where the action
of $P$ on $A$ is given by $P\cdot A$. 

The invariants $B^G $ is obtained from the $r=1$ case by a {\em
curious} operation: Let $D_{XY}$ denote the operator:
\[ D_{XY}=Y_1 \frac{\partial}{\partial X_1 } +
\ldots +Y_n \frac{\partial}{\partial X_n } \]
The ring $B^G $ is obtained from $R^G $ by applying the operators
$D_{XY},D_{XW},D_{WX}$ and so on, to elements of $R^G $. As an
example, we have
\[ e_2 (X)=X_1 X_2 +X_1 X_3 +\ldots +X_{n-1} X_n \]
We have $D_{XY}(e_2 )$ as:
\[ D_{XY} (e_2 (X))=\sum_{i \neq j} X_i Y_j  \]
This is clearly an element of $B^G $.
 
\chapter{The Group $SL_n $}

\noindent {\em References:} \cite{FulH,nagata}

\section{The Canonical Representation}

Let $V$ be a vector space of dimension $n$, and let $x_1 ,\ldots
,x_n $ be a basis for $V$. Let $X_1 ,\ldots ,X_n $ be the dual
basis of $V^* $. 

$SL(V)$ will denote the group of all unimodular linear
transformations on $V$. In the above basis, this group is
isomorphic to that of all $n\times n$ matrices of determinant $1$,
or in other words $SL_n (\C )$. 

The standard representation of $SL(V)$ is obviously $V$ itself:
Given $\phi \in SL(V)$ and $v\in V$, we have $\phi \cdot v=\phi(v)$ 
is the action of $\phi $ on $v$.

In terms of the basis $x$ above we may write $v=[x_1 ,\ldots
,x_n][\alpha_1 ,\ldots ,\alpha_n ]^T $ and thus $\phi \cdot v$ as 
$[\phi \cdot x_1 ,\ldots ,\phi \cdot x_n ]
[\alpha_1 ,\ldots ,\alpha_n ]^T$. If $\phi \cdot x_i
=[x_1 ,\ldots ,x_n ][a_{1i} ,\ldots ,a_{ni}]^T $, then we have 
\[ \phi \cdot v=[x_1 ,\ldots ,x_n ]\left[ \begin{array}{ccc} 
a_{11} & \ldots & a_{1n} \\
\vdots & & \vdots \\
a_{n1} & \ldots & a_{nn} \end{array} \right] \left[
\begin{array}{c} 
\alpha_1 \\
\vdots \\
\alpha_n \end{array} \right] \]
We denote this matrix as $A_{\phi }$. 

We may now work with $SL_n (\C)$ or simply $SL_n $. Given a vector
$a=[\alpha_1 ,\ldots ,\alpha_n ]^T $, we see that the matrix
multiplication $A\cdot a$ is
the action of $A$ on the column vector $a$. 

Let us now understand the orbits of typical elements in the column
space $\C^n $. The typical $v\in \C^n $ is a non-zero column
vector. For any non-zero vector $w$, we see that there is an
element $A\in SL_n $ such that $w=Av$. Furthermore, for any $B\in
SL_n $, clearly $Bv\neq 0$. Thus we see that $\C^n $ has exactly
two orbits:
\[ O_0 =\{ 0\} \: \: \: O_1 \{ v \in \C^n | v\neq 0 \} \]
Note that $O_1 $ is dense in $V=\C^n $ and its closure includes the
orbit $O_0 $ and therefore, the whole of $V$. 

Let $R=\C[X_1 ,\ldots X_n ]$ be the
ring of polynomial functions on $\C^n $. We examine the action of
$SL_n $ on $\C [X]$. Recall that the action of $A$ on $X$ should be
such that the evaluations $X_i (x_j )=\delta_{ij}$ must be
conserved. Thus if the column vectors $x_i  =[0,\ldots ,0,1,0 \ldots
,0]^T $ are the basis vectors of $V$ and the row vectors 
$X_i =[0,\ldots ,0,1,0\ldots ,0]$ that of $V^* $, then a matrix
$A\in SL_n $ transforms $x_i $ to $Ax_i $ and $X_j $ to $X_j
A^{-1}$. Thus $X_j (x_i )=X_j /cdot x_i $ goes to $X_j A^{-1}A x_i
=X_j \cdot x_i $. 

Next, we examine $R=\C [X]$ for invariants. First note that the
action of $SL_n $ is homogeneous and thus we may assume that an
invariant $p$ is actually homogeneous of degree $d$. Next 
we see that $p$ must be constant on $O_1 $ and $O_0 $. In this  
case, if $p(O_1 )=\alpha $ then by the density of $O_1 $ in $V$,
we see that $p$ is actually constant on $V$. Thus $R^G =R_0 =\C $.

\section{The Diagonal Representation}

Let us now consider the diagonal representation of the above
representation. In other words, let $V^r $ be the space of all
complex $n\times r$ matrices $x$. The action of $A$ on $x$ is
obvious $A \cdot x$. Let $X$ be the $r\times n$-matrix dual to
$x$. It transforms according to $X\rightarrow XA^{-1} $. 

Let us examine the case when $r<n$ and compute the orbits in $V^r
$. Let $x\in V^r $ be a matrix and $y$ be a column vector such
that $x\cdot y=0$. We see that $(A \cdot x) \cdot y=0$ as well.
Thus if $ann(x)=\{ y \in \C^r | x\cdot y =0 \}$ is the annihilator
space of $x$, then we see that $ann(x)=ann(A \cdot x)$. 

We show that the converse is also true:
\begin{prop}
Let $r<n$ and $x,x'\in V^r $ be such that $ann(x)=ann(x')$. Then
there is an $A \in SL_n $ such that $x=A\cdot x'$. 
\end{prop}

\noindent
{\bf Proof}: Use the row-echelon form construction. Make the
pivots as $1$ using appropriate diagonal matrices. Since $r<n$
these can be chosen from $SL_n $. $\Box $

This decides the orbit structure of $V^r $ for $r<n$. The largest
orbit $O_r $ is of course when $ann(x)=0$. Thus, when $x\in V_r $
is a mtrix of rank $r$, we see that $ann(x)=0$ is trivial. The
{\em generic} element of $V^r $ is of this form. Thus $O_r $ is
dense in $V^r $. For this reason, there are no non-constant
invariants. 

Another calculation is the computation of the closure of orbits.
Let $O,O'$ be two arbitrary orbits of $V^r $. We have:
\begin{prop}
$O'$ lies in the closure of $O$ if and only if $ann(O')\supseteq
ann(O)$.
\end{prop}

\noindent
{\bf Proof}: One direction is clear. We prove the other direction
when $ann(O)\subseteq ann(O')$ and $dim(ann(O))=dim(ann(O'))-1$. 
Thus, we may assume that $O=[x]$ and $O'=[x']$ with both $x$ and
$x'$ such that $rowspace(x)\supseteq
rowspace(x')$ with $rank(x')=rank(x)-1$. Then, upto $SL_n $, we
may assume that the rows $x'[1], \ldots ,x'[k]$ match the first
$k$ rows of $x$, and that $x'[k+1]=x'[k+2]=\ldots =x'[n]=0$. Note
that $x[k+1]$ is non-zero and $x[k+2]$ exists and is zero. We
then construct the matrix $A(t)$ as follows:
\[ A(t)=\left[ \begin{array}{c|cc|c}
    I_{k} & 0 & 0 & 0 \\ \hline 
    0& t & 0 & 0 \\
    0 & 0 & t^{-1} & 0 \\ \hline
    0 & 0 & 0 & I_{n-k-2} \end{array} \right] \]
We see that $A(t)\in SL_n $ for all $t\neq 0$. Next, if we let
$x(t)=A(t)\cdot x$, then we see that:
\[ \lim_{t\rightarrow 0} x(t)=x' \]
This shows that $x'$ lies in the closure of $O$. $\Box $

We now look at the case when $r\geq n$. We have:
\begin{prop}
Let $r\geq n$ and $x,x'\in V^r $ be such that $ann(x)=ann(x')$. 
If (i) $rank(x)<n$
then there is an $A\in SL_n $ such that $x'=Ax$,  (ii) if $rank(x)=n$
there is a unique $A\in SL_n $ and a $\lambda \in \C^* $ such that
if $z=A\cdot x$, then the first $n-1$ rows of $z$ equal those of
$x'$ and $z[n]=\lambda x'[n]$. 
\end{prop}

The proof is easy.
Let $x$ be  matrix in $V^r $ of rank $n$, and $C$ be a subset of
$[r]=\{ 1,2,\ldots ,r\}$ such that $det(x[C])\neq 0$. 

\begin{prop}
Let $x$ be as above. Then $O(x)$, the orbit of $x$, equals all
points $x'\in V^r $ such that (i) $ann(x)=ann(x')$ and (ii)
$det(X'[C])=det(x[C])$. The set of all rank $n$ points in $V^r $
is dense in $V^r $. 
\end{prop}

The proof is easy. 
\begin{prop}
The orbit $O(x)$ as above is closed.
\end{prop}

\noindent
{\bf Proof}: Notice that if $A\in SL_n $ and $z=Ax$ then
$det(x[C])=det(z[C])$. We may rewrite condition (i) above as
$ann(x')\supseteq ann(x)$. Condition (ii) ensures that
$rank(x')=n$ and that condition (i) holds with equality. Thus is
$y_1 ,\ldots ,y_{r-n}$ are column vectors generating $ann(x)$,
then the equations $x'y_s =0$ and $det(x'[C])=det(x[C])$
determines the orbit $O(x)$. Thus $O(x)$ is the zero-set of some
algebraic equations and thus is closed. $\Box $

\begin{prop}
Let $x\in V^r $ be such that $rank(x)<n$. Then (i) $O(x)$ equals
those $x'$ such that $ann(x)=ann(x')$, (ii) $O(x)$ is not closed
and its closure $\overline{O(x)}$ equals points $x'$ such that
$ann(x')\supseteq ann(x)$. 
\end{prop}

We now move to the computation of invariants. The space $\C[V^r ]$
equals the space of all polynomials in the variable matrix
$X=(X_{ij})$ where $i=1,\ldots ,n$ and $j=1,\ldots ,r$. It is
clear that for any set $C$ of $n$ columns of $X$, we see that
$det(X[C])$ is an invariant. We will denote $C$ as $C=c_1 <c_2
<\ldots <c_n $ and $det(X[C])$ as $p_C $, the $C$-th {\bf Plucker}
coordinate. We aim to show now that these are the only invariants.

Let $C_0 =\{ 1<2<\ldots <n\}$ and 
$W\subseteq V^r $ be the space of all matrices $x\in V^r $
such that $x[C_0 ]=diag(1,\ldots ,1,\lambda )$. 
Let $W'$ be those elements of $W$ for which $\lambda \neq 0$. 

\begin{lemma}
Let $W'$ be as above. (i) If $x\in W'$, then $O(x) \cap W'=x$, (ii) 
for any $x\in V^r $ such that $det(x[C_0 ])\neq 0$, there is 
a unique $A\in SL_n $ such that $Ax\in W'$. 
\end{lemma}

Let us call $Z' \subseteq V^r $ as those $x$ such that $det(x[C_0
])\neq 0$. We then have the projection map 
\[ \pi :Z' \rightarrow W' \]
given by the above lemma. Note that $Z'$ is $SL_n$-invariant: if
$x\in Z'$ and $A\in SL_n $ then $A\cdot x \in Z'$ as well.

The ring $\C[Z']$ of regular functions on $Z'$ is 
precisely $\C[X]_{det([X[C_0 ])} $, the localization of $\C [X]$
at $det(X[C_0 ])$. 

We may parametrize $W' $ as:
\[ W=\left[ 
\begin{array}{ccc|ccc}
1 & \ldots & 0 & w_{1,n+1} & \ldots & w_{1,r} \\
\vdots & & \vdots & \vdots & & \vdots \\
0 & \ldots & w_{n,n} & w_{n,n+1} & \ldots & w_{n,r} \end{array}
\right] \]
Let 
\[ W=\{ W_{i,j} | i=1,\ldots ,n , \: j=n+1,\ldots ,r \} \cup \{
W_{n,n} \} \]
be a set of $(n-r)\cdot r +1 $ variables. 
The ring $\C[W']$ of regular function on $W'$ is precisely
$\C[W]_{W_{n,n}}$. 

Let $x\in Z'$ be an arbitrary point and $A=x[C_0 ]$. 
Let $C_{i,j}$ be the set $C_0 -i+j $. The map $\pi $ is given by:
\[ \begin{array}{rcl}
\pi (x)_{n,n}&=&det(A)\\
\pi (x)_{i,j}&=& \left\{ \begin{array}{ll}
			 det(x[C_{i,j} ])/det(A) & \: for \: i\neq
                                                 n \\
			 det(x[C_{i,j} ]) & \: otherwise
			 \end{array} \right. \\
			 \end{array} \]
The map $\pi $ causes the map:
\[ \pi^* : \C[W'] \rightarrow \C[Z'] \] 
given by:
\[ \begin{array}{rcl}
\pi^* (W_{n,n})&=& det(X[C_0 ]) \\
\pi^* (W_{i,j})&=& \left\{ \begin{array}{ll}
			 det(X[C_{i,j} ])/det(X[C_0 ]) & \: for \: i\neq
                                                 n \\
			 det(X[C_{i,j} ]) & \: otherwise
			 \end{array} \right. \\
			 \end{array} \]

Now we note that $Z'$ is dense in $V^r $. Let $p\in \C[X]^{SL_n }$
be an invariant. Clearly $p$ restricted to $W'$ defines an element
$p_{W'} $ of $\C[W']$. This then extends to $Z'$ via $\pi ^* $.
Clearly $\pi^* (p_{W'})$ must match $p$ on $Z'$. Thus we have that
$p$ is a polynomial in $det(X[C_{i,j}])$ possibly localized at
$det(X[C_0 ])$. 

Note that for a general $C$, $det(X[C])$ is already expressible in
$det(X[C_0 ])$ and $det(X[C_{i,j}])$. 

\section{Other Representations}
We discuss two other representations:

\noindent
{\bf The Conjugate Action}: 
Let ${\cal M}$ be the space of all $n\times n$ matrices with complex
entries. we define the action of $A\in SL_n $ on ${\cal M}$ as 
follows. For an $M\in {\cal M}$, $A$ acts on it by conjugation.   
\[ A\cdot M= AMA^{-1} \]
Note that $dim_{\C}({\cal M})=n^2 $. Let $X=\{ X_{ij}|1\leq
i,j\leq n \}$ be the dual space to ${\cal M}$. The invariants for
this action are $Tr(X), \ldots ,Tr(X^i ) ,\ldots ,Tr(X^n )$. 
That these are invariants ic clear, for $Tr(AM^i A^{-1})=Tr(M)$
for any $A,M$. 

Also note that $Tr(X)=X_{11}+\ldots +X_{nn} $ is linear in the
$X$'s. This indicates an invariant hyperplane in ${\cal M}$. This
is precisely the {\em trace zero matrices}. Thus ${\cal M}={\cal
M}_0 \oplus{\cal M}_1 $ where ${\cal M}_0 $ are all matrices $M$
such that $Tr(M)=0$. The one-dimensional complementary space 
${\cal M}_1 $ is
composed of multiples of the identity matrix. 

The orbits are parametrized by the {\bf Jordan canonical form}
$JCF(M)$. In other words, if $JCF(M)=JCF(M')$ them $M'\in O(M)$.
Furthermore, if $JCF(M)$ is diagonal then the orbit of $M$ is {\bf
closed} in ${\cal M}$.

\noindent
{\bf The Space of Forms}: Let $V$ be a complex vector space of
dimension $n$ and$V^* $ its dual. Let $X_1 ,\ldots ,X_n $ be the
dual basis. The space $Sym^d (V^* )$ consists 
of degree $d$ forms are formal linear
combinations of the monomials $X_1^{i_1 }\ldots X_n^{i_n }$ with
$i_1 +\ldots +i_n =d$. The typical form may be written as:
\[ f(X)=\sum_{i_1 +\ldots +i_n =d} B_{i_1 \ldots i_n } X_1^{i_1}
\ldots X_n^{i_n } \]
Let $A\in SL_n $ be such that 
\[ A^{-1}=\left[ \begin{array}{ccc} 
a_{11} & \ldots &a_{n1} \\
\vdots & & \vdots \\
a_{1n} & \ldots & a_{nn} \end{array} \right] \]
The space $Sym^d (V^* )$ is the formal linear combination of the
generaic coefficients $\{ B_{\overline{i}} | |\overline{i}|=d \}$.
The action of $A$ is obtained by substituting 
\[ X_i \rightarrow \sum_j a_{ij} X_i \]
in $f(X)$ and recoomputing the coefficients. we illustrate this
for $n=2$ and $d=2$. Then, the generic form is given by:
\[ f(X_1 ,X_2 )=B_{20}X_1^2 +B_{11}X_1^1 X_2^1 +B_{02} X_2^2 \]
Upon substitution, we get $f(a_{11}X_1 +a_{12}X_2 ,a_{21}X_1
+a_{22} X_2 )$:
Thus the new coefficients are:
\[ \begin{array}{cclll}
B_{20} & \rightarrow & B_{20} a_{11}^2 + &B_{11} a_{11}a_{21} +&B_{02}
a_{21}^2 \\
B_{11} & \rightarrow & 
B_{20} 2a_{11}a_{12} +&B_{11}(a_{11}a_{22}+ a_{12}a_{21}) +&B_{02} 
2a_{21}a_{22} \\
B_{02} & \rightarrow & B_{20} a_{12}^2 + &B_{11}a_{12}a_{22} +
&B_{02}a_{22}^2  \\
\end{array} \]

We thus see that the variables $B$ move linearly, with
coefficients as homogeneous polynomials of degree $2$ in the
coefficients of the group elements. In the above case, we know
that the {\em discriminant} $B_{11}^2 -4B_{20} B_{02}$ is an
invariant, 

These spaces have been the subject of intense analysis and their
study by Gordan, Hilbert and other heralded the beginning of
commutative algebra and invariant theory.

\section{Full Reducibility}

Let $W$ be a representation of $SL_n $ and let $Z\subseteq W$ be
an invariant subspace. The reducibility question is whether there
exists a complement $Z'$ such that $W=Z\oplus Z'$ and $Z'$ is $SL_n
$-invariant as well. 

The above result is indeed true although we will not prove it
here. There are many proofs known, each with a specific objective
in a specific situation, and each extremely instructive. 

The simplest is possibly
through the Weyl Unitary trick. In this, a suitable {\bf compact}
subgroup $U\subseteq SL_n $ is chosen. The theory of compact
groups is much like that of finite groups and a complement $Z'$ may
easily be found. It is then shown that $Z'$ is  $SL_n$-invariant.

The second attack is through showing the fill reducibility of the
module $\otimes^d V$, the $d$-th tensor product representation of
$V$. This goes through the construction of the commutator of the
same module regarded as a {\em right} $S_d $-module, with the
symmetric group permutating the contents of the $d$ positions. The
full reducibility then follows from Maschke's theorem and that of
the finite group $S_d $. 

The oldest approach was through the construction of a symbolic
reynold's operator, which is the {\bf Cayley} $\Omega $-process.
Let $\C[W]$ be the ring of polynomial functions on the space $W$.
The operator $\Omega $ is a map:
\[ \Omega: \C[W] \rightarrow \C[W]^{SL_n } \]
such that if $p$ is an invariant then $\Omega (p\cdot p')=p \cdot
\Omega
(p')$. The definition of $\Omega $ is nothing but {\bf curious}.
 
\chapter{Invariant Theory}

\noindent {\em References:} \cite{FulH,nagata}

\section{Algebraic Groups and affine actions}

An algebraic group (over $\C$) is an affine variety $G$ equipped
with (i) the group product, i.e., a morphism of algebraic
varieties $\cdot: G\times G \rightarrow G$ which is associative, 
i.e. $(g_1 \cdot g_2 )\cdot g_3 =g_1 \cdot (g_2 \cdot g_3 )$ for
all $g_1 ,g_2 ,g_3 \in G$, (ii) and
algebraic inverse $i:G \rightarrow G$, i.e, $g \cdot
i(g)=i(g)\cdot g=1_G $, where $1_G $ is (iii) a special element
$1_G \in G$, which functions as the identity, i.e., $1_G \cdot g
=g \cdot 1_G =g$ for all $g\in G$. 
Let $\C [G]$ be 
the ring of regular functions on $G$. The requirements (i)-(iii)
above are via morphisms of algebraic varities, and thus the group
product and the inverse must be defined algebraically. Thus the
product $\cdot :G \times G \rightarrow G$ results in a morphism of
$\C$-algebras $\cdot^* :\C[G] \rightarrow \C[G] \otimes \C[G] $
and the inverse into another morphism $i^* :\C[G]\rightarrow
\C[G]$. 

The essential example is obviously $SL_n $. Clearly for $G=SL_n $,
we have $\C[G]=\C[X]/(det(X)-1)$, where $X$ is the $n\times n$
indeterminate matrix. The morphism $\cdot^* $ and $i^*$ are clearly:
\[ \begin{array}{rcl}
  \cdot^* ( X_{ij}) &=& \sum_k X_{ik} \otimes X_{kj} \\
   i^* (X_{ij})&=&det(M_{ji}) \end{array} \]
where $M_{ji}$ is the corresponding minor of $X$. 

Next, let $Z$ be another affine variety with $\C[Z]$ as its ring
of regular functions. We say that $Z$ is a {\bf $G$-variety} if there is
a morphism $\mu: G \times Z \rightarrow Z$ which is a group action.
Thus, not only must $G$ act on $Z$, it must do so {\em
algebraically}. 
This $\mu $ induces the map $\mu^* $:
\[ \mu^* : \C[Z] \rightarrow \C[G] \otimes \C[Z] \]
Thus every function $f\in \C[Z]$ goes to a finite sum:
\[ \mu^* (f)=\sum_{i=1}^k h_i \otimes f_i \]
where $f_i \in \C[Z]$ and $h_i \in \C[G]$ for all $i$. 

To continue with our example, consider $SL_2 $ and $Z=Sym^d (V^* )$.
$\C[Z]=\C[B_{20}, B_{11},B_{02} ]$ and
$\C[G]=\C[A_{11},A_{12},A_{21},A_{22}]/(A_{11}A_{22}-A_{21}A_{12}-1)$.
We have for example:
\[\mu^* (B_{20}) =A_{11}^2 \otimes B_{20} +A_{11}A_{21} \otimes
B_{11} +A_{21}^2 \otimes B_{02} \]

Let $\mu: G \times Z \rightarrow Z$ and $\mu' : G \times Z'
\rightarrow Z'$ be two $G$-varities and let $\phi :Z \rightarrow
Z'$ be a morphism. We say that $\phi $ is {\bf $G$-equivariant} if
$\phi (\mu (g,z))=\mu'(g, \phi (z))$ for all $g\in G$ and $z\in
Z$. Thus $\phi $ commutes with the action of $G$. 

\section{Orbits and Invariants}
Every $g\in G$ induces an algebraic bijection on $Z$ by
restricting the map $\mu :G \times Z \rightarrow Z$ to $g$. We
denote this map by $\mu (g)$ and call it {\bf translation by $g$}. 
The map:
\[ \mu (g) :Z \rightarrow Z \]
induces the isomorphism of $\C$-algebras:
\[ \mu (g)^* : \C[Z] \rightarrow \C[Z] \]
Given any function $f\in \C[Z] $, $\mu (g)^* (f) \in \C[Z]$ is the 
{\em translated} function and denotes the action of $G$ on $f$. 
This makes $\C[Z]$ into a $G$-module.
If $\phi :Z \rightarrow Z'$ is a $G$-equivariant map, then the map
$\phi^* :\C [Z'] \rightarrow \C[Z]$ is also $G$-equivariant, for
the $G$ action on $\C[Z]$ and $\C[Z']$ as above.

We next examine the equation: 
\[ \mu^* (f)=\sum_{i=1}^k h_i \otimes f_i \]
where $f_i \in \C[Z]$ and $h_i \in \C[G]$ for all $i$. 
For a fixed $g$, we see that
\[ \mu (g)^* (f)=\sum_{i=1}^k h_i (g) \otimes f_i \]
Thus we see that every translate of $f$ lies in the
$k$-dimensional vector space $\C\cdot \{ f_1 ,\ldots ,f_k \}$. 
Let
\[ M(f)=\C \cdot\{ \mu (g)^* (f) |g \in G\} \subseteq \C \cdot \{
f_1 ,\ldots ,f_k \} \]
Clearly, $M(f)$ is a $G$-invariant subspace of $\C [Z]$.  This may
be generalized:
\begin{prop}
Let $S=\{ s_1 ,\ldots ,s_m \}$ be a finite subset of $\C[Z]$. Then there
is a finite-dimensional $G$-invariant subspace $M(S)$ of $\C[Z]$
containing $s_1 ,\ldots ,s_m$. 
\end{prop}

Next, let us consider $\mu :G \times Z \rightarrow Z$ and fix a
$z\in Z$. We get the map $\mu_z :G \rightarrow X$. Since $G$ is an
affine variety, we see that the image $\mu_z (G)$ is a {\bf
constructible set}, whose closure is an affine variety. The image
is precisely the orbit $O(z)\subseteq Z$. The closure of $O(z)$
will be denoted by $\Delta (z)$. 

If $O(z)=\Delta (z)$, then we have the $G$-equivariant embedding 
$i_z :O(z) \rightarrow Z$. This gives us the map $i_z^*
:\C[Z]\rightarrow \C[O(z)]$ which is a surjection. Thus there is
an ideal $I(z)=ker(i_z^* )$ such that $\C[O(z)]\cong \C[Z]/I(z)$.
Since the map $i_z $ is $G$-equivariant, $I(z)$ is a $G$-submodule
of $\C[Z]$ and is the ideal of definition of the orbit $O(z)$. 

In general, if $I\subseteq \C[Z]$ is an ideal which is
$G$-invariant, then the variety of $I$ is also $G$-invariant and
is the union of orbits.

The second construction that we make is that of the {\bf quotient}
$Z/G$. Since $\C[Z]$ is a $G$-module, we examine $\C[Z]^G $, the
subring of $G$-invariant functions in $\C[Z]$. We define $Z/G$ as
the spectrum $Spec(\C[Z]^G )$. The inclusion $\C[Z]^G \rightarrow
\C[Z]$ gives us the quotient map:
\[ \pi :Z \rightarrow Z/G \]

\begin{ex}
Let us consider $Z=Sym^2 (V^* )\cong \C^3 $ where $V$ is the standard
representation of $G=SL_2 $. As we have seen,
$\C[Z]=\C[B_{20},B_{11},B_{02} ]$. There is only one invariant
$\delta=B_{11}^2 -4B_{02}B_{20} $. Thus $\C[Z]^G =\C[\delta ]$.
Thus $Z/G$ is precisely $Spec(\C[\delta ])=\C$, the complex
plane. The map $\pi $ is executed as follows: given a form
$aX_1^2 +bX_1 X_2 +cX_2^2 \equiv (a,b,c)\in Z$, we eveluate the
invariant $\delta $ at the point $(a,b,c)$. Thus $\pi :\C^3
\rightarrow \C$ is given by:
\[ \pi (a,b,c)=b^2 -4ac \]
Clearly, if $f,f'\in Z$ such that $f=g\cdot f'$ for $g \in SL_2 $
then $\delta (f)=\delta (f')$. We look at the converse: if $f,f'$
are such that $\delta (f)=\delta (f')$ then is it that $f\in
O(f')$? We begin with $f=aX_1^2 +bX_1 X_2 +cX_2^2 $. We assume for
the moment that $a\neq 0$. If that is the case, we make the
substitution $X_1 \rightarrow X_1 -\alpha X_2 $ and $X_2
\rightarrow X_2 $. Note that this transformation is unimodular for
all $\alpha $. This transforms $f$ to:
\[ a(X_1 -\alpha X_2 )^2 +b(X_1 -\alpha X_2 )X_2 +cX_2^2 \]
The coefficient of $X_1 X_2 $ is $-2a\alpha +b $. Thus by choosing
$\alpha $ as $b/2a$ we see that $f$ is transformed into $a''X_1^2
-c''X_2^2 $ for some $a'',c''$. By a similar token, even if $a=0$
one may do a similar transform. Thus in general, if $f$ is not the
zero form, there is a point
in $O(f)$ which is of the form $aX_1^2 -cX_2^2 $. Thus, we may
assume that both $f$ and $f'$ are in this form.  We can simplify
the form further by a diagonal element of $SL_2 $ to put both $f$
and $f'$ as $X_1^2 -cX_2^2 $ and $X_1^2 -c'X_2^2 $. 
It is now clear that $\delta (f)=\delta(f')$ implies that $c=c'$.
Thus the general answer is that if $f,f'\neq 0$ and $\delta
(f)=\delta (f')$ then $f\in O(f')$. 

Next, let us examine the form $0$. we see that $\delta (0)=0$.
Thus we see that (i) for any point $d\in \C$, if $d\neq 0$, then
$\pi^{-1} (d)$ consists of a single orbit, (ii) $\pi^{-1}(0)$
consists of two orbits, $O(0)$ and $O(X_1^2 )$, the perfect
square, (iii) the orbits $O(f)$ are closed when $\delta (f)\neq
0$, (iv) $O(X_1^2 )$ is not closed. Its closure includes the closed
orbit $0$.
\end{ex}

The above example illustrates the utility of contructing the
quotient $Z/G$ as a variety parametrizing closed orbits albeit 
with some deviant points. In the example above, the discrepancy was
at the point $0$, wherein the pre-image is closed but
decomposes into two orbits. 

We state the all-important theorem linking a space and its
quotient in the restricted case when $G$ is finite and $Z$ is a
finite-dimensional $G$-module. Thus $\C[Z]$ is a polynomial ring
and the action of $G$ is homogeneous.

\begin{theorem} \label{thm:finite}
Let $G$ be a finite group and act on the space $Z$. Let $R=\C[Z]$
and $R^G =\C[Z]^G $ be the ring of invariants. Let $\pi :Z
\rightarrow Z/G $ be the quotient map. Then
\begin{itemize}
\item[(i)] For any ideal $J \subseteq R^G $, 
we have $(J\cdot R )\cap R^G =J $.
\item[(ii)] The map $\pi $ is surjective. Further, for any $x\in Z/G$,
$\pi^{-1} (x)$ is a single orbit in $Z$.
\end{itemize}
\end{theorem}

\noindent
{\bf Proof}: Let $f_1 ,\ldots ,f_k $ be elements of $R^G $ and
$f\in R^G $ be such that:
\[ f=r_1 f_1 +\ldots +r_k f_k \]
where $r_i $ are elements of $R$. Since $G$ is finite, we apply
the Reynolds operator $p\rightarrow p^G $. 
In other words, we
have:
\begin{eqnarray*}
f &=& \frac{1}{|G|} \sum_{g\in G} f \\
 &=& \frac{1}{|G|} \sum_{g\in G} \sum_i r_i f_i  \\
 &=& \frac{1}{|G|} \sum_i f_i \sum_{g\in G} r_i   \\
 &=& \sum_i f_i \frac{1}{|G|}  \sum_{g\in G} r_i   \\
&=& \sum_i f_i r_i^G 
\end{eqnarray*}
Note that we have used the fact that if $h\in R^G
$ and $p\in R$ then $(h\cdot p)^G =h \cdot p^G $. 
Thus we see that any element $f$ of $R^G $ which is expressible as
an $R$-linear combination of elements $f_i $'s of $R^G $ is
already expressible as an $R^G $-linear combination of the same
elements. This proves (i) above. 

Now we prove (ii). Firstly, let $J$ be a maximal ideal of $R^G $.
By part (i) above, $J \cdot R$ is a proper ideal of $R$ and thus
$\pi $ is surjective. Let $x\in Z/G$ and $J_x \subseteq R^G $ be the
maximal ideal for the point $x$. Let $z$ be such that $\pi (z)=x$
and let $I_{O(x)} \subseteq R$ be the ideal of all functions in
$R$ vanishing at all points of the orbit $O(z)$ of $z$. We show
that $rad(J_x \cdot R)=I_{O(z)}$ which proves (ii). Towards that, it is
clear that (a) $J_x R \subseteq I_{O(z)}$ and (b) the variety of
$J_x R$ is $G$-invariant. Let $O(z')$ is another orbit in the
variety of $J_x\cdot R$. If $O(z') \neq O(z)$ then we already have
that there is a $p\in R^G $ such that $p(O(z))=0$ and
$p(O(z'))=1$. Since this $p\in J_x \cdot R$, the variety of $J_x
\cdot R$ excludes the orbit $O(z')$ proving our claim. $\Box $ 

\begin{theorem} \label{thm:finitegen}
Let $Z$ be a finite-dimesional $G$-module for a finite group $G$.
Then $\C[Z]^G $, the ring of $G$-invariants, is finitely generated
as a $\C$-algebra.
\end{theorem}

\noindent
{\bf Proof}: Since the action of $G$ is homogeneous, every
invariant $f\in \C[Z]^G $ is the sum of homogeneous invariants.
Let $I_G \subseteq \C[Z]$ be the ideal generated by the {\bf
positive} degree invariants. By the Hilbert Basis theorem, $I_G
=(f_1 ,\ldots ,f_k )\cdot \C[Z]$, where each $f_i $ is itself an
invariant. we claim that $\C[Z]^G =\C[f_1 ,\ldots ,f_k ]\subseteq
\C[Z]$, or in other words, every invariant is a polynomial over
$\C$ in the terms $f_1 ,\ldots ,f_k $. To this end, let $f$ be a
homogeneous invariant. We prove this by induction
over the degree $d$ of $f$. Since $f\in I_G $, we have
\[ f=\sum_i r_i f_i \: \: \: \mbox{where for all $i$, } 
r_i \in \C[Z] \] 
By applying the reynolds operator, we have:
\[ f=\sum_i h_i f_i \: \: \: \mbox{where for all $i$, } 
h_i \in \C[Z]^G \] 
Now since each $h_i $ has degree less than $d$, by our induction
hypothesis, each is an element of $\C[f_1 ,\ldots ,f_k ]$. This
then shows that $f$ too is an element of $\C[f_1 ,\ldots ,f_k ]$.
$\Box $

\section{The Nagata Hypothesis}
We now generalize the above two theorems to the case of more
general groups. This generlization is possible if the group $G$
satisfies what we call the {\em Nagata Hypothesis}.

\begin{defn}
Let $G$ be an algebraic group. We say that $G$ satisfies the
Nagata hypothesis if for every finite-dimensional module $M$ over
$\C$ and a $G$-invariant hyperplane $N\subseteq M$, there is a
decomposition $M=H\oplus P$ as $G$-modules, where $H$ is a
hyperplane and $P$ is an invariant (i.e., $P$ is the trivial
representation. 
\end{defn}

The group $SL_n $ satisfies the hypothesis, and so do the
so-called reductive groups over $\C$. 

\begin{theorem}
Let $G$ satisfy the Nagata hypothesis and let $Z$ be an affine
$G$-variety. 
The ring $\C[Z]^G $ is finitely generated as a $\C$-algebra.
\end{theorem}

The proof will go through several steps.

Let $f\in \C[Z]$ be an arbitrary element. We define two modules
$M(f)$ and $N(f)$. 
\begin{eqnarray*}
 M(f)&=&\C \cdot \{ s\cdot f | s\in G \}  \\
 N(f)&=& \C \cdot \{ s\cdot f -t\cdot f| s,t\in G \} 
 \end{eqnarray*}
Thus $M(f)\supseteq N(f)$ are finite-dimensional submodules of 
$\C[Z]$. 

\begin{lemma} \label{lemma:Nf}
Let $f\in \C[Z]$ be an arbitrary element. Then there exists an
$f^* \in \C[Z]^G \cap M(f)$ such that $f-f^* \in N(f)$.
\end{lemma}

\noindent
{\bf Proof}: We prove this by induction over $dim (M(f))$. Note
that $s\cdot f=f+(s\cdot f -1\cdot f)$ and thus $M(f)=N(f)\oplus
\C \cdot f$ as vector spaces. Thus $N(f)$ is a $G$-invariant
hyperplane of $M(f)$. By the Nagata hypothesis, $M(f)=f' \oplus
H$, where $f'$ is an invariant. If $f' \not \in N(f)$ then
$M(f)=f'\oplus N(f)$ and the lemma follows. However, if $f'\in
N(f)$, then $f=f'+h$ where $h\in H$. Since $H$ is a $G$-module, we
see that $M(h)\subseteq H$ which is of dimension less than that of
$M(f)$. Thus there is an invariant $h^* $ such that $h-h^* \in
N(h)$. Note that $f-f'=h$ with $f'$ invariant, implies that
$s\cdot h-t\cdot h=s\cdot f-t\cdot f$. Thus $N(h)\subseteq N(f)$.
We take $f^* =h^*$. Examining $f-f^*$, we see that
\[ f-h^* =f'+h-h^* \in N(f) \]
This proves the lemma. $\Box $

\begin{lemma} \label{lemma:idealext}
Let $f_1 ,\ldots ,f_r \in \C[Z]^G $. Then $\C[Z]^G \cap (\sum_i
\C[Z]\cdot f_i )= \sum_i \C[Z]^G \cdot f_i $. Thus if an invariant
$f$ is a $\C[Z]$linear combination of invariants, then it is
already a $\C[Z]^G $ linear combination.
\end{lemma} 

\noindent
{\bf Proof}: This is proved by induction on $r$. Say
$f=\sum_{i=1}^r h_i f_i $ with $h_i \in \C[Z]$. Applying the above
lemma to $h_r $, there is an $h'' \in \C[Z]^G $ and an $h'\in
N(h_k )$ such that $h_r =h'+h''$. We tackle $h'$ as follows.  
Since $f$ is an invariant, we have for all $s,t\in G$:
\[ \sum_{i=1}^r (s\cdot h_i -t\cdot h_i )f_i =s\cdot f -t\cdot f=0
\]
Hence:
\[ (s\cdot h_r -t\cdot h_r )f_r =\sum_{i=1}^{r-1} 
(s\cdot h_i -t\cdot h_i )f_i \]
It follows from this that $h' f_r =\sum_{i=1}^{r-1} h'_i f_i $ for
some $h'_i \in \C[Z]$. Substituting this in the expression 
\[ f=\sum_{i=1}^{r-1} h_i f_i +(h'+h'')f_r \]
we get:
\[ f-h''f_r =\sum_{i=1}^{r-1} (h_i +h'_i )f_i \]
Thus the invariant $f-h''f_r $ is $\C[Z]$-linear combination of
$r-1$ invariants, and thus the induction hypothesis applies. This
then results in an expression for $f$ as a $\C[Z]^G$-linear
combination for $f$. $\Box $

The above lemma proves part (i) of Theorem~\ref{thm:finite} for
groups $G$ for which the Nagata hypothesis holds. The route to
Theorem~\ref{thm:finitegen} is now a straight-forward adaptation of 
its proof for finite groups. In the case when $Z$ is a $G$-module,
$\C[Z]^G$ would then be homogeneous and the proof of
Theorem~\ref{thm:finitegen} holds. For general $Z$ here is a trick
which converts it to the homogeneous case:

\begin{prop}
Let $Z$ be an affine $G$-variety. Then there is a $G$-module $W$
and an equivariant embedding $\phi :Z \rightarrow W$. 
\end{prop}

\noindent
{\bf Proof}: Since $\C[Z]$ is a finitely generated $\C$-algebra, 
$\C[Z]=\C[f_1 ,\ldots ,f_k ]$ where $f_1 ,\ldots
,f_k $ are some specific elements of $\C[Z]$. 
By expanding the list of generators, we may assume that the vector
space $\C \cdot\{ f_1,\ldots ,f_k \} \subseteq \C[Z]$ is a
finite-dimensional $G$-module. Let us construct $W$ as an
isomorphic copy of this module with the dual basis $W_1 ,\ldots
,W_k $. We construct the map $\phi^* :\C[W] \rightarrow \C[Z]$ by
defining $\phi (W_i )=f_i $. Since $\C[W]$ is a free algebra, we
indeed have the surjection:
\[ \phi^* :\C[W]  \rightarrow \C[Z] \] 
This proves the proposition. $\Box $

We are now ready to prove:
\begin{theorem} \label{thm:generalfinitegen}
Let $G$ satisfy the Nagata hypothesis. Let $Z$ be an affine
$G$-variety with coordinate ring $\C[Z]$. Then $\C[Z]^G $ is
finitely generated as a $\C$-algebra.
\end{theorem}

\noindent
{\bf Proof}: We construct the equivariant surjection $\phi^*
:\C[W]\rightarrow \C[Z]$. Note that $\C[W]^G $ is already finitely
generated over $\C$, say $\C[W]^G =\C[h_1 , \ldots ,h_k ]$. We
claim that $\phi^* (h_1 ), \ldots , \phi^* (h_k )$ generate $\C[Z]^G $.

Now, let $f\in \C[Z]^G$. By the surjectivity
of $\phi^* $, there is an $h\in \C[W]$ such that $\phi^* (h)=f$. 
Consider the space $N(h)$. A typical generator of $N(h)$ is
$s\cdot h -t\cdot h$. Applying $\phi^* $ to this, we see that:
\begin{eqnarray*}
\phi^* (s\cdot h-t \cdot h)&=& s\cdot \phi^*(h) -t \cdot \phi^*
(h)\\
&=& f-f=0 
\end{eqnarray*}
By an earlier lemma there is an invariant $h^* $ such that $h-h^*
\in N(h)$. thus applying $\phi^* $ we see that $\phi^* (h^* )=\phi
(h)=f$. Thus there is an invariant $h^* $ such that $\phi^* (h^*
)=f$. Now $h^* \in \C[h_1 ,\ldots ,h_k ]$ implies that $f=\phi^*
(h^* ) \in \C[\phi^* (h_1 ), \ldots ,\phi^* (h_k ) ]$. $\Box $

\chapter{Orbit-closures}

\noindent {\em Reference:} \cite{kempf,nagata} 

In this chapter we will analyse the validity of
Theorem~\ref{thm:finite} for general groups $G$ with the Nagata
hypothesis. The objective is to analyse the map $\pi :Z
\rightarrow Z/G$. 

\begin{prop}
Let $G$ satisfy the Nagata hypothesis and act on the affine
variety $Z$. Then the map $\pi :Z \rightarrow Z/G$ is surjective.
\end{prop}

\noindent
{\bf Proof}: Let $J\subseteq \C[Z]^G $ be a maximal ideal. By
lemma~\ref{lemma:idealext} $J\cdot \C[Z]$ is a proper ideal of
$\C[Z]$. This implies that $\pi $ is surjective. $\Box $

\begin{theorem} \label{theorem:closed-separation}
Let $G$ satisfy the Nagata hypothesis and act on the affine
variety $Z$. 
Let $W_1 $ and $W_2 $ be $G$-invariant (Zariski-)closed subsets of
$Z$ such that $W_1 \cap W_2 $ is empty. 
Then there is an invariant $f\in \C[Z]^G $ such that $f(W_1
)=0$ and $f(W_2 )=1$. 
\end{theorem}

\noindent
{\bf Proof}: Let $I_1 $ and $I_2 $ be the ideals of $W_1 $ and
$W_2 $ in $\C[Z]$. Since their intersection is empty, by Hilbert
Nullstellensatz, we have $I_1 +I_2 =1$. Whence there are functions
$f_1\in I_1 $ and $f_2 \in I_2 $ such that $f_1 +f_2 =1$. For
arbitrary $s,t\in G$ we have:
\[ (s\cdot f_1 -t\cdot f_1 )+(s\cdot f_2 -t\cdot f_2 )=1-1=0 \]
Note that $M(f_i )$ and $N(f_i )$ are submodules of the
$G$-invariant ideal $I_i $. Now applying lemma~\ref{lemma:Nf} to
$f_1 $, we see that there are elements $f'_i \in N(f_i )$ such
that $f_1 +f'_1 \in \C[Z]^G $. Whence we see that:
\[ (f_1 +f'_1 )+(f_2 +f'_2 )=1 \]
where $f=f_1 +f'_1 $ is an invariant. Since $f_1 +f'_1 \in I_1 $ we
see that $f(W_1 )=0$. On the other hand, $f=1-(f_2 +f'_2)\in 1+I_2
$ and thus $f(W_2 )=1 $. $\Box $

Recall that $\Delta (z)$ denotes the closure of the orbit $O(z)$
in $Z$.

\begin{theorem} \label{thm:orbitclosure}
Let $G$ satisfy the Nagata hypothesis and act on an affine variety
$Z$. We define the relation $\approx $ on $Z$ as follows: 
$z_1 \approx z_2 $ if and only if $\Delta (z_1 ) \cap \Delta (z_2
)$ is non-empty. Then 
\begin{itemize}
\item[(i)] $\approx $ is an equivalence relation.
\item[(ii)] $z_1 \approx z_2 $ iff $f(z_1 )=f(z_2 )$ for all $f\in
\C[Z]^G $.
\item[(iii)] Within each $\Delta (z)$ there is a unique closed
orbit, and this is of minimum dimension among all orbits in
$\Delta (z)$.
\end{itemize}
\end{theorem}

\noindent
{\bf Proof}: It is clear that (ii) proves (i). Towards (ii), if
$\Delta (z_1 ) \cap \Delta (z_2 ) $ is empty, then by
Theorem~\ref{theorem:closed-separation}, there is an invariant
separating the two points. Thus it remains to show that if $\Delta
(z_1 0 \cap \Delta (z_2 )$ is {\em non-empty}, and $f$ is any
invariant, then $f(z_1 )=f(z_2 )$. So let $z$ be an element of the
intersection. Since $f$ is an invariant, we have $f(O(z_1 ))$ is a
constant, say $\alpha $. Since $O(z_1 )\subseteq \Delta (z_1 )$ is
dense, and $f$ is continuous, we have $f(z)=\alpha $. Thus $f(z_1
)=f(z_2 )=f(z)=\alpha $. 

Now (iii) is easy. Clearly $\Delta (z)$ cannot have {\em two}
closed distinct orbits, for otherwise they woulbe separated by
an invariant. But this must take the same value on all points of
$\Delta (z)$. That it is of minimum dimension follows algebraic
arguments. $\Box $

\begin{defn}
Let $Z$ be a $G$-variety and $z\in Z$. We say that $z$ is {\bf
stable} if the orbit $O(z)\subseteq Z$ is closed in $Z$.
\end{defn}

By the above theorem, every point $x$ of $Spec(C[Z]^G )$ corresponds
to exactly one stable point: the point whose orbit is of minimum
dimension in $\pi^{-1}(x)$.

\begin{ex}
Consider the action of $G=SL_n $on ${\cal M}$, the space of $n\times
n$-matrices by conjugation. Thus, given $A\in SL_n $ and $M\in
{\cal M}$, we have:
\[A\cdot M=AMA^{-1} \]
Let $R=\C[{\cal M}]=\C[X_{11},\ldots ,X_{nn} ]$ be the ring of
functions on ${\cal M}$. The invariants $R^G $ is generated as a
$\C$-algebra by the forms $e_i (X)=Tr(X^i )$, for $i=1,\ldots ,n$. 
The forms $\{ e_i |1\leq i \leq n\}$ are algebraically independent
and thus $R^G $ is the polynomial ring $\C[e_1 ,\ldots ,e_n ]$.
Clearly then $Spec(R^G ) \cong \C^n $ and we have:
\[ \pi : {\cal M} \cong \C^{n^2} \rightarrow \C^n \]
Given a matrix $M$ with eigenvalues $\{ \lambda_1 ,\ldots
,\lambda_n \}$, we have:
\[ e_i (M)=\lambda_1^i +\ldots +\lambda_n^i \]
Thus, by the fundamental theorem of algebra, the image $\pi (M)$ 
determines the set $\{ \lambda_1 ,\ldots ,\lambda_n \}$ (with
multiplicities). On the
other hand, given any tuple $\mu=(\mu_1 ,\ldots ,\mu_n )$ there is a
unique set 
$\lambda_{\mu } =\{ \lambda_1 ,\ldots ,\lambda_n \}$ such that $\sum_r
\lambda_r^i =\mu_i $. Clearly, for the diagonal matrix 
$D(\lambda_{\mu })=
diag(\lambda_1 ,\ldots ,\lambda_n )$, we have that $\pi (D(\lambda
))=\mu $. This verifies that $\pi $ is surjective. 

For a given $\mu $, the set $\pi^{-1}(\mu )$ are all matrices  
$M$ with $Spec(M)=\lambda =\lambda_{\mu }$. By the Jordan
canonical form (JCF), this set may be stratified by the various Jordan
canonical blocks of spectrum $\lambda $. If $\lambda $ has no
multiplicities then $\pi^{-1} (\mu )$ consists of just one orbit:
matrices $M$ such that $JCF(M)=D(\lambda )$. For a general
$\lambda $, the orbit of $M$ with $JCF(M)=D(\lambda )$ is the
unique closed orbit of minimum dimension. All other orbits contain
this orbit in its closure. Thus stable points $M\in {\cal M}$ are
the diagonalizable matrices. 

As an example, consider the case when
$n=2$ and the matrix:
\[ N=\left[\begin{array}{cc}
	 \lambda & 1 \\
	 0 & \lambda \end{array} \right] \]
Consider the family $A(t)=diag(t,t^{-1} )\in SL_2 $. We see
that:
\[ N(t)=A(t)NA(t)^{-1}=\left[\begin{array}{cc}
	 \lambda & t^2  \\
	 0 & \lambda \end{array} \right] \]
Thus $\lim_{t\rightarrow 0} N(t)=diag(\lambda ,\lambda )$, the
diagonal matrix. 
\end{ex}

Thus, we see that the invariant ring $\C[Z]^G $ puts a different
equivalence relation $\approx $ on points in $Z$ which is coarser
than $\cong $, the orbit equivalence relation. The relation
$\approx $ is more `topological' than group-theoretic and
correctly classifies orbits by their separability by invariants.
The special case of Theorem~\ref{thm:orbitclosure} when $Z$ is a
representation was analysed by Hilbert in 1893. The point $0\in Z$
is then the smallest closed orbit, and the equivalence class
$[0]_{\approx}$ is termed as the {\bf null-cone} of $Z$. We see
that the null-cone consists of all points $z\in Z$ such that $0$
lies in the orbit-closure $\Delta (z)$ of $z$. It was Hilbert who
discovered that if $0\in \Delta (z)$ then $0$ lies in the
orbit-closure of a $1$-parameter diagonal subgroup of $SL_n $. 
To understand the intricay of Hilbert's constructions, it is
essential that we understand diagonal subgroups of $SL_n $.
 
\chapter{Tori in $SL_n $}

\noindent {\em Reference:} \cite{kempf,nagata} 

Let $\C^* $ denote the multiplicative group of non-zero complex
numbers. A torus is the abstract group $(\C^* )^m $ for some $m$.
Note that $\C^* $ is an abelian algebraic group with 
$\C[G]=\C[T,T^{-1}]$. Furthermore, $\C^* $ has a compact subroup
$S^1 =\{ z\in \C, |z|=1\}$, the unit circle. 

Next, let us look at representations of tori. For $\C^* $, the
simplest representations are indexed by integers $k\in \Z$. So let
$k\in \Z$. The representation $\C_{[k]}$ corresponds to the
$1$-dimensional vector space $\C$ with the action:
\[ t \cdot z= t^k z \]
Thus a non-zero $t\in \C^* $ acts on $z$ by multiplication of the
$k$-th power. 
Next, for $(\C^* )^m $, let $\chi =(\chi [1],\ldots ,\chi [m])$ be
a sequence of integers. For such a $\chi $, we define the
representation $\C_{\chi }$ as follows: Let $\overline{t}=
(t_1 ,\ldots ,t_m
)\in (\C^* )^m $ be a general element and $z\in \C$. The action is
given by:
\[ \overline{t}\cdot z =t_1^{\chi [1]}\ldots t_m^{\chi [m]} z \]
Such a $\chi $ is called a {\bf character} of $(\C^* )^m $. 

These $1$-dimensional representations of tori are crucial in the
analysis of algebraic group actions.

Let us begin by understanding the structure of algebraic
homomorphisms from $\C^* $ to $SL_n (\C )$. So let 
\[ \lambda :\C^* \rightarrow SL_n (\C) \]
be such a map such that $\lambda (t)=[a_{ij}(t)]$ where $a_{ij}(T)
\in \C[T,T^{-1} ]$. An important substitution is for $t=e^{ i
\theta }$, and we obtain a $2\pi $-periodic map 
\[ \overline{\lambda} :\C \rightarrow SL_n \]
We see that $\overline{\lambda }(0)=I$. Let the derivative at $0$
for $\overline{\lambda }$ be $X$. 

We have the following general lemma:
\begin{lemma}
Let $f:\C \rightarrow SL_n $ be a smooth map such that $f(0)=I$
and $f'(0)=X $, where $X$ is an $n\times n$-matrix. Then for
$\theta \in \C$, 
\[ \lim_{k\rightarrow \infty}
\left[f \left( \frac{\theta}{k}\right)\right]^k = e^{\theta X} \]
\end{lemma}
The proof follows from the local diffeomorphism of the exponential
map in the neighborhood of the identity matrix.

Applying this lemma to $\overline{\lambda}$ we see that $e^{\theta
X} $ is in the image of $\overline{\lambda }$ for all $\theta $.
Now, since $\overline{\lambda }$ is $2\pi $-periodic, we must have
$e^{(n2\pi +\theta )X}= e^{\theta X}$. This forces (i) $X$ to be
diagonalizable, and (ii) with integer eigenvalues. This proves:
\begin{prop} \label{prop:diagonal}
Let $\lambda :\C^* \rightarrow SL(V)$ be an algebraic
homomorphism. Then the image of $\lambda $ is closed and 
$V\cong \C_{[m_1]}\oplus \ldots \oplus \C_{[m_n
]} $, for some integers $m_1 ,\ldots ,m_n $, where
$n=dim_{\C}(V)$. 
\end{prop}

Based on this, we have the generalization:
\begin{prop} \label{prop:rdiagonal}
Let $\lambda :(\C^* )^r  \rightarrow SL(V)$ be an algebraic
homomorphism. Then the image of $\lambda $ is closed and 
$V\cong \C_{\chi_1 }\oplus \ldots \oplus \C_{\chi_n } $, for some 
integers $\chi_1 ,\ldots ,\chi_n $, where $n=dim_{\C}(V)$. 
\end{prop}

Thus, in effect, for every homomorphism $\lambda :(\C^* )^r 
\rightarrow SL_n $, there is a fixed invertible matrix $A$ such that
for all $\overline{t} \in (\C^* )^r $, the conjugate $A\lambda
(\overline{t})A^{-1}$ is diagonal.

A torus in $SL_n $ is defined as an abstract subgroup $H$ of $SL_n
$ which is isomorphic to $(\C^* )^r $ for some $r$. The {\bf
standard maximal torus} $D$ of $SL_n $ is the diagonal matrices
$diag(t_1 ,\ldots ,t_n )$ where $t_i \in \C^* $ and $t_1 t_2
\ldots t_n =1$. 

This clears the way for the important theorem:

\begin{theorem} \label{thm:sln-tori}
\begin{itemize}
\item[(i)] Every torus is contained in a maximal torus. 
All maximal tori in $SL_n $ are isomorphic to
$(\C^*)^{n-1}$.

\item[(ii)] If $T$ and $T'$ are two maximal tori then there is an
$A\in SL_n $ such that $T'=ATA^{-1}$. Thus all maximal tori are
conjugate to $D$ above. 

\item[(iii)] Let $N(D)$ be the normalizer of $D$ and $N(D)^o $ be
the connected component of $N(D)$. Then $N(D)^o =D$ and $N(D)/D$
is the {\bf Weyl group} $W$, isomorphic to the symmetric group
$S_n $. 

\end{itemize}
\end{theorem}

\begin{defn}
Let $G$ be an algebraic group. $\Gamma (G)$ will denote the
collection of all {\bf $1$-parameter subgroups} of $G$, i.e., 
morphisms $\lambda :\C^* \rightarrow G$. $X(G)$ will denote
the collection of all {\bf characters} of $G$, i.e., homomorphisms
$\chi :G \rightarrow \C^* $.
\end{defn}

We consider the case when $G=(\C^* )^r $. Clearly, for a given
$\lambda :\C^* \rightarrow G$, there are integers $m_1 ,\ldots
,m_r \in \Z$ such that:
\[ \lambda (t)=(t^{m_1 },\ldots ,t^{m_r }) \]
In the same vein, for the character $\chi :G \rightarrow \C^* $,
we have integers $a_1 ,\ldots ,a_r $ such that 
\[ \chi (t_1 ,\ldots ,t_r )=\prod_i t_i^{a_i } \]
We also have the composition $\lambda \circ \chi :\C^* \rightarrow
\C^* $, where by:
\[ \lambda \circ \chi (t)=t^{m_1 a_1 +\ldots +m_r a_r } \]
Consolidating all this, we have:

\begin{theorem}
Let $G=(\C^* )^r $. Then $\Gamma (G) \cong \Z^r $ and $X(G) \cong
\Z^r $. Furthermore, there is the pairing 
\[ <\: ,\: >:\Gamma (G) \times X(G) \rightarrow \Z \]
which is a unimodular pairing on lattices.
\end{theorem}

\begin{ex}
Let $G=(\C^*)^3 $ and $\lambda $ and $\chi $ be as follows:
\begin{eqnarray*}
\lambda (t) &=& (t^3 ,t^{-1} ,t^2 ) \\
\chi (t_1 ,t_2 ,t_3 ) &=& t_1^{-1} t_2 t_3^2 
\end{eqnarray*}
Then, $\lambda \cong [3,-1,2]$ and $\chi \cong [-1,1,2]$. We
evaluate the pairing:
\[ <\lambda , \chi >=3\cdot -1 +(-1)\cdot1 +2\cdot 2=0 \]
\end{ex}

We now turn to the special case of $D\subseteq SL_n $, the maximal
torus which is isomorphic to $(\C^* )^{n-1}$. By the above
theorem, $\Gamma (D), X(D) \cong \Z^{n-1}$. However, it will more
convenient to identify this space as a subset of $\Z^n$. So let:
\[ \Y^n =\{ [m_1 ,\ldots ,m_n ]\in \Z^n | m_1 +\ldots +m_n =0 \}
\]
It is easy to see that $\Y^n \cong \Z^{n-1}$. In fact, we will set
up a special bijection $\theta :\Y^n \rightarrow \Z^{n-1} $
defined as:
\[ \theta ([m_1 ,m_2 ,\ldots ,m_n ])=[m_1 ,m_1 +m_2 ,\ldots ,m_1
+\ldots +m_{n-1} ] \]
The inverse $\theta^{-1}$ is also easily computed:
\[ \theta^{-1} [a_1 ,\ldots ,a_{n-1}]=[a_1 ,a_2 -a_1 ,a_3 -a_2 ,
\ldots ,a_{n-1}-a_{n-2}, -a_{n-1} ] \]
This $\theta $ corresponds to the $\Z$-basis of $\Y^n $
consisting of the vectors $e_1 -e_2 ,\ldots ,e_{n-1}-e_n$ where $e_i
$ is the standard basis of $\Z^n $. This is also equivalent to the
embedding $\theta^* :(\C^* )^{n-1} \rightarrow D$ as follows:
\[ (t_1 ,\ldots ,t_{n-1}) \rightarrow \left[ 
		   \begin{array}{ccccc}
		   t_1 & 0 & \ldots &  & 0 \\
		   0 & t_1^{-1} t_2 & 0 & \ldots & 0\\
		   & & \vdots & & \\
		   0 & \ldots & 0 & t_{n-2}^{-1} t_{n-1} & 0\\
		   0 & & \ldots & 0 & t_{n-1}^{-1} \end{array}
		   \right] \]
A useful computation is to consider the inclusion $D\subseteq D^*
$, where $D^* \subseteq GL_n $ is subgroup of {\em all} diagonal
matrices. Clearly $\Gamma (D) \subseteq \Gamma (D^* )$, however
there is a surjection $X(D^* )\rightarrow X(D)$. It will be useful
to work out this surjection explicitly via $\theta $ and $\theta^*
$. If $[m_1
,\ldots .m_n ]\in \Z^n \cong X(D^* )$, then it maps to $[m_1 -m_2 ,
\ldots ,m_{n-1} -m_n ] \in \Z^{n-1} \cong X((\C^* )^{n-1}) $ via
$\theta^* $. If
we push this back into $\Y^n $ via $\theta^{-1}$, we get:
\[ [m_1 ,\ldots ,m_n ]\rightarrow [m_1 -m_2 ,2m_2 -m_1 -m_3 ,
,\ldots , 2m_{n-1}-m_{n-2}-m_{n} , m_n -m_{n-1} ] \]

We are now ready to define the {\bf weight spaces} of an $SL_n
$-module $W$. So let $W$ be such a module. By restricting this
module to $D\subseteq G$, via Proposition \ref{prop:rdiagonal}, we 
see that $W$ is a direct sum $W=\C_{\chi_1 } \oplus \ldots \oplus
\C_{\chi_N } $, where $N=dim_{\C}(W)$. Collecting identical
characters, we see that:
\[ W =\oplus_{\chi \in X(D)} \C_{\chi}^{m_{\chi }} \]
Thus $W$ is a sum of $m_{\chi }$ copies of the module $\C_{\chi
}$. Clearly $m_{\chi }=0$ for all but a finite number, and is
called the {\bf multiplicity} of $\chi $. For a given module $W$,
computing $m_{\chi }$ is an intricate combinatorial exercise and
is given by the celebrated {\bf Weyl Character Formula}. 

\begin{ex}
Let us look at $SL_3 $ and the weight-spaces for some modules of
$SL_3 $. All modules that we discuss will also be $GL_3 $-modules
and thus $D^* $ modules. The formula for converting $D^* $-modules
to $D$-modules will be useful. This map is $\Z^3 \rightarrow \Y^3
$ and is given by:
\[ [m_1 ,m_2 ,m_3 ]\rightarrow [m_1 -m_2 , 2m_2 -m_1 -m_3 ,m_3 -m_2 ] \]

The simplest $SL_n $ module is $\C^3 $ with the basis $\{ X_1 ,X_2
,X_3 \}$ with $D^* $ weights $[1,0,0], [0,1,0]$ and $[0,0,1]$.
This converted to $D$-weights give us $\{ [1,-1,0], [-1,2,-1],
[0,-1,1] \}$, with $\C_{[1,-1,0]} \cong \C\cdot X_1 $ and so on. 

The next module is $Sym^2 (\C^3 )$ with the basis $X_i^2 $ and
$X_i X_j $. The six $D^* $ and $D$-weights with the weight-spaces
are given below:
\[ 
\begin{array}{|c|c|c|}\hline 
\mbox{$D^*$-wieghts} & \mbox{$D$-weights} & \mbox{weight-space}
\\ \hline \hline
[2,0,0] & [2,-2,0] &X_1^2 \\ \hline
[0,2,0] & [-2,4,-2] & X_2^2 \\ \hline
[0,0,2] & [0,-2,2] & X_3^2 \\ \hline
[0,1,1] & [-1,1,0]& X_2 X_3 \\ \hline
[1,0,1] & [1,-2,1]& X_1 X_3 \\ \hline
[1,1,0] & [0,1,-1]& X_1 X_2 \\ \hline
\end{array} \]

The final example is the space of $3\times 3$-matrices ${\cal M}$ 
acted upon by conjugation. We see at once that ${\cal M}={\cal
M}_0 \oplus \C \cdot I $ where ${\cal M}_0 $ is the
$8$-dimensional space of trace-zero matrices, and $\C \cdot I$ is
$1$-dimensional space of multiples of the idenity matrix. Weight
vectors are $E_{ij}$, with $1\leq i,j \leq 3$. The $D^* $ weights
are $[1,-1,0], [1,0,-1],[0,1,-1], [-1,0,1], [0,-1,1], [-1,1,0]$
and $[0,0,0]$. The multiplicity of $[0,0,0]$ in ${\cal M}$ is $3$
and in ${\cal M}_0 $ is $2$. Note that $E_{ii}\not \in {\cal M}_0
$. The $D$-weights are $[2,-3,1],[1,0,-1],[-1,3,-2]$ and its
negatives, and obviously $[0,0,0]$. 

\end{ex}

The normalizer $N(D)$ gives us an action of $N(D)$ on the weight
spaces. If $w$ is a weight-vector of weight $\chi $, $t\in D$
and $g\in N(D)$, then $g\cdot w$ is also a weight vector. Afterall
$t\cdot (g \cdot w)=g\cdot t' \cdot w$ where $t'=g^{-1}tg$. Thus
\[ t\cdot (g \cdot w)=\chi (t') (g \cdot w) \]
whence $g\cdot w$ must also be a weight-vector with some weight
$\chi '$. This $\chi '$ is easily computed via the action of $D^*
$. Here the action of $N(D^* )$ is clear: if $\chi =[m_1 ,\ldots
,m_n ]$, then $\chi ' =[m_{\sigma (1)},\ldots ,m_{\sigma (n)} ]$
for some permutation $\sigma \in S_n $ determined by the component
of $N(D^* )$ containing $g$. Thus the map $\chi $ to $\chi '$ for
$D$-weights in the case of $SL_3 $ is as follows:
\[ [m_1 -m_2 , 2m_2 -m_1 -m_3 ,m_3 -m_2 ]  \rightarrow [m_{\sigma
(1)} -m_{\sigma(2)} ,2m_{\sigma (2)} -m_{\sigma (1)} -m_{\sigma
(3)}, m_{\sigma (3)}-m_{\sigma (2)} ] \]
\noindent
{\bf Caution}: Note that though $\Y^3 \subseteq \Z^3 $ is an $S_3
$-invariant subset, the action of $S_3 $ on $\chi \in \Y^3 $ is
{\bf different}. Note that, e.g., in the last example above,
$[2,-3,1]$ is a weight but not the `permuted' vector $[-3,2,1]$. 
This is because of our
peculiar embedding of $\Z^{n-1} \rightarrow \Y^n $. 
 
\chapter{The Null-cone and the Destabilizing flag}

\noindent {\em Reference:} \cite{kempf,nagata} 

The fundamental result of {\bf Hilbert} states:
\begin{theorem} \label{thm:hilbert}
Let $W$ be an $SL_n $-module, and let $w\in W$ be an element of
the null-cone. Then there is a $1$-parameter subgroup $\lambda
:\C^* \rightarrow SL_n $ such that 
\[ \lim_{t \rightarrow 0} \lambda (t)\cdot w =0_W \]
\end{theorem}

In other words, if the zero-vector $0_W $ lies in the
orbit-closure of $w$, then there is a $1$-parameter subgroup
taking it there, in the limit. We will not prove this statement
here. Our objective for this chapter is to interpret the geometric
content of the theorem. We will show that there is a {\em standard
form} for an element of the null-cone. For well-known
representations, this standard form is easily identified by
geometric concepts. 

\section{Characters and the half-space criterion}
To begin, let $D$ be the fixed maximal torus. For any $w\in W$, we
may express:
\[ w=w_1 +w_2 +\ldots + w_r \]
where $w_i \in W_{\chi_i }$, the weight-space for character
$\chi_i $. Note the the above expression is unique if we insist
that each $w_i $ be non-zero. The set of characters $\{ \chi_1
,\ldots ,\chi_r \}$ will be called the {\bf support} of $w$ and
denoted as $supp(w)$. 
Let $\lambda :\C^* \rightarrow SL_n $ be such that $Im(\lambda
)\subseteq D$. In  this case, the action of $t\in \C^* $ via
$\lambda $ is easily described:
\[ t \cdot w= t^{(\lambda ,\chi_1 )}w_1 +\ldots +t^{(\lambda
,\chi_r )} w_r \]
Thus, if $\lim_{t \rightarrow 0} t \cdot w$ exists (and is $0_W $),
then for all $\chi \in supp(w)$, we have $(\lambda ,\chi ) \geq 0$
(and further $(\lambda ,\chi )>0$).

Note that $(\lambda ,\chi )$ is implemented as a linear functional
on $\Y^n $. Thus, if $\lim_{t \rightarrow 0} t\cdot w $ exists (and
is $)_W $) then there is a {\bf hyperplane} in $\Y^n $ such that
the support of $w$ is on one side of the hyperplane ({\bf
strictly} on one side of the hyperplane). The normal to this 
hyperplane is given
by the conversion of $\lambda $ into $\Y^n $ notation.

On the other hand if the support $supp(w)$ enjoys the 
geometric/combinatorial 
property, then by the approximability of reals by rationals, we
see that there is a $\lambda $ such that $\lim_{t \rightarrow 0}
t\cdot w$ exists (and is zero). 

Thus for $1$-parameter subgroups of $D$, Hilbert's theorem
translates into a combinatorial statement on the lattice subset
$supp(w)\subset \Y^n $. We call this the ({\bf strict}) 
{\bf half-space} property. In the general case, we know that given
any $\lambda :\C^* \rightarrow SL_n $, there is a maximal torus
$T$ containing $Im(\lambda )$. By the conjugacy result on maximal
tori, we know that $T=ADA^{-1}$ for some $A\in SL_n $. Thus, we
may say that $w$ is in the null-cone iff there is a translate
$A\cdot w$ such that $supp(A\cdot w)$ satisfies the strict half-space  
property. 

\begin{ex}
Let us consider $SL_3 $ acting of the space of forms of degree
$2$. For the standard torus $D$, the weight-spaces are $\C \cdot
X_i^2 $ and $\C \cdot X_i X_j $. Consider the form $f=(X_1 +X_2 +X_3
)^2 $. We see that $supp(f)$ is set of all characters of $Sym^2
(\C^3 )$ and {\bf does not} satisfy the combinatorial property.
However, under a basis change $A$:
\[ \begin{array}{rcl}
	  X_1 & \rightarrow & X_1 +X_2 +X_3 \\
	  X_2 & \rightarrow & X_2 \\
	  X_3 & \rightarrow & X_3 \end{array} \]
we see that $A\cdot f=X_1^2 $. Thus $A\cdot f$ {\bf does} satisfy
the strict half-space property. Indeed consider the $\lambda $
\[ \lambda (t) = \left[ \begin{array}{ccc}
t & 0 & 0 \\
0 & t^{-1} & 0 \\
0 & 0 & 1 \end{array} \right] \]
We see that 
\[ \lim_{t \rightarrow 0} t\cdot (A \cdot f)=t^2 X_1^2 =0 \]
Thus we see that every form in the null-cone has a standard form
with a very limited sets of possible supports. 

Let us look at the module ${\cal M}$ of $3\times 3$-matrices under
conjugation. Let us fix a $\lambda $: 
\[ \lambda (t) = \left[ \begin{array}{ccc}
t^{n_1} & 0 & 0 \\
0 & t^{n_2 } & 0 \\
0 & 0 & t^{n_3} \end{array} \right] \]
such that $n_1 +n_2 +n_3 =0$. We may assume that $n_1 \geq n_2
\geq n_3 $. Looking at the action of $\lambda (t)$ on a general
matrix $X$, we see that:
\[ t \cdot X= (t^{n_i -n_j } x_{ij} ) \]
Thus if $\lim_{t\rightarrow 0} t\cdot X$ is to be $0$ then
$x_{ij}=0$ for all $i>j$. In other words, $X$ is strictly
upper-triangular. Considering the general $1$-parameter group 
tells us that $X$ is in the null-cone iff there is an $A$ such
that $AXA^{-1} $ is strictly upper-triangular. In other words, $X$
is {\bf nilpotent}. The $1$-parameter subgroup identifies this
transformation and thus the flag of nilpotency.
\end{ex}

\section{The destabilizing flag}
In this section we do a more refined analysis of elements of the
null-cone. The basic motivation is to identify a unique set of
$1$-parameter subgroups which drive a null-point to zero. The
simplest example is given by $X_1^2 \in Sym^2 (\C^3 )$. Let $\lambda
$, $\lambda '$  and $\lambda ''$ be as below:
\[ \lambda (t) = \left[ \begin{array}{ccc}
t & 0 & 0 \\
0 & t^{-1} & 0 \\
0 & 0 & 1 \end{array} \right] 
\: \: \:  \lambda' (t) = \left[ \begin{array}{ccc}
t & 0 & 0 \\
0 & 1 & 0 \\
0 & 0 & t^{-1} \end{array} \right] 
\: \: \:  \lambda'' (t) = \left[ \begin{array}{ccc}
t & 0 & 0 \\
0 & 0& -1 \\
0 & t^{-1} & 0 \end{array} \right] \]
We see that all the three $\lambda $, $\lambda '$  and $\lambda ''$ 
drive $X_1^2 $ to
zero. The question is whether these are related, and 
to classify such $1$-parameter subgroups. Alternately, 
one may view this to a
more refined classification of points in the null-cone, such as
the stratification of the nilpotent matrices by their Jordan
canonical form.

There are two aspects to this analysis. Firstly, to identify a metric
by which to choose the 'best' $1$-parameter subgroup driving a 
null-point to zero. Next, to show that there is a unique equivalence 
class of such 'best' subgroups.

Towards the first objective, let $\lambda: \C^* \rightarrow SL_n $ 
be a $1$-parameter subgroup. Without loss of generality, we may assume
that $Im(\lambda ) \subseteq D$. If $w$ is a null-point then we have:
\[ t\cdot w =t^{n_1} w_1 + \ldots + t^{n_k} w_k \]
where $n_i >0$ for all $i$. Clearly, a measure of how fast $\lambda $
drives $w$ to zero is $m(\lambda )=min \{ n_1 ,\ldots ,n_k \}$. 
Verify that this really does not depend on the choice of the maximal
torus at all, and thus is well-defined.

Next, we see that for a $\lambda $ as above, we consider $\lambda^2 
:\C^* \rightarrow SL_n $ such that $\lambda^2 (t)=\lambda (t^2 )$. 
It is easy to see that $m(\lambda^2 )=2 \cdot m(\lambda )$. Clearly, 
$\lambda $ and $\lambda^2 $ are intrinsically identical and 
we would like to have a measure invariant under such scaling. This 
comes about by associating a length to each $\lambda $. Let $\lambda $ 
be as above and let $Im(\lambda )\subseteq D$. Then, there are 
integers $a_1 ,\ldots ,a_n $ such that 
\[ \lambda (t) = \left[ \begin{array}{cccc}
t^{a_1 } & 0 & 0 & 0 \\ 
0 & t^{a_2 } & 0 & 0 \\
  &          & \vdots & \\
  0 & 0 & 0 & t^{a_n } \\ \end{array} \right] \]

We define $\| \lambda \|$ as 
\[ \| \lambda \| =\sqrt{a_1^2 + \ldots + a_n^2 } \]

We must show that this does not depend on the choice of the maximal 
torus $D$. Let ${\cal T} (SL_n )$ denote the collection of all maximal tori
of $SL_n $ as abstract subgroups. For every $A \in SL_n $, we may 
define the map $\phi_A : {\cal T} \rightarrow {\cal T} $ defined by 
$T \rightarrow
ATA^{-1}$. The stabilizer of a torus $T$ for this action of $SL_n $ is
clearly $N(T)$, the normalizer of $T$. Also recall that $N(T)/T=W$ is 
the (discrete) weyl group. Let $Im(\lambda )\subseteq D \cap D'$ for 
some two maximal tori $D$ and $D'$. Since there is an $A$ such that 
$AD'A^{-1}=D$, it is clear that $\| \lambda \| =\| A \lambda 
A^{-1} \|$. Thus, we are left to check if $\| \lambda ' \| =\| 
\lambda \| $ when (i) $Im (\lambda ), Im(\lambda ') \subseteq D$, and 
(ii) $\lambda ' =A \lambda A^{-1}$ for some $A\in SL_n $. This throws 
the question to invariance of $\| \lambda \|$ under $N(D)$, or in 
other words, symmetry under the weyl group. Since $W \cong S_n $, the 
symmetric group, and since $\sqrt{a_1^2 +\ldots + a_n^2 }$ is a 
symmetric function on $a_1 ,\ldots ,a_n $, we have that $\| \lambda 
\| $ is well defined.

We now define the {\em efficiency} of $\lambda $ on a null-point $w$ 
to be
\[ e(\lambda )= \frac{m(\lambda )}{\| \lambda \|} \]
We immediately see that $e(\lambda )=e(\lambda^2 )$.

\begin{lemma} \label{lemma:unique}
Let $W$ be a representation of $SL_n $ and let $w\in W$ be a 
null-point. Let ${\cal N} (w,D)$ be the collection of all 
$\lambda :\C^* \rightarrow D$ such that $\lim_{t\rightarrow 0}
t \cdot w =0_W $. If ${\cal N}(w,D)$ is non-empty then there is 
a unique $\lambda' \in {\cal N}(w,D)$ which maximizes the efficiency, 
i.e., $e(\lambda' )>e(\lambda )$ for all $\lambda \in {\cal N}(w,D)$ 
and $\lambda \neq (\lambda')^k $ for any $k\in \Z$. This $1$-parameter
subgroup will be denoted by $\lambda (w,D)$.
\end{lemma}

\noindent
{\bf Proof}: Suppose that ${\cal N}(w,D)$ is non-empty. Then in the
weight-space expansion of $w$ for the maximal torus $D$, we see that
$supp(w)$ staisfies the half-space property for some $\lambda 
\in \Y^n $. Note that the $\lambda \in {\cal N}(w,D)$ are parametrized
by lattice points $\lambda \in \Y^n $ such that $(\lambda , \chi ) 
>0$ for all $\chi \in supp (w)$. Let $Cone(w)$ be the conical
combination (over $\R $) of all $\chi \in supp(w)$ and $Cone(w)^{\circ}
$ its {\bf polar}. Thus, in other words, ${\cal N}(w,D)$ is 
precisely the collection of lattice points in the cone 
$Cone(w)^{\circ}$. Next, we see that $e(\lambda )$ is a convex 
function of $Cone(w)^{\circ}$ which is constant over rays $\R^+ \cdot 
\lambda $ for all $\lambda \in Cone(w)^{\circ}$. By a routine 
analysis, the maximum of such a function must be a unique ray 
with rational entries. This proves the lamma. $\Box $

This covers one important part in our task of identifying the 'best' 
$1$-parameter subgroup driving a null-point to zero. The next part is 
to relate $D$ to other maximal tori.

Let $\lambda :\C^* \rightarrow SL_n $ and let $P(\lambda )$ be defined
as follows:
\[ P(\lambda )=\{ A \in SL_n | \lim_{t\rightarrow 0} \lambda (t) 
A \lambda (t^{-1}) =I \in SL_n \} \]

Having fixed a maximal torus $D$ containing $IM(\lambda )$, we 
easily identify $P(\lambda )$ as a {\bf parabolic} subgroup, i.e., 
block upper-triangular. Indeed, let 
\[ \lambda (t) = \left[ \begin{array}{cccc}
t^{a_1 } & 0 & 0 & 0 \\ 
0 & t^{a_2 } & 0 & 0 \\
  &          & \vdots & \\
  0 & 0 & 0 & t^{a_n } \\ \end{array} \right] \]
with $a_1 \geq a_2 \geq \ldots \geq a_n $ (obviously with 
$a_1+ \ldots + a_n =0$). Then 
\[ P(\lambda )=\{ (x_{ij} | x_{ij}=0 \mbox{ for all $i,j$ such that 
$a_i <a_j$} \} \]
The {\bf unipotent radical} $U(\lambda )$ is a normal subgroup of 
$P(\lambda )$ defined as:
\[ U(\lambda )=(x_{ij}) \mbox{ where } = \left\{ \begin{array}{rl}
x_{ij}=0 & \mbox{ if } a_i <a_j  \\
x_{ij}=\delta_{ij} & \mbox{ if } a_i =a_j \\
\end{array} \right. \]

\begin{lemma}
Let $\lambda \in {\cal N}(w,D)$ and let $g\in P(\lambda )$, then (i) 
$g \lambda g^{-1} \subseteq P(\lambda )$ and $P(g \lambda g^{-1})=
P(\lambda )$, (ii) $g \lambda g^{-1} \in {\cal N}(w,gDg^{-1})$, and 
(iii) $e(\lambda )=e(g \lambda g^{-1})$. 
\end{lemma}

This actually follows from the construction of the explicit
$SL_n$-modules and is left to the reader. We now come to the unique 
object that we will define for each $w\in W$ in the null-cone. This is 
the parabolic subgroup $P(\lambda )$ for any 'best' $\lambda $. We 
have already seen above that if $\lambda'$ is a 
$P(\lambda )$-conjugate of a best $\lambda $ then $\lambda'$ is 'equally
best' and $P(\lambda )=P(\lambda ')$. 

We now relate two general equally best $\lambda $ and $\lambda '$. For 
this we need a preliminary definition and a lemma:

\begin{defn}
Let $V$ be a vector space over $\C$. A {\bf flag} ${\cal F}$ of $V$ 
is a sequence $(V_0 ,\ldots ,V_r )$ of nested subspaces $0=V_0 \subset
V_1 \subset \ldots \subset V_r =V$.
\end{defn}

\begin{lemma}
Let $dim_{\C} (V)=r$ and let ${\cal F}=(V_0 ,\ldots ,V_r )$ and 
${\cal F}'=(V'_0 ,\ldots, V'_r )$ be two (complete) flags for $V$. 
Then there is a basis $b_1 ,\ldots ,b_r $ of $V$ and a permutation 
$\sigma \in S_r $ such that $V_i =\overline{\{ b_1 ,\ldots ,b_i \}}$
and $V'_i =\overline{\{ b_{\sigma (1)},\ldots ,b_{\sigma (i)} \}}$
for all $i$.
\end{lemma}

This is proved by induction on $r$.

\begin{cor} \label{cor:inter}
Let $\lambda $ and $\lambda '$ be two $1$-parameter subgroups and 
$P(\lambda )$ and $P(\lambda ')$ be their corresponding parabolic
subgroups. Then there is a maximal torus $T$ of $SL_n $ such that 
$T \subseteq P(\lambda ) \cap P(\lambda ')$.
\end{cor}

\noindent
{\bf Proof}: It is clear that there is a correspondence between 
parabolic subgroups of $SL_n $ and flags. We refine the flags associated
to the parabolic subgroups $P(\lambda )$ and $P(\lambda ')$ to complete
flags and apply the above lemma. $\Box $

We are now prepared to prove Kempf'd theorem:

\begin{theorem} \label{thm:kempf}
Let $W$ be a representation of $SL_n $ and $w\in W$ a null-point. Then 
there is a $1$-parameter subgroup $\lambda \in \Gamma (SL_n )$ such
that (i) for all $\lambda ' \in \Gamma (SL_n )$, we have $e(\lambda )
\geq e(\lambda ')$, and (ii) for all $\lambda '$ such that $e(\lambda )
=e(\lambda ')$ we have $P(\lambda )=P(\lambda ')$ and that there is a 
$g \in P(\lambda )$ such that $\lambda '=g\lambda g^{-1}$.
\end{theorem}

\noindent
{\bf Proof}: Let ${\cal N}(w)$ be all elements of $\Gamma (SL_n )$ 
which drive $w$ to zero. Let $\Xi (W)$ be the (finite) collection of 
$D$-characters appearing in the representation $W$. For every 
$\lambda (w,T)$ such that $Im(\lambda )\subseteq D$, we may consider an
$A\in SL_n $ such that $A \lambda A^{-1}\in {\cal N} (A \cdot w, D)$ 
and $e(\lambda )=e(A\lambda A^{-1})$. Since the 'best' element of 
${\cal N}(A\cdot w,D)$ is determined by $supp(A\cdot w) \subseteq 
\Xi$, we see that there are only finitely many possibilities for 
$e(A\cdot w, A \lambda A^{-1})$ and therefore for $e(\lambda )$ for 
the 'best' $\lambda $ driving $w$ to zero.

Thus the length $k$ of any sequence $\lambda (w,T_1 ), \ldots ,\lambda
(w,T_k )$ such that $e(\lambda (w,T_1 ))< \ldots <  e(\lambda (w,T_k ))
$ must be bounded by the number $2^{\Xi }$. This proves (i).

Next, let $\lambda_1 =\lambda(w,T_1 )$ and $\lambda_2 =\lambda (w,T_2 )$
be two 'best' elements of ${\cal N}(w,T_1 )$ and ${\cal N}(w,T_2 )$
respectively. By corollary \ref{cor:inter}, we have a torus , say $D$, 
and $P(\lambda (w,T_i ))$-conjugates $\lambda_i $ such that (i) 
$e(\lambda_i )=e(\lambda (w,T_i ))$ and (ii) $Im(\lambda_i ) \subseteq 
D$. By lemma \ref{lemma:unique}, we have $\lambda_1 =\lambda_2 $
and thus $P(\lambda_1 )=P(\lambda_2 )$. On the other hand, 
$P(\lambda (w,T_i ))=P(\lambda_i )$ and this proves (ii). $\Box $

Thus \ref{thm:kempf} associates a unique parabolic subgroup $P(w)$
to every point in the null-cone. This subgroup is called the 
{\bf destabilizing flag} of $w$. Clearly, if $w$ is in the 
null-cone then so is $A \cdot w$, where $A\in SL_n $. Furthermore, it 
is clear that $P(A \cdot w)=AP(w)A^{-1}$. 

\begin{cor} \label{cor:stab}
Let $w\in W$ be in the null-cone and let $G_w \subseteq SL_n $
stabilize $w$. Then $G_w \subseteq P(w)$.
\end{cor}

\noindent
{\bf Proof}: Let $g \in G_w $. Since $g \cdot w=w$, we see that
$gP(w)g^{-1} =P(w)$, and that $g$ normalizes $P(w)$. Since the 
normalizer of any parabolic subgroup is itself, we see that $g \in
P(w)$. $\Box $
 
\chapter{Stability}

\noindent {\em Reference:} \cite{kempf,GCT1} 

Recall that $z \in W$ is stable iff its orbit $O(z)$ is closed in $W$. 
In the last chapter, we tackled the points in the null-cone, i.e.,
points in the set $[0_W ]_{\approx}$, or in other words, points which 
close onto the stable point $0_W $. A similar analysis may be done for 
arbitrary stable points. 

Following kempf, let $S\subseteq W$ be a closed $SL_n $-invariant 
subset. Let $z\in W$ be arbitrary. If the orbit-closure $\Delta (z)$ 
intersects $S$, then we associate a unique parabolic subgroup $P_{z,S} 
\subseteq SL_n $ as a witness to this fact. The construction of this 
parabolic subgroup is in several steps.

As the first step, we construct a representation $X$ of $SL_n $ and 
a closed $SL_n $-invariant embedding $\phi :W \rightarrow X$ such that
$\phi^{-1} (0_X )=S$, scheme-theoretically. This may be done as 
follows: since $S$ is a closed sub-variety of $W$, there is an ideal
$I_s =(f_1 ,\ldots ,f_k )$ of definition for $S$. We may further
assume that the vector space $\overline{\{ f_1 ,\ldots ,f_k \}}$ is 
itself an $SL_n $-module, say $X$. We assume that $X$ is
$k$-dimensional. 

We now construct the map $\phi :W \rightarrow X$ as follows:
\[ \phi (w) = (f_1 (x), \ldots ,f_k (x)) \]
Note that $\phi (S)=0_X$ and that $I_S =(f_1 ,\ldots ,f_k )$ ensure that
the requirements on our $\phi $ do hold.

Next, there is an adaptation of (Hilbert's) Theorem \ref{thm:hilbert}
which we do not prove:

\begin{theorem} \label{thm:hilbert2}
Let $W$ be an $SL_n $-module and let $y\in W$ be a stable point. Let 
$z \in [y]_{\approx}$ be an element which closes onto $y$. Then there 
is a $1$-parameter subgroup $\lambda :\C^* \rightarrow SL_n $ such that
\[ \lim_{t \rightarrow 0} \lambda (t) \cdot w \in O(y) \]
Thus the limit exists and lies in the closed orbit of $y$. 
\end{theorem}

Now suppose that $\Delta (z) \cap S$ is non-empty. Then there must 
be stable $y \in \Delta (z)$. We apply the theorem to $O(y)$ and 
obtain the $\lambda $ as above. This shows that there is indeed a 
$1$-parameter subgroup driving $z$ into $S$. Next, it is easy 
to see that 
\[ \lim_{t \rightarrow 0} [\lambda (t) \cdot \phi (z) ] =0_X \]
Thus $\phi (z)$ actually lies in the null-cone of $X$. We may 
now be tempted to apply the techniques of the previous chapter to 
come up with the 'best' $\lambda $ and its parabolic, now called 
$P(z,S)$. This is almost the technique to be adopted , except
that this 'best' $\lambda$ drives $\phi (z)$ into $0_X$ but 
$\lim_{t \rightarrow 0} [\lambda (t)\cdot z]$ (which is supposed 
to be in $S$) may not exist! This is because we are using the 
unproved (and untrue) converse of the assertion that $1$-parameter 
subgroups which drive $z$ into $S$ drive $\phi (z)$ into $0_X$.

This above argument is rectified by limiting the domain of allowed 
$1$-parameter subgroups to (i) $Cone(supp(\phi (z))^{\circ}$ as 
before, and (ii) those $\lambda $ such that $\lim_{t \rightarrow 
0} [\lambda (t) \cdot z]$ exists. This second condition is also a 
'convex' condition and then the 'best' $\lambda $ 
does exist. This completes the construction of $P(z,S)$. 

As before, if $G_z \subseteq SL_n $ stabilizes $z$ then it normalizes 
$P(z,S)$ thus must be contained in it:

\begin{prop} \label{prop:stable}
If $G_z $ stabilizes $z$ then $G_z \subseteq P(z,S)$.
\end{prop}

Let us now consider the {\bf permanent} and the {\bf determinant}. Let 
${\cal M}$ be the $n^2 $-dimensional space of all $n\times n$-matrices.
Since $det$ and $perm$ are homogeneous $n$-forms on ${\cal M}$, we 
consider the $SL({\cal M})$-module $W=Sym^n ({\cal M}^* )$. We recall 
now certain stabilizing groups of the $det$ and the $perm$. We will 
need the definition of a certain group $L'$. This is defined as the 
group generated by the permutation and diagonal matrices in $GL_n $. 
In other words, $L'$ is the normalizer of the complete standard 
torus $D^* \subseteq GL_n $. $L$ is defined as that subgroup of
$L'$ which is contained in $SL_n $. 

\begin{prop}
\begin{itemize}
\item[(A)] Consider the group $K=SL_n \times SL_n $. We define 
the action $\mu_K $ of  typical element $(A,B) \in K$ on $X \in 
{\cal M}$ as given by:
\[ X \rightarrow AXB^{-1} \]
Then (i) ${\cal M}$ is an irreducible representation of $K$ and 
$Im(K)\subseteq SL({\cal M})$, and (ii) $K$ stabilizes the determinant.

\item[(B)] Consider the group $H= L \times L$. We define 
the action $\mu_H $ of  typical element $(A,B) \in H$ on $X \in 
{\cal M}$ as given by:
\[ X \rightarrow AXB^{-1} \]
Then (i) ${\cal M}$ is an irreducible representation of $H$ and 
$Im(H)\subseteq SL({\cal M})$, and (ii) $H$ stabilizes the permanent.
\end{itemize}
\end{prop}

We are now ready to show:

\begin{theorem} \label{thm:detperm}
The points $det$ and $perm$ in the $SL({\cal M}$-module $W=Sym^n ({\cal
M}^*)$ are stable.
\end{theorem}

\noindent
{\bf Proof}: Lets look at $det$, the $perm$ being similar. 
If $det$ were not stable, then there 
would be a closed $SL({\cal M})$-invariant subset $S \subset W$ 
such that $det \not \in S$ but closes onto $S$: just take $S$ 
to be the unique closed orbit in $[det]_{\approx}$. Whence there is a 
parabolic $P(det,S)$ which, by Proposition \ref{prop:stable}, would 
contain $K$. This would mean that there is a $K$-invariant flag in 
${\cal M}$ corresponding to $P(det,S)$. This contradicts the
irreducibility of ${\cal M}$ as a $K$-module. $\Box $

\end{document}